\PassOptionsToPackage{unicode}{hyperref}
\PassOptionsToPackage{hyphens}{url}
\PassOptionsToPackage{dvipsnames,svgnames,x11names}{xcolor}

\documentclass[12pt]{article}

\usepackage{amsmath,amssymb,amsthm}
\usepackage{iftex}
\ifPDFTeX
  \usepackage[T1]{fontenc}
  \usepackage[utf8]{inputenc}
  \usepackage{textcomp} 
\else 
  \usepackage{unicode-math}
  \defaultfontfeatures{Scale=MatchLowercase}
  \defaultfontfeatures[\rmfamily]{Ligatures=TeX,Scale=1}
\fi
\usepackage{lmodern}

\makeatletter
\@ifundefined{KOMAClassName}{
  \setlength{\parindent}{6mm}          
  \setlength{\parskip}{0pt plus 1pt}   
}{
  \KOMAoptions{parskip=half}}
\makeatother

\usepackage{xcolor}

\usepackage{titletoc}
\titlecontents{section}[0pt]{\bfseries}{\thecontentslabel\hspace{0.5em}}{}{\titlerule*[0.5pc]{.}\thecontentspage}
\titlecontents{subsection}[1.5em]{}{\thecontentslabel\hspace{0.5em}}{}{\titlerule*[0.5pc]{.}\thecontentspage}
\usepackage{setspace}
\usepackage{longtable,booktabs,array}
\usepackage{calc} 
\usepackage{etoolbox}
\makeatletter
\patchcmd\longtable{\par}{\if@noskipsec\mbox{}\fi\par}{}{}
\makeatother

\usepackage{graphicx}
\makeatletter
\def\maxwidth{\ifdim\Gin@nat@width>\linewidth\linewidth\else\Gin@nat@width\fi}
\def\maxheight{\ifdim\Gin@nat@height>\textheight\textheight\else\Gin@nat@height\fi}
\makeatother
\setkeys{Gin}{width=\maxwidth,height=\maxheight,keepaspectratio}
\def\fps@figure{htbp}

\usepackage{float}
\usepackage{bookmark}
\IfFileExists{xurl.sty}{\usepackage{xurl}}{} 
\usepackage{mathtools}
\usepackage{multirow,multicol}
\usepackage[mathscr]{euscript}
\usepackage{bbm}
\usepackage{pifont} 
\usepackage{tabularx}
\usepackage{array} 
\usepackage{algorithm}
\usepackage{algpseudocode}
\usepackage{enumitem}
\usepackage{arydshln}
\usepackage{subfig}
\usepackage[english]{babel}
\usepackage[toc,page]{appendix}

\hypersetup{colorlinks=true, linkcolor=blue, citecolor=blue}
\usepackage{xr}
\usepackage{xr-hyper}
\makeatletter
\newcommand*{\addFileDependency}[1]{
\typeout{(#1)}
\@addtofilelist{#1}
\IfFileExists{#1}{}{\typeout{No file #1.}}
}
\makeatother

\usepackage{natbib}
\newtheorem{assumption}{Assumption}

\newtheorem{definition}{Definition}
\newtheorem{lemma}{Lemma}
\newtheorem{proposition}{Proposition}
\newtheorem{theorem}{Theorem}

\newtheorem{remark}{Remark}

\newcommand{\bm}{\boldsymbol}
\newcommand{\cm}[1]{\mbox{\boldmath$\mathscr{#1}$}}

\definecolor{turquoise}{rgb}{0.03, 0.7, 0.87}

\DeclareMathAlphabet\mathrsfso{U}{rsfso}{m}{n}

\renewcommand{\arraystretch}{1.2}
\numberwithin{equation}{section}
\allowdisplaybreaks

\newcommand{\anon}{1} 

\begin{document}

\def\spacingset#1{\renewcommand{\baselinestretch}{#1}\small\normalsize} 

\if1\anon
{
  \title{\bf Mixed-Frequency Time Series Forecasting via Depth-Separable Neural Networks}
  \author{Yize Wang$^{a}$, Qianqian Zhu$^{b}$\thanks{Corresponding author. Address: School of Statistics and Data Science, Institute of Big Data Research, 777 Guoding Rd., Shanghai, 200433, P. R. China. Email: \href{mailto:zhu.qianqian@mail.shufe.edu.cn}{zhu.qianqian@mail.shufe.edu.cn}}, and Guodong Li$^{a}$ \\
  $^{a}$The University of Hong Kong and $^{b}$Shanghai University of Finance and Economics
  }
  \maketitle
} \fi

\if0\anon
{
  \bigskip
  \bigskip
  \bigskip
  \begin{center}
    {\LARGE\bf Mixed-Frequency Time Series Forecasting via Depth-Separable Neural Networks}
  \end{center}
  \medskip
} \fi

\bigskip
\begin{abstract}
To better forecast mixed-frequency time series, it is the key to choose a suitable way for frequency alignment. However, the existing methods are all limited to linear transformations, and this may overlook the possible nonlinearity, leading to a worse prediction. We alternatively consider a deep neural network for each frequency alignment, and hence a depth-separable neural network. Moreover, a parameter-sharing mechanism is adopted across the alignment at each stage, making possible a deeper network for a large set of higher-frequency predictors. This paper establishes an approximation theory for the proposed depth-separable network, and a non-asymptotic prediction error bound is also derived. Simulation studies demonstrate the finite-sample performance of the proposed method, and an empirical application to forecasting U.S. quarterly macroeconomic variables using monthly and daily indicators, highlights its superior predictive accuracy over existing mixed-frequency methods. 
\end{abstract}

\noindent%
\vfill

\newpage
\spacingset{1.8} 

\section{Introduction}\label{sec-intro}

In the era of big data, the growing availability of complex high-dimensional datasets presents substantial opportunities alongside significant analytical challenges. These challenges are particularly pronounced in the context of mixed-frequency data, where variables are recorded at different temporal intervals, a common scenario in fields such as economics \citep{andreou-2013}, finance \citep{engle-2013}, and environmental science \citep{xu-2022}. A typical example comes from macroeconomics, where forecasting a low-frequency quarterly variable like Gross Domestic Product (GDP) growth inherently relies on higher-frequency information flows. These include monthly economic indicators like the consumer and producer price indices, as well as daily financial variables such as stock returns and interest rates. To jointly model such series, mixed-frequency data need to be aligned to a common time scale before conventional single-frequency models can be applied. Frequency alignment critically affects how information from higher-frequency variables is utilized, thereby influencing the estimation efficiency and prediction accuracy. 
Moreover, economic and financial time series often exhibit nonlinear relationships such as threshold and interaction effects \citep{audrino-2011}, which are poorly captured by traditional linear models. Consequently, it is critical to develop flexible and computationally feasible methods that can simultaneously perform adaptive frequency alignment and model complex nonlinear dependencies. 

Within mixed-frequency settings, the core methodological challenge is \textit{frequency alignment}. Two predominant frameworks in the literature address this task:
\begin{enumerate}
	\item[(i)] \textbf{Low-to-High}: Low-frequency series are treated as high-frequency series with systematically missing observations. Missing values are imputed using state-space models (SSMs), typically estimated via the expectation-maximization algorithm \citep{schumacher-2007}, the Kalman filter \citep{giannone-2008}, or Bayesian methods \citep{mariano-2002,schorfheide-2015,marcellino-2016}.   
	\item[(ii)] \textbf{High-to-Low}: High-frequency series are aggregated or transformed to match the lowest frequency in the dataset. Common approaches include simple temporal aggregation, series stacking \citep{ghysels-2016a}, or Mixed-Data Sampling (MIDAS) regression \citep{ghysels-2004,ghysels-2007,clements-2008,andreou-2010,andreou-2013}.
\end{enumerate}
State-space models provide a theoretically coherent mechanism for mixed-frequency analysis through a system of measurement equations, which link observed variables to latent states, and state equations, which govern the dynamics of these states. By embedding low-frequency variables within a high-frequency latent process, SSMs enable joint modeling across frequencies. This framework has been widely applied, for example, in nowcasting quarterly GDP using monthly and daily indicators \citep{banbura-2011}.
However, SSMs exhibit several limitations. They rely strongly on linear and Gaussian assumptions, and their performance is sensitive to model specification \citep{andreou-2013}. Misspecification of either the state or measurement equation can introduce substantial forecasting errors. Moreover, as the data dimension grows, SSMs require an increasing number of parameters, and the associated Kalman filter becomes computationally demanding \citep{bai-2013}. These constraints limit the scalability and practical applicability of SSMs in high‑dimensional mixed-frequency settings, motivating the development of more flexible and efficient alternatives.

The MIDAS framework \citep{ghysels-2004,ghysels-2007} offers a computationally efficient alternative to state-space models, typically requiring far fewer parameters and making it especially attractive for datasets with many high-frequency predictors \citep{andreou-2013}. It has consequently become the dominant approach for High-to-Low frequency alignment. One key advantage of MIDAS lies in aligning high-frequency series via parsimonious weighting schemes, such as the exponential Almon lag polynomial \citep{ghysels-2007} or the unrestricted lag polynomial \citep{foroni-2015}, which reduce the information loss of simple aggregation and avoid the parameter proliferation from direct series stacking. 
The literature has since extended MIDAS in diverse directions to enhance its flexibility: the autoregressive distributed lag MIDAS model \citep{clements-2008} for adding autoregressive dynamics, the factor-augmented versions \citep{frale-2010,andreou-2013} that incorporate latent factors, and nonparametric MIDAS \citep{breitung-2015} to allow nonparametric weighting schemes. Other notable extensions integrate MIDAS with Lasso to handle high-dimensional settings \citep{babii-2022}, with GARCH-type models for volatility modeling \citep{engle-2013, amendola-2018}, with quantile regression to capture asymmetric effects \citep{ghysels-2014, ghysels-2016b}, with classical nonlinear or nonparametric models \citep{galvao-2013,wei-2025} to capture nonlinearity, and with neural networks for weighting schemes \citep{lin-2024} or nonlinear forecasting \citep{xu-2018,xu-2021,wang-2025}. Nevertheless, existing nonparametric extensions are mostly confined to univariate or additive specifications, and neural-network variants have yet to offer reliable estimation theory.

As mixed-frequency data grow in sample size and dimension, conventional frequency alignment approaches including SSMs and MIDAS-type methods, will encounter three fundamental issues: 
(a) alignment functions are typically parametric and pre-specified, thereby ignoring potential nonparametric structures and increasing the risk of model misspecification; 
(b) the alignment is usually confined to linear combinations of the original predictors, which fails to capture more complex nonlinear dependencies; and   
(c) the alignment process inherently relies on dimension reduction, excluding dimension-expansion techniques that could help represent richer nonlinear patterns, like the reproducing kernel Hilbert space in support vector machines \citep[Chapter 5]{murty-2016}. 
To address the above issues, we introduce a novel neural-network-based framework that jointly achieves adaptive frequency alignment and captures complex nonlinear dependencies. It performs alignment and prediction sequentially, i.e., from high to medium frequency, then from medium to low frequency, and finally to prediction, and hence a depth-separable neural network from the perspective of machine learning \citep{chollet-2017,wang-2020}. The proposed framework is coupled with parameter sharing across alignment networks, ensuring computational scalability even with many high-frequency predictors.
Moreover, the framework allows either dimension reduction or expansion by controlling output dimensions. 
Crucially, the proposed network integrates neural networks end-to-end, applying them to both alignment and forecasting. This unified design moves beyond the parametric, linear, and dimension-reducing constraints inherent in conventional and commonly used mixed-frequency models, providing a powerful and scalable tool for mixed-frequency data.

The main contributions of this paper are threefold.  
First, we develop a \textit{Depth-Separable Neural Network (DSNN)}, a novel architecture that {performs sequential} frequency alignment through multiple deep Rectified Linear Unit (ReLU) networks and then conducts forecasting on the fully aligned series at the final stage. The DSNN overcomes the key limitations of conventional approaches by employing neural networks end-to-end for both alignment and prediction, and enabling either dimension reduction or expansion as needed.  
We establish an approximation theory for the DSNN over a general class of hierarchical composition models.  
Secondly, we propose a least squares estimation method for the DSNN and derive its non-asymptotic prediction error bound, thereby filling a notable gap in the theoretical foundation of neural-network-based nonparametric MIDAS estimation.
Finally, we apply the proposed methodology to forecast key quarterly macroeconomic variables including GDP, consumption, investment, government spending and price deflator, using a large set of quarterly, monthly, daily economic and financial indicators. Empirical results show that the DSNN delivers significant prediction gains over existing mixed-frequency methods, highlighting its practical utility.

The remainder of this paper is organized as follows. Section \ref{sec:method} introduces the model or architecture of Depth-Separable Neural Networks and presents the corresponding least squares estimation method, along with the optimization algorithm and hyper-parameter selection procedure. Theoretical analysis, including an approximation theory and a non‑asymptotic prediction error bound, is provided in Section \ref{sec:theory}. Sections \ref{sec:simulation} and \ref{sec:real data} contain simulation studies and an empirical application, respectively. Concluding remarks and discussions are given at Section \ref{sec:conclusion and discussion}. All technical proofs and additional supporting materials are collected in the Appendix.
Throughout this paper, we adopt the following notation.
Let $[m] = \{1, \ldots, m\}$. Denote by $\mathbb{R}$ the set of real numbers, $\mathbb{N}_0 = \{0,1,2,\ldots\}$ the set of nonnegative integers, and $\mathbb{N}_+ = \{1,2,3,\ldots\}$ the set of positive integers. 
Bold lowercase letters denote vectors, e.g., $\mathbf{x} = [x_1, \ldots, x_d]^\top$ for a $d$-dimensional vector. Its $\ell_q$ norm is $\|\mathbf{x}\|_q = \left(\sum_{i=1}^{d} |x_i|^q \right)^{1/q}$, and its $\ell_\infty$ norm is $\|\mathbf{x}\|_\infty = \max_{1 \leq i \leq d} |x_i|$. The elementwise square root is defined as $\sqrt{\mathbf{x}} = (\sqrt{x_1}, \ldots, \sqrt{x_d})^\top$. For $\mathbf{x}, \mathbf{y} \in \mathbb{R}^d$, elementwise division is $\mathbf{x}/\mathbf{y} = (x_1/y_1, \ldots, x_d/y_d)^\top$, and Hadamard product is $(\mathbf{x} \odot \mathbf{y})_i = x_i y_i$ for $i \in [d]$. 
Bold uppercase letters denote matrices, e.g., $\mathbf{A} = [A_{i,j}]_{i \in [n], j \in [m]}$. We define $\|\mathbf{A}\|_2 = \sup_{\mathbf{x} : \|\mathbf{x}\|_2 = 1} \|\mathbf{A}\mathbf{x}\|_2$,  and $\|\mathbf{A}\|_{\infty} = \max_{i \in [n]} \sum_{j=1}^m |A_{i,j}|$.  
For functions $f(n)$ and $g(n)$, we write $f(n) \lesssim g(n)$ to mean $f(n) \leq C \cdot g(n)$ for some constant $C$ independent of $n$. Similarly, $f(n) \gtrsim g(n)$ means there exists a constant $C > 0$ independent of $n$ such that $f(n) \geq C \cdot g(n)$. We write $f(n) \asymp g(n)$ if both $f(n) \lesssim g(n)$ and $f(n) \gtrsim g(n)$ hold.
The dataset in Section \ref{sec:real data} and computer programs for the analysis are available at \url{https://github.com/Dexter97531/MF_DSNN.git}.

\section{Methodology}\label{sec:method}
\subsection{Data and Model Settings}

This study aims to forecast a low-frequency response using predictors observed at the same or higher frequencies. For simplicity, we focus on three distinct frequencies: low, medium, and high, indexed respectively by $t$, $q$ and $s$. 
For $t \in [T], q \in [Q], s \in[S]$, we define $\mathbf{y}[t]=(y_1[t],\ldots,y_{d_0}[t])^{\prime}\in \mathbb{R}^{d_0}$ as the low-frequency response vector, $\mathbf{z}[t]=(z_1[t],\ldots,z_{d_1}[t])^{\prime}\in \mathbb{R}^{d_1}$ the low-frequency predictor vector, $\mathbf{m}[t, q]=(m_1[t, q],\ldots,m_{d_2}[t, q])^{\prime}\in \mathbb{R}^{d_2}$ the medium-frequency predictor vector, and $\mathbf{h}[t,q,s] = (h_1[t,q,s], \ldots, h_{d_3}[t,q,s])^{\prime}\in \mathbb{R}^{d_3}$ the high-frequency predictor vector. 
Here, high-frequency data are observed $S$ times per medium-frequency period, which in turn are observed $Q$ times per low-frequency period.
For the example with yearly, quarterly and monthly data, we have $S=3$ and $Q=4$, and $T$ is the number of years.

Consider the nonparametric regression model as follows: 
\begin{align}
	\mathbf{y}[{t}] = m_0({\mathbf{x}}[t-1], \cdots, {\mathbf{x}}[t-K]) + \bm\epsilon[t], \quad t \in [T],
	\label{eq:total_initial}
\end{align}
where $\mathbf{x}[t] = (\mathbf{y}^{\prime}[{t}], \mathbf{z}^{\prime}[{t}], \mathbf{m}^{\prime}[t, \cdot], \mathbf{h}^{\prime}[t,\cdot,\cdot])^{\prime}$, $K$ is the maximum lag order, and $m_0$ is an unknown smooth function. 
The errors $\bm\epsilon[t] \in \mathbb{R}^{d_0}$ are  independent and identically distributed (i.i.d.) random vectors with zero mean, and are independent with the predictors $\mathbf{x}[s]$ for any $s<t$. 

The number of covariates at model \eqref{eq:total_initial} is $K(d_0+d_1+Qd_2+QSd_3)$ and this, together with the nature of mixed frequencies, makes infeasible the direct estimation of $m_0$. A common strategy in the literature is to adopt sequential alignment and prediction: first, a function $g_1$ maps high-frequency series to a medium frequency; second, $g_2$ further aggregates to the low-frequency level; finally, $g_3$ performs forecasting using the fully aligned low-frequency series, often implemented as a VAR model or a neural network. 

This work focuses on the functions $g_1$ and $g_2$ for frequency alignment, which is the core challenge for mixed-frequency data analysis. In the literature, these functions typically take the same form, with two common choices being stacking \citep{ghysels-2016a} and MIDAS \citep{ghysels-2004}. 
In stacking, $g_1=g_2=\mathcal{S}$, where the stacking operation $\mathcal{S}$ transforms each high-frequency series $h_i[t,q,s]$ into $S$ medium-frequency sub-series, and each medium-frequency series (including $m_i[t,q]$'s and the sub-series from $h_i[t,q,s]$'s) into $Q$ low-frequency sub-series. 
In MIDAS, 
$g_1=B(L^{1/S}; \bm{\theta})$ and $g_2 = B(L^{1/Q}; \bm{\theta})$, where $\bm{\theta}$ is an unknown parameter vector and $B(L^{1/m}; \bm{\theta})= \sum_{k=1}^{m} w(k; \bm{\theta}) \, L^{k/m}$ is the MIDAS lag polynomial \citep{ghysels-2007}. The weights \(w(k; \bm{\theta})\) are often specified using the exponential Almon form, $w(k; \bm{\theta}) = \exp\left(\theta_1 k + \theta_2 k^2 + \cdots + \theta_p k^p \right)/\sum_{j=1}^{m} \exp\left(\theta_1 j + \theta_2 j^2 + \cdots + \theta_p j^p \right)$,  or the Beta polynomial form, $w(k; \bm{\theta}) = \left(k/m\right)^{\theta_1 - 1} \left(1 - k/m\right)^{\theta_2 - 1}/\sum_{j=1}^{m} \left(j/m\right)^{\theta_1 - 1} \left(1 - j/m\right)^{\theta_2 - 1}$. 
However, as mixed-frequency data increase in sample size and dimension, these approaches face three major issues:  
\begin{enumerate}
	\item[(I1)] The parametric and pre-specified nature of $g_1$ and $g_2$ typically overlooks potential nonparametric structures, increasing the risk of model misspecification; 
    \item[(I2)] The frequency alignment is usually limited to linear combinations of original mixed-frequency predictors, failing to capture nonlinear dependencies; 
    \item[(I3)] The process inherently relies on dimension reduction of the aligned series, excluding dimension-expansion techniques that could represent richer nonlinear patterns. 
\end{enumerate}

To overcome these issues, we propose to implement $g_1$ and $g_2$ as deep ReLU networks for sequential frequency alignment (high $\to$ medium $\to$ low), followed by a prediction network $g_3$. This sequential design is analogous to the depth-separable architecture in the machine learning literature \citep{chollet-2017,wang-2020}. We therefore name our method the Depth-Separable Neural Network (DSNN). 
The DSNN address (I1)--(I3) directly: the deep networks provide flexible nonparametric approximations, can model nonlinear dependencies, and allow either dimension reduction or expansion through the choice of output dimensions. As for the MIDAS framework, we further incorporate parameter sharing across series within $g_1$ and $g_2$. This sharing mechanism, together with the increasing number of available mixed-frequency series, makes it possible to apply deeper neural networks to $g_1$ and $g_2$.

\subsection{Depth-Separable Neural Network (DSNN)}
\label{subsec:DSNN}

This section presents the DSNN for data with three distinct frequencies. The neural network is designed with three hierarchical stages corresponding to the three frequencies, where the first two stages perform hierarchical frequency alignment via stacking and then factor extraction, and the last stage conducts prediction using neural networks; see Figure \ref{fig:DSNN} for illustration. The methodology can be easily extended to settings with more frequencies, though with more elaborate notation.

The first stage of DSNN aligns high-frequency data to a medium frequency. It first applies the stacking operator $\mathcal{S}_1$ to each high-frequency series $h_i[t,q,s]$ for $i \in [d_3]$, producing $S$ medium-frequency sub‑series $\widetilde{\mathbf{m}}_i[t,q] = (\widetilde{m}_{(i-1)S+1}[t,q], \widetilde{m}_{(i-1)S+2}[t,q], \ldots, \widetilde{m}_{iS}[t,q]) \in \mathbb{R}^{S}$, where $\widetilde{m}_j[t,q]=h_i[t,q,s]$ with $j = (i-1)S + s$.
Subsequently, a shared neural network extracts $r_1$ latent factors from each sub-series $\widetilde{\mathbf{m}}_i[t,q]\in \mathbb{R}^{S}$:
\begin{equation}
	 (f_{i,1}^{(1)}[t,q], \ldots, f_{i,r_1}^{(1)}[t,q])^{\prime} = g_{1,\mathrm{NN}}(\widetilde{\mathbf{m}}_i[t,q]) \in \mathbb{R}^{r_1}, \quad i \in [d_3], t \in [T], q \in [Q],
	\label{eq:FM_1}
\end{equation}
where $g_{1,\mathrm{NN}}: \mathbb{R}^{S} \rightarrow \mathbb{R}^{r_1}$ is a deep ReLU network shared across all the $d_3$ high-frequency series. Specifically, $g_{1,\mathrm{NN}}$ is a fully connected deep neural network with the ReLU activation function $\sigma(x) = \max\{0, x\}$, depth $L \in \mathbb{N}_+$, and width parameter vector $\mathbf{w} = (w_0, w_1, \ldots, w_{L+1}) \in \mathbb{N}_{+}^{L+2}$ with $w_0=S$ and $w_{L+1}=r_1$, which is defined as follows 
\begin{equation}
	g_{1,\mathrm{NN}}(\mathbf{x}; L, \mathbf{w}) = \varphi_{L+1} \circ \sigma_L \circ \varphi_L \circ \ldots \circ \varphi_2 \circ \sigma_1 \circ \varphi_1(\mathbf{x}),
	\label{eq:NN}
\end{equation}
where $\varphi_l(\mathbf{x}) =  \mathbf{W}_l \mathbf{x} + \mathbf{b}_l$ is an affine transformation with the weight matrix $\mathbf{W}_l \in \mathbb{R}^{w_l \times w_{l-1}}$ and bias vector $\mathbf{b}_l \in \mathbb{R}^{w_l}$, and $\sigma_l: \mathbb{R}^{w_l} \to \mathbb{R}^{w_l}$ applies the ReLU activation function to each entry of a $w_l$-dimensional vector. We build our framework using a deep ReLU network due to its great empirical success. Given the large number of parameters in $g_{1,\mathrm{NN}}$, sharing this deep ReLU network across all inputs in the first stage is essential for computational feasibility. 
The factor dimension $r_1$ is a flexible hyper-parameter, where $r_1 \geq S$ enables dimension expansion to capture richer nonlinearities, whereas $r_1 < S$ achieves dimension reduction, which is particularly beneficial for large $S$, e.g., $S=20$ in day-to-month alignment. In practice, $r_1$ is chosen by cross validation; see Section \ref{subsec:hyperparameter} for details.

Let $\mathbf{F}_{i,r}[\cdot,\cdot]=(f_{i,r}^{(1)}[t,q])_{t \in [T], q \in [Q]} \in \mathbb{R}^{T \times Q}$ be the matrix of the $r$-th factor for the $i$-th high-frequency variable. All such matrices are aggregated into a factor tensor $ \cm{F} \in \mathbb{R}^{r_1 d_3 \times T \times Q}$, with its first mode consisting of $ \{\mathbf{F}_{i,r}[\cdot,\cdot], 1 \leq r \leq r_1, 1 \leq i \leq d_3\} $. 
This factor tensor is then concatenated with the original medium-frequency predictor tensor $ \cm{M} \in \mathbb{R}^{d_2 \times T \times Q} $ along the variable mode, forming the augmented medium-frequency predictor tensor $ \bar{\cm{M}} \coloneqq \{\bar{m}_i[t,q], i \in [D_2], t \in [T], q \in [Q]\} \in \mathbb{R}^{D_2 \times T \times Q} $, where $D_2 = d_2 + r_1 d_3$ denotes the total number of variables passed to the second stage.

The second stage of DSNN further aligns all the augmented medium-frequency data in $\bar{\cm{M}}$ to low-frequency one via stacking and factor extraction as in the first stage. Specifically, stacking operator $\mathcal{S}_2$ is applied to each augmented medium-frequency series $\bar{m}_i[t,q]$ for $i \in [D_2]$, generating $Q$ low-frequency sub-series $\widetilde{\mathbf{z}}_i[t] = (\widetilde{z}_{(i-1)Q+1}[t], \widetilde{z}_{(i-1)Q+2}[t], \ldots, \widetilde{z}_{iQ}[t]) \in \mathbb{R}^{Q} $ with $ \widetilde{z}_j[t] = \bar{m}_i[t,q]$ for $ j = (i-1)Q + q $.
Then $r_2$ latent factors are extracted from each series $\widetilde{\mathbf{z}}_i[t]$ via a shared deep ReLU network $g_{2,\mathrm{NN}}$ as follows 
\begin{equation} 
(f_{i,1}^{(2)}[t], \ldots, f_{i,r_2}^{(2)}[t])
= g_{2,\mathrm{NN}}(\widetilde{\mathbf{z}}_i[t]) \in \mathbb{R}^{r_2}, \quad i \in [D_2], t \in [T],
\label{eq:FM_2}
\end{equation} 
where $g_{2,\mathrm{NN}}: \mathbb{R}^{Q} \rightarrow \mathbb{R}^{r_2}$ is defined similarly as for $g_{1,\mathrm{NN}}$ in \eqref{eq:NN}, and is shared across all the $D_2$ medium-frequency series to ensure computational feasibility. Here, the factor dimension $r_2$ plays an analogous role to $r_1$, and can be selected in a similar manner; see Section \ref{subsec:hyperparameter} for further discussion. 
Denote the aggregated factor matrix by $ \mathbf{F} = (\mathbf{f}_{1,1}, \ldots, \mathbf{f}_{1,r_2}, \ldots, \mathbf{f}_{D_2,1}, \ldots, \mathbf{f}_{D_2,r_2})^{\prime} \in \mathbb{R}^{r_2 D_2 \times T}$, where $\mathbf{f}_{i,r} \in \mathbb{R}^{T}$ is the vector of the $r$-th factor for the $i$-th medium-frequency variable. Then concatenate $\mathbf{F}$ with the low-frequency predictors $\mathbf{Z} \coloneqq \{z_i[t], i \in [d_1], t \in [T]\} \in \mathbb{R}^{d_1 \times T}$ and the responses $\mathbf{Y} \coloneqq \{y_i[t], i \in [d_0], t \in [T]\} \in \mathbb{R}^{d_0 \times T}$, producing the final low-frequency data matrix $\bar{\mathbf{Z}} \coloneqq \{\bar{z}_i[t], i \in [D_3], t \in [T]\} \in \mathbb{R}^{D_3 \times T}$, where $D_3 \coloneqq d_0 + d_1 + r_2 d_2 + r_1 r_2 d_3$ denotes the total number of low-frequency variables that serve as the feature inputs for the final predictive stage.

The final stage of the DSNN serves to forecast the low-frequency response $\mathbf{y}[t]$ using the aligned low-frequency predictors $\bar{\mathbf{z}}[t-j]$'s for $j\in [K]$, where $\bar{\mathbf{z}}[t]=(\bar{z}_1[t],\ldots,\bar{z}_{D_3}[t])^{\prime}$. Given the unified frequency of the data, this stage can adopt any model suitable for single-frequency forecasting, such as VAR models. Here, we alternatively consider a neural network, 
\begin{equation}
\widehat{\mathbf{y}}[t] = g_{3,\mathrm{NN}}(\bar{\mathbf{z}}[t-1], \ldots, \bar{\mathbf{z}}[t-K]), \quad K+1 \leq t \leq T,
\label{eq:Reg_3}
\end{equation}
where $g_{3,\mathrm{NN}}:\mathbb{R}^{KD_3} \rightarrow \mathbb{R}^{d_0}$ is a deep ReLU network defined similarly as for $g_{1,\mathrm{NN}}$ in \eqref{eq:NN}.

Combining \eqref{eq:FM_1}--\eqref{eq:Reg_3}, we have $\widehat{\mathbf{y}}[{t}] = m(\mathbf{x}[t-1], \cdots, \mathbf{x}[t-K])$, where $m$ is a composed function in form of  
\begin{align} \label{eq:total}
m = g_{3,\mathrm{NN}} \circ &
\begin{pmatrix}
    g_{2,\mathrm{NN}} \\
    \vdots \\
    g_{2,\mathrm{NN}} 
\end{pmatrix} 
\circ \mathcal{S}_2 \circ 
\begin{pmatrix}
    g_{1,\mathrm{NN}} \\
    \vdots \\
    g_{1,\mathrm{NN}} 
\end{pmatrix} 
\circ \mathcal{S}_1
\coloneqq g_{3,\mathrm{NN}} \ast g_{2,\mathrm{NN}} \ast g_{1,\mathrm{NN}}, \\
&\quad \scriptscriptstyle D_2 \; \text{terms} \quad \quad \quad \quad \scriptscriptstyle D_1 \; \text{terms} \nonumber
\end{align}
with $D_1 = Q d_3$ and $D_2 = d_2 + r_1 d_3$. 
The proposed method models the unknown function $m_0$ using a structured function $m$ in \eqref{eq:total}, which aligns data sequentially from high to medium frequency, then to low frequency, and finally performs prediction. This depth-separable design, combined with parameter sharing across alignment networks, ensures computational scalability even with numerous high-frequency predictors.
Hence, we name it the Depth-Separable Neural Network (DSNN); see Figure \ref{fig:DSNN} for the illustration. 

\begin{figure}[t!]
	\includegraphics[width=1.\textwidth]{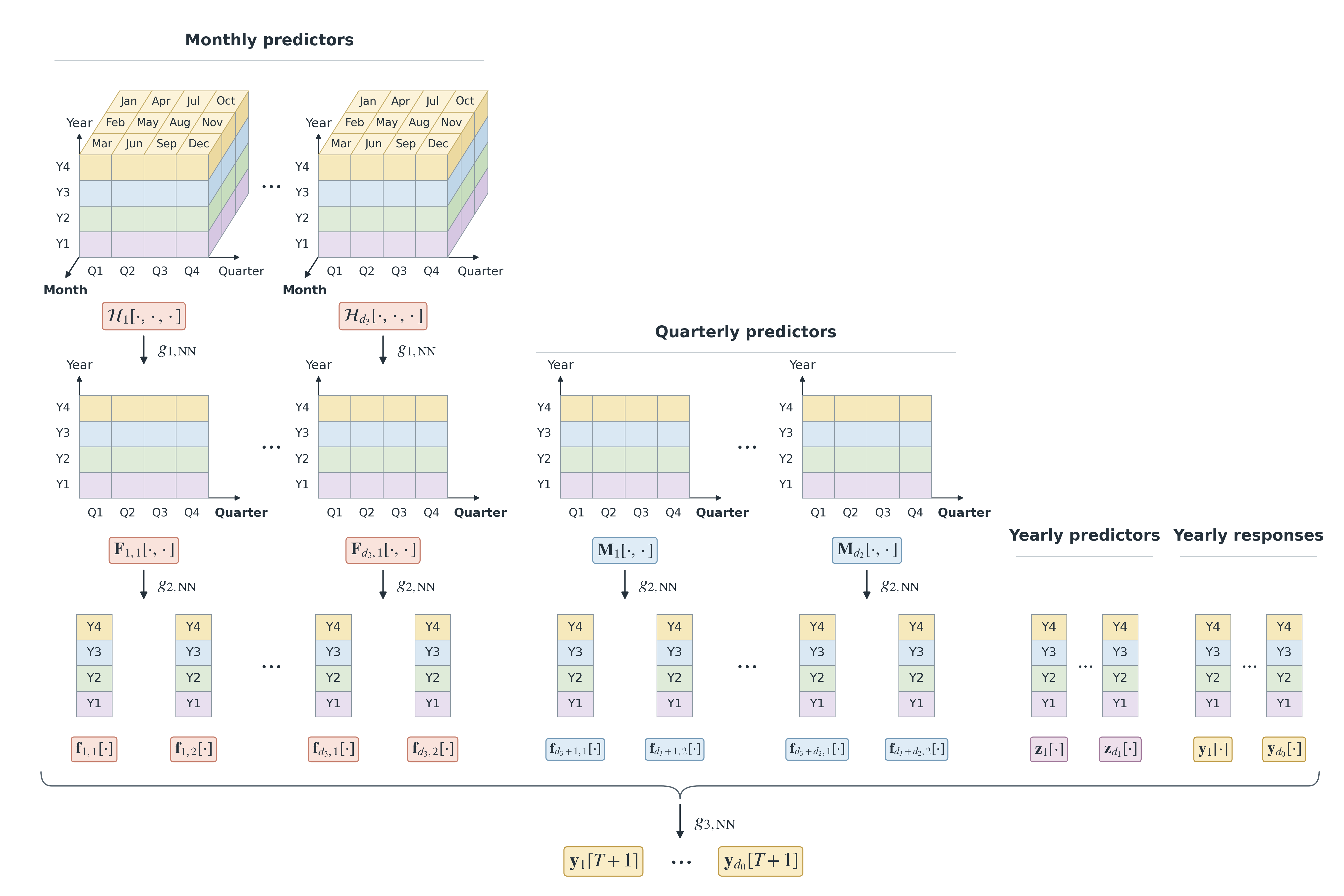}
	\caption{Illustration of the proposed Depth-Separable Neural Network (DSNN). The $\mathcal{H}_i[\cdot,\cdot,\cdot]=\{h_i[t,q,s], t \in [T], q \in [Q], s \in[S]\}$ denotes the high-frequency data tensor, and $\bm{M}[\cdot,\cdot]=\{m_i[t,q], t \in [T], q \in [Q]\}$ denotes the medium-frequency data matrix. The bold axis represents the time dimension to be transformed by the deep ReLU network.}
	\label{fig:DSNN}
\end{figure}

The proposed DSNN offers greater flexibility and generality than existing frameworks, owing to its highly nonlinear components $g_{j,\mathrm{NN}}$'s and the adaptable factor dimensions $r_i$'s. It addresses the key issues outlined in (I1)–(I3) while subsuming a wide range of existing methods as special cases, thereby establishing a unified framework for nonlinear mixed-frequency analysis. Notably, the DSNN reduces to several classical methods under specific constraints. For instance, 
(i) when $g_{1,\mathrm{NN}}$ and $g_{2,\mathrm{NN}}$ are restricted to linear mappings with fixed lag polynomials and $g_{3,\mathrm{NN}}$ to linear predictor, it recovers MIDAS-type regressions, including standard MIDAS \citep{ghysels-2007,andreou-2010}, MIDAS with an autoregressive term \citep{clements-2008}, and factor‑augmented distributed‑lag MIDAS \citep{andreou-2013}; 
(ii) if $g_{1,\mathrm{NN}}$ and $g_{2,\mathrm{NN}}$ reduce to the identity mapping and $g_{3,\mathrm{NN}}$ to a linear mapping, the model simplifies to stacked-series regression as in (2.1) of \citet{ghysels-2016a}.
Unlike the aforementioned models, which rely on linear or parsimonious nonlinear forms with few parameters, the DSNN employs trainable deep networks that can adaptively learn rich nonlinear representations from data. This flexibility allows it to capture complex temporal and cross-sectional dependencies that are often pre‑specified or omitted in traditional approaches. Consequently, the DSNN is expected to improve forecasting accuracy, particularly in settings where the data generating process exhibits strong nonlinearities.

To estimate the function $m$ in \eqref{eq:total}, we consider the least squares estimation (LSE) method. Specifically, define the empirical risk function as follows
\begin{align}\label{eq:empirical risk}
\mathcal{R}_T(m) 
= \frac{1}{T-K}
\sum_{t=K+1}^{T}
\|
\mathbf{y}[t] - m({\mathbf{x}}[t-1], \cdots, {\mathbf{x}}[t-K])
\|_2^2.
\end{align}
The corresponding estimator of $m$ is then defined to minimize $\mathcal{R}_T(m)$, and it can be calculated by the {gradient descent}; see Algorithm \ref{algorithmDSNN} in Section \ref{subsec:hyperparameter} for details.

\subsection{Implementation Issues}\label{subsec:hyperparameter}


This subsection first introduces an algorithm for optimization at Section \ref{subsec:DSNN}. 
The estimation of unknown function $m$ reduces to a least squares optimization over the parameters of the deep ReLU networks $g_{j,\mathrm{NN}}$ for $j\in [3]$, which is equivalent to estimating all weight matrices $\mathbf{W}_{j,l}$'s and bias vectors $\mathbf{b}_{j,l}$'s in these networks. Denote by $\bm{\theta}_j \coloneqq \big( \operatorname{vec}(\mathbf{W}_{j,1})^{\prime},\, \mathbf{b}_{j,1}^{\prime},\, \ldots,\, \operatorname{vec}({\mathbf{W}_{j,L_j+1}})^{\prime},\, {\mathbf{b}_{j,L_j+1}}^{\prime} \big)^{\prime}$ the parameter vector of $g_{j,\mathrm{NN}}$. The empirical risk in \eqref{eq:empirical risk} is minimized with respect to $\bm{\theta}_j$'s by using stochastic gradient descent with the Adam optimizer \citep{kingma-2014}, and it is implemented via the \textit{torch.optim} package. Adam is chosen for its adaptive learning-rate mechanism, which uses running estimates of the gradient's first moment ($\bm{\mu}_j$) and second moment ($\mathbf{v}_j$) for each network $g_{j,\mathrm{NN}}$. This approach enhances stability and convergence in high‑dimensional, ill-conditioned optimization settings. To ensure robustness and avoid overfitting, standard practices such as mini‑batching and early stopping are employed throughout training. The optimization procedure for LSE of DSNN is summarized in Algorithm \ref{algorithmDSNN}.

\begin{algorithm}
	\caption{Gradient descent algorithm for the least squares optimization problem}\label{algorithmDSNN}
	\begin{algorithmic}[1]
		\State \textbf{Input:} Factor dimensions $r_1, r_2$, lag order $K$, depth $L_j$ and width $\mathbf{w}_j$ for $g_{j,\mathrm{NN}}$ with $j\in [3]$, learning rate $\alpha$, decay rates $\beta_1, \beta_2 \in [0, 1)$, stability constant $\epsilon > 0$, max iterations $I$, and tolerance $\delta > 0$.
		\State Initialize $\bm{\theta}^{(0)}_j$ via Kaiming Uniform, $\bm{\mu}^{(0)}_j,\mathbf{v}^{(0)}_j \gets \bm{0}$ for $j\in [3]$.
		
		\For{$t=1$ to $I$}
		\State Forward: $m^{(t)} \gets g_{3,\mathrm{NN}}(g_{2,\mathrm{NN}}(g_{1,\mathrm{NN}}(\bm{x}; L_1, \mathbf{w}_1, \bm{\theta}^{(t-1)}_1); L_2, \mathbf{w}_2, \bm{\theta}^{(t-1)}_2); L_3, \mathbf{w}_3, \bm{\theta}^{(t-1)}_3)$.
		\State Gradients: $\nabla \mathbf{g}_{j}^{(t)} \gets \nabla_{\bm{\theta}_j}\mathcal{R}_T(m^{(t)})$ for $j=3,2,1$.
		\For{$j=3$ down to $1$}
		\State $\bm{\mu}^{(t)}_j \gets \beta_1 \bm{\mu}^{(t-1)}_j+(1-\beta_1)\nabla \mathbf{g}_{j}^{(t)}$,\quad $\mathbf{v}^{(t)}_j \gets \beta_2 \mathbf{v}^{(t-1)}_j+(1-\beta_2)\nabla \mathbf{g}_{j}^{(t)}\odot \nabla \mathbf{g}_{j}^{(t)}$;
		\State $\widetilde{\bm{\mu}}^{(t)}_j \gets \bm{\mu}^{(t)}_j/({1-\beta_1})$,\quad $\widetilde{\mathbf{v}}^{(t)}_j \gets \mathbf{v}^{(t)}_j/({1-\beta_2})$;
		\State $\bm{\theta}^{(t)}_j \gets \bm{\theta}^{(t-1)}_j - \alpha \widetilde{\bm{\mu}}^{(t)}_j / \left(\sqrt{\widetilde{\mathbf{v}}^{(t)}_j} + \epsilon \mathbf{1}\right)$. 
		\EndFor
		\If{$\sum_{j=1}^3\|\nabla \mathbf{g}_{j}^{(t)}\|<\delta$} \textbf{break} \EndIf
		\EndFor
		
		\State \Return $\bm{\theta}^{(t)}_1, \bm{\theta}^{(t)}_2$, and $\bm{\theta}^{(t)}_3$.
	\end{algorithmic}
\end{algorithm}

We then consider the selection of hyper-parameters, including factor dimensions $r_1, r_2$ and lag order $K$. They are selected jointly through cross-validation to optimize out-of-sample forecasting performance. Specifically, the data are divided into a training set of length $T_{\mathrm{train}}$ and a validation set of length $T_{\mathrm{val}}$. Using a rolling forecast procedure with an expanding window, the DSNN is re-estimated using data up to time $T_{\mathrm{train}} + t - 1$, and a one-step-ahead forecast is generated at each step $t \in [T_{\mathrm{val}}]$. The optimal combination $(r_1, r_2, K)$ is chosen by minimizing the Root Mean Square Forecast Error (RMSFE) on the validation set, which is searched over $1 \leq r_1 \leq r_{1,\max}$, $1 \leq r_2 \leq r_{2,\max}$, and $1 \leq K \leq K_{\max}$. Here, $r_{1,\max}$ and $r_{2,\max}$ are pre-specified upper bounds for the factor dimensions, and $K_{\max}$ is a predetermined maximum lag.


\section{Theoretical Properties}\label{sec:theory}

This section establishes theoretical results for the proposed DSNN, including an approximation theory of hierarchical composition models by the DSNN in Proposition \ref{prop:approximation_our}, and a generalization error bound for the DSNN estimated by LSE in Proposition \ref{prop:generalization_our}. These two together lead to a non-asymptotic prediction error bound in Theorem \ref{thm:1}.

The DSNN is defined by the function $m$ in \eqref{eq:total}, a composition of deep ReLU networks with stacking. We first define a function class for $m$. 

\begin{definition}[The DSNN class]\label{def:DNN}
	For any $L \in \mathbb{N}_+$ and $\mathbf{w} = (w_0, w_1, \ldots, w_{L+1}) \in \mathbb{N}_{+}^{L+2}$, the class of deep ReLU networks with depth $L$ and width parameter $\mathbf{w}$ can be defined as 
	\begin{equation*} 
		\mathcal{G}(L, \mathbf{w}) = \left\{ g : \mathbb{R}^{w_0} \to \mathbb{R}^{w_{L+1}} \;\Big|\; g(\mathbf{x}; L, \mathbf{w}) = \varphi_{L+1} \circ \sigma_L \circ \varphi_L \circ \ldots \circ \varphi_2 \circ \sigma_1 \circ \varphi_1(\mathbf{x}) \right\}, 
	\end{equation*} 
	where $\varphi_l$'s and $\sigma_l$'s are defined as in \eqref{eq:NN}. 
	We denote it as $\mathcal{G}(L, d_{\mathrm{in}}, d_{\mathrm{out}}, w)$ if $\mathbf{w} = (d_{\mathrm{in}}, w, \ldots, w, d_{\mathrm{out}})$, which is referred as deep ReLU network with depth $L$ and width $w$. 
	We then further define the function class for DSNNs below, 
	\begin{align*}
		\mathcal{F}_{\mathrm{DSNN}} = \left\{m= g_{3,\mathrm{NN}} \ast g_{2,\mathrm{NN}} \ast g_{1,\mathrm{NN}} \;\Big|\; g_{j,\mathrm{NN}}\in \mathcal{G}(L_j,{d}_{\mathrm{in},j}, {d}_{\mathrm{out},j}, w_j) \;\text{for}\; {j\in[3]} \right\},
	\end{align*}
	where $L_j \in \mathbb{N}_+$ are the depths, $w_j \in \mathbb{N}_+$ are the widths, $(d_{\mathrm{in},1}, d_{\mathrm{in},2}, d_{\mathrm{in},3})=(S, Q, KD_3)$ are the input dimensions, and $(d_{\mathrm{out},1}, d_{\mathrm{out},2}, d_{\mathrm{out},3}) =(r_1, r_2, d_0)$ are the output dimensions.
\end{definition}

We next define a class of smooth functions for $m_0$ at model \eqref{eq:total_initial}.    
Given the hierarchical structure of $m$, it is natural to consider a class of hierarchical composition models \citep{kohler-2021}, and it is defined via a concept of $(p, C)$-smooth functions. 

\begin{definition}[$(p, C)$-smooth function]\label{p-C-smooth}
	Let $C > 0$, and $p = r + s$ for some $r \in \mathbb{N}_0$ and $0 < s \leq 1$. A function $f: \mathbb{R}^d\to \mathbb{R}$ is called $(p, C)$-smooth if for every  $\bm{\alpha}= (\alpha_1,\ldots,\alpha_d) \in \mathbb{N}_0^d$ with $\sum_{j=1}^d \alpha_j = r$, the partial derivative  
	$\partial^r f/(\partial x_1^{\alpha_1} \cdots \partial x_d^{\alpha_d})$
	exists and satisfies for all $\mathbf{x},\mathbf{z}\in \mathbb{R}^d$, 
	\begin{align*}
		\left| \frac{\partial^r f}{\partial x_1^{\alpha_1} \cdots \partial x_d^{\alpha_d}}(\mathbf{x}) - \frac{\partial^r f}{\partial x_1^{\alpha_1} \cdots \partial x_d^{\alpha_d}}(\mathbf{z}) \right| \leq C \|\mathbf{x} - \mathbf{z}\|_2^s. 
	\end{align*}
\end{definition}

\begin{definition}[Hierarchical composition model]\label{def:HCM}
	The hierarchical composition model $ \mathcal{H}(d, l, \mathcal{P}) $, with $ l, d \in \mathbb{N}_+ $ and $ {
    \varnothing\neq} \mathcal{P} \subseteq [1, \infty) \times \mathbb{N}_+ $ satisfying 
	$\sup_{(p, \tilde{d}) \in \mathcal{P}} (p \vee \tilde{d}) < \infty$,  
	is defined recursively: 
	\begin{align*}
		\text{For } l = 1, \quad 
		\mathcal{H}(d, 1, \mathcal{P}) = &
		\{ h : \mathbb{R}^{d} \to \mathbb{R}, 
		h(\boldsymbol{x}) = \phi(x_{\pi(1)}, \ldots, {x_{\pi(\tilde{d})}}), 
		\text{where } \phi : \mathbb{R}^{\tilde{d}} \to \mathbb{R} \\ 
		&\text{ is } (p, C)\text{-smooth for some } (p, \tilde{d}) \in \mathcal{P}\ \text{and } \pi : [\tilde{d}] \to [d] \}.
	\end{align*}
	\begin{align*}
		\text{For } l > 1, \quad 
		\mathcal{H}(d, l, \mathcal{P}) = &
		\{ h : \mathbb{R}^{d} \to \mathbb{R}, 
		h(\boldsymbol{x}) = \phi({f_1(\boldsymbol{x})}, \ldots, {f_{\tilde{d}}(\boldsymbol{x})}),\ 
		\text{where } \phi : \mathbb{R}^{\tilde{d}} \to \mathbb{R} \\
		&\text{ is } (p, C)\text{-smooth for some } (p, \tilde{d}) \in \mathcal{P}\ \text{and } f_i \in \mathcal{H}(d, l-1, \mathcal{P}) \}.
	\end{align*}
\end{definition}

For each $g_j: \mathbb{R}^{d_{\mathrm{in},j}} \to \mathbb{R}^{d_{\mathrm{out},j}}$ with ${j\in[3]}$, denote its $i$-th component by ${g_{j,i}:\mathbb{R}^{d_{\mathrm{in},j}}\to\mathbb{R}}$ with $i\in [d_{\mathrm{out},j}]$.
We then can define the class of smooth functions for $m_0$: 
\begin{align*}
	\mathcal{F}_* = \left\{ 
		m_0= g_3 \ast g_2 \ast g_1 \;\Big|\; 
		g_{j,i}\in \mathcal{H}({d}_{\mathrm{in},j}, l_j, \mathcal{P}) 
		\;\;\text{for}\;\; {j\in[3]} \;\;\text{and}\;\; i\in [{d}_{\mathrm{out},j}]
	\right\}.
\end{align*}
This class provides a general and flexible framework for characterizing the smoothness of $m_0$, accommodating different smoothness and order constraints within the same compositional level \citep{kohler-2021}.
This smoothness condition is standard in nonparametric analysis and ensures that $m_0$ can be effectively approximated by deep ReLU networks. 
Building on these specifications for $m$ and $m_0$, we establish an approximation theory for the DSNN over a broad class of hierarchical composition models, under the boundedness conditions summarized below.

\begin{assumption}[Boundedness]\label{assump:Boundness}
(i) The support of $\mathbf{x}[t]$ is bounded; 
(ii) The true function $m_0$ is bounded, i.e. $\|m_0\|_{\infty} \leq B$;
(iii) The $\ell_2$-norms of weight matrices and bias vectors in the deep ReLU networks $g_{j,\mathrm{NN}}$ are bounded, i.e., $\max\{\|\mathbf{W}_{j,l}\|_2,\|\mathbf{b}_{j,l}\|_2\} < U_{j,l} < \infty$ and $l \in [L_j+1]$ for $j\in [3]$. Moreover, there exist constants $M_j>0$ such that $\prod_{l=1}^{L_j+1}  U_{j,l} \leq M_j$.  
\end{assumption}

Assumption \ref{assump:Boundness}(i) imposes compact support on the predictors and responses, which is standard in the theory of neural networks and nonparametric estimation. The bounded restriction on the underlying function $m_0$ in Assumption \ref{assump:Boundness}(ii) is a common condition in the literature on nonparametric regression and deep learning \citep{farrell-2021,kengne-2023,caner-2025}. 
Bounding the $\ell_2$-norms of weight matrices and bias vectors in Assumption \ref{assump:Boundness}(iii) is a mild requirement that controls the complexity of the network. {{Similar {$\ell_\infty$}-norm bounds are often used interchangeably {\citep{kengne-2023,xiu-2024}}. A more restrictive condition with {$\ell_\infty$}-norm bounded by one, which resembles nearly orthogonal weight matrices, has been considered under orthogonal initialization {\citep{schmidt-hieber-2020}}.}}
In our DSNN, because the ReLU activation is 1-Lipschitz, the product bound $M_j$ guarantees stable error propagation across the multi-depth architecture; see also a similar stability condition in \citet{feng-2023} {and {\citet{jiao-2023}}}. Note that Assumptions \ref{assump:Boundness}(i) and (iii) ensure the boundedness of the function $m$, a standard requirement in neural network theory \citep{schmidt-hieber-2020, fan-2023, fan-2024, xiu-2024}.
Under the aforementioned smoothness and boundedness conditions, we establish the approximation error in Proposition \ref{prop:approximation_our}.

\begin{proposition}[Approximation error bound] \label{prop:approximation_our}
    Suppose that Assumption \ref{assump:Boundness} holds.  
    Then there exist universal constants $c_1, c_2, c_3 > 0$ (depending only on $l{=\max_{j\in[3]}l_j}$, $\sup_{(p, \tilde{d}) \in \mathcal{P}} \max\{p, {\tilde{d}}\}$, and the constant $C$ in Definition \ref{def:HCM}), such that for any $\bar{L}_j, \bar{w}_j \geq 3$ with ${j\in[3]}$,
    \begin{align*}
    &\sup_{m_0\in\mathcal{F}_{\ast}} \inf_{m\in\mathcal{F}_{\mathrm{DSNN}}} \|m - m_0\|_{\infty} 
    \leq & c_1 \Big( \sqrt{r_1 r_2} M_2 M_3 (\bar{L}_1 \bar{w}_1)^{-2 \gamma^*} 
        + \sqrt{r_2} M_3 (\bar{L}_2 \bar{w}_2)^{-2 \gamma^*} 
        + (\bar{L}_3 \bar{w}_3)^{-2 \gamma^*} \Big),
    \end{align*}
    where $\gamma^*=\min_{(p,\tilde d)\in\mathcal P}p/\tilde d$, and the corresponding depths and widths of the deep ReLU networks $g_{j,\mathrm{NN}}$ are given by $L_j = c_2 \big[\bar{L}_j \log \bar{L}_j\big]$ and $w_j = c_3d_{\mathrm{out},j}\big[\bar{w}_j \log \bar{w}_j\big]$. 
\end{proposition}

Proposition \ref{prop:approximation_our} establishes that the DSNN can approximate the function $m_0$ in the hierarchical composition model class $\mathcal{H}$. This result builds on the approximation rate of $(Lw)^{-2p/{\tilde{d}}}$ in Proposition 3.4 of \citet{fan-2024} for a deep ReLU network, by propagating the error through the multi-stage architecture of the DSNN. Consequently, the resulting error bound depends on the depths $L_j$ and widths $w_j$ of the three deep ReLU networks, as well as the factor dimensions ($r_1,r_2$) and the product bounds on the network weight matrices.

We next consider the generalization error for DSNN. Define the population risk function as $\mathcal{R}(m)=\mathbb{E}\left\|\mathbf{y}[t] - m({\mathbf{x}}[t-1], \ldots, {\mathbf{x}}[t-K])\right\|_2^2$, corresponding to the empirical risk function $\mathcal{R}_T(m)$ in \eqref{eq:empirical risk}. 
To bound the generalization error $\mathcal{R}_T(m)-\mathcal{R}(m)$ for any $m\in \mathcal{F}_{\mathrm{DSNN}}$, we need to characterize the dependence structure of the time series process $\{\mathbf{x}[t]\}_{t=1}^{T}$. 
We quantify this via the $\theta_{\infty,k}$-weak dependence coefficients introduced at Definition 2.3 of \citet{dedecker-2007}, which provides a mixing measure weaker than the classical $\gamma$-mixing and is well-suited for bounded processes.

\begin{definition}[$\theta_{\infty,k}$-weak dependence]\label{def:weak_dependence}  
	Let ${\{\mathbf{x}[t]\}_{t\in\mathbb Z}}$ be a bounded, stationary, and ergodic process. 
	For any $1$-Lipschitz function $f$ and any $k \in \mathbb{N}_+$, the $\theta_{\infty,k}$-weak dependence coefficients are defined as
	\begin{align*}
		\theta_{\infty,k} := \sup_{f{:\,\operatorname{Lip}(f)\leq1},\, 0<j_1<\ldots<j_k} 
		\Big\| 
		\mathbb{E}\big[f(\mathbf{x}_{j_1}, \ldots, \mathbf{x}_{j_k}) \mid \mathbf{x}_{t}, t \leq 0\big] 
		- \mathbb{E}[f(\mathbf{x}_{j_1}, \ldots, \mathbf{x}_{j_k})] 
		\Big\|_{\infty},
	\end{align*}
	where $\|\cdot\|_{\infty}$ denotes the essential supremum of a random variable and $\operatorname{Lip}(f)$ denotes the Lipschitz constant of the function $f$. 

\end{definition}

\begin{assumption}[Stationarity and weak dependence]\label{assump:Weakly dependence}
	The process $\{\mathbf{x}[t]\}_{t=1}^{T}$ is stationary and ergodic, and its $\theta_{\infty,k}$-weak dependence coefficients satisfy $\theta_{\infty,k} < \infty$ for all $k > 0$. 
\end{assumption}

Assumption \ref{assump:Weakly dependence} enables the analysis of generalization error under $\theta_{\infty,k}$-weak dependence; see also \citet{kengne-2023}. This condition is notably weaker than the $\beta$-mixing assumption adopted in some related works such as \cite{feng-2023}.

\begin{proposition}[Generalization error bound]\label{prop:generalization_our}
	Suppose that Assumptions \ref{assump:Boundness}--\ref{assump:Weakly dependence} hold.  
	For all ${T} \in \mathbb{N}_+$ and $\nu_0 \in (0,1]$, there exists a constant $c_4> 0$ {independent of $\nu_0$} such that
	\begin{align*}
		\mathbb{E}\sup_{m\in \mathcal{F}_{\mathrm{DSNN}}} 
		\left| \mathcal{R}(m) - \mathcal{R}_T(m) \right| 
		\leq & c_4 {C_{\operatorname{NN}}} T^{-\nu_0} \log(T)\,
		+ C_\theta T^{\nu_0 -1},
	\end{align*}
	where     $C_{\operatorname{NN}}=L^\ast\big(D_1L_1w_1^2+D_2L_2w_2^2+L_3w_3^2\big)$ with $L^{\ast} = L_1 + L_2 + L_3$,
    and $ C_\theta=(\theta_{\infty,T}+U_x)^2$ with $\theta_{\infty, T}$ defined as in Definition \ref{def:weak_dependence} and $U_x \coloneqq \sup_t\|\mathbf{x}[t] \|_2$.
    
    \noindent Moreover, if $T$ satisfies $\log T/T\leq C_\theta/(c_4 {C_{\operatorname{NN}}})<T\log T$, then
    \begin{align*}
    \mathbb{E}\sup_{m\in\mathcal{F}_{\mathrm{DSNN}}}
    \left|\mathcal{R}(m)-\mathcal{R}_T(m)\right|&\leq 2\sqrt{c_4{C_{\operatorname{NN}}}C_\theta}
    \left(\frac{\log T}{T}\right)^{1/2}.
    \end{align*}
\end{proposition}
The generalization error bound in Proposition
{\ref{prop:generalization_our}} depends on the complexity of DSNNs and the dependence of data process.
{The term ${C_{\operatorname{NN}}}$ accounts for the complexity of deep ReLU networks, which grows with their depths and widths.}
This is a standard finding in nonparametric regression with deep learning, as a more complex architecture tends to induce a larger generalization error. The specific rate originates from covering-number arguments, {where the logarithm of the covering number for a deep ReLU network} of depth $L$ and width $w$ scales as $L w^2 \log(U^L/\epsilon)$. This scaling is essentially the same as that obtained in \citet{ou-2026}, though our derivation employs a different technical approach. 
The term $C_\theta$ quantifies the impact of temporal dependencies under the {$\theta_{\infty, T}$}-weak dependence framework of \citet{kengne-2023}.
 As expected, stronger dependence leads to a larger error. This contrasts with the analysis in \citet{feng-2023}, which assumes a stronger $\beta$-mixing and uses an independent-block scheme (with $2\mu_n$ blocks of length $a_n$), thereby incurring an extra error $\mu_n \beta_{a_n}$. Our result, established under the weaker $\theta_{\infty, T}$-mixing condition, therefore provides a less restrictive theoretical guarantee.   
In addition, {{when {$T$} is sufficiently large such that the stated condition holds}}, the generalization error decays at the rate of $(\log T/T)^{1/2}$ when network complexity and temporal dependence remain bounded.

We now analyze the prediction error of DSNNs estimated by the least squares method. The LSE of $m$ is defined by minimizing the empirical risk over the DSNN function class: 
\begin{align*} 
	\widehat{m}_T &= \arg\min_{m\in \mathcal{F}_{\mathrm{DSNN}}} \mathcal{R}_T(m), 
\end{align*} 
where $\mathcal{R}_T(m)$ is defined in \eqref{eq:empirical risk}. Under the squared error loss, the prediction error coincides with the excess risk $\mathcal{R}(\widehat{m}_T) - \mathcal{R}(m_0)$ up to a constant. Here, the underlying true function $m_0$ is the minimizer of population risk over the smooth function class $\mathcal{F}_{\ast}$, i.e. $m_0 = \arg\min_{f\in\mathcal{F}_{\ast}} \mathcal{R}(f)$. For any $m \in \mathcal{F}_{\mathrm{DSNN}}$, the excess risk can be decomposed as
\begin{align*}
	\mathcal{R}(\widehat{m}_T) - \mathcal{R}(m_0) &= \mathcal{R}(\widehat{m}_T) - \mathcal{R}_T(\widehat{m}_T) + \mathcal{R}_T(\widehat{m}_T) - \mathcal{R}_T(m) 
	+ \mathcal{R}_T(m) - \mathcal{R}(m) + \mathcal{R}(m) - \mathcal{R}(m_0) \nonumber\\
	& \overset{(i)}{\leq} \left\{\mathcal{R}(\widehat{m}_T) - {\mathcal{R}}_T(\widehat{m}_T)\right\} + \left\{{\mathcal{R}}_T(m) - \mathcal{R}(m)\right\} + \left\{\mathcal{R}(m) - \mathcal{R}(m_0)\right\} \nonumber\\
	& \overset{(ii)}{\leq} 2 \sup_{m \in \mathcal{F}_{\mathrm{DSNN}}} \left| {\mathcal{R}}(m) - \mathcal{R}_T(m) \right| + d_0 \|m - m_0\|_{\infty}^2. 
\end{align*} 
Inequality (i) uses the definition of the LSE, which implies ${\mathcal{R}}_T(\widehat{m}_T) \leq {\mathcal{R}}_T(m)$. 
Inequality (ii) follows because $\mathcal{R}(m) - \mathcal{R}(m_0) \leq \sup_{\mathbf{x}}\|m(\mathbf{x}) - m_0(\mathbf{x})\|_{2}^2 \leq d_0 \|m - m_0\|_{\infty}^2$ holds for any $m \in \mathcal{F}_{\mathrm{DSNN}}$ under the independence condition between innovation $\bm\epsilon[t]$ and predictors $\mathbf{x}[s]$ for $s<t$. 
Hence, it suffices to upper bound $d_0\inf_{m \in \mathcal{F}_{\mathrm{DSNN}}} \|m - m_0\|_{\infty}^2$ for the second term.
Combining the approximation bound in Proposition \ref{prop:approximation_our} and the generalization bound in Proposition \ref{prop:generalization_our}, we obtain the following non-asymptotic prediction error bound for the DSNN.

\begin{theorem}[Prediction error bound]\label{thm:1}
	Suppose that Assumptions \ref{assump:Boundness}--\ref{assump:Weakly dependence} hold. 
	For arbitrary $\nu_0 \in (0,1]$, let $L_j = c_2 \lceil \bar{L}_j \log(\bar{L}_j) \rceil$ and $w_j = c_3 d_{\mathrm{out},j} \lceil \bar{w}_j \log(\bar{w}_j) \rceil$ for ${j\in[3]}$.  
	If we choose 
	\begin{align*}
		\bar{L}_1 \bar{w}_1 \asymp \left(\frac{r_1 r_2 d_0 M_2^2 M_3^2 \, T^{\nu_0}}{D_1 \log^{5}(T)}\right)^{\frac{1}{4\gamma^*+2}}, \quad
		\bar{L}_2 \bar{w}_2 \asymp \left(\frac{r_2 d_0 M_3^2\, T^{\nu_0}}{D_2 \log^{5}(T)}\right)^{\frac{1}{4\gamma^*+2}}, \quad
		\bar{L}_3 \bar{w}_3 \asymp \left(\frac{d_0 T^{\nu_0}}{\log^{5}(T)}\right)^{\frac{1}{4\gamma^*+2}},
	\end{align*}
	then there exists a constant $c_5 > 0$ such that
	\begin{align} 
    \label{eq:prediction_error}
		\mathbb{E}\left[ 
		\mathcal{R}(\widehat{m}_T) - \mathcal{R}(m_0)\right]
		&\leq c_5 C_D^{\frac{1}{2\gamma^*+1}} 
        \cdot \left(\frac{ \log^5(T)}{T^{\nu_0}}\right)^{\frac{2\gamma^*}{2\gamma^*+1}}
		+ 2 C_\theta T^{\nu_0 -1},
	\end{align} 
	where $C_D=r_1r_2d_0\bar D^{2\gamma^\ast}\bar M_2\bar M_3$ with $\bar{M}_j = \max\{M_j^2,1\}$ for $j=2,3$, $\bar{D}=\max\{D_1,D_2\}$ and $\gamma^*$ defined as in Proposition \ref{prop:approximation_our}, $C_\theta$ is defined as in Proposition \ref{prop:generalization_our}, and the  expectation is taken over the training sample.
    Moreover, taking $\nu_0=(2\gamma^\ast+1)/(4\gamma^\ast+1)$ gives
\begin{equation}\label{eq:prediction_fixed_nu}
\mathbb{E}\big[\mathcal{R}(\widehat m_T)-\mathcal{R}(m_0)\big]\leq\left[c_5C_D^{\frac{1}{2\gamma^\ast+1}}(\log T)^{\frac{10\gamma^\ast}{2\gamma^\ast+1}}+2C_\theta\right]T^{-\frac{2\gamma^\ast}{4\gamma^\ast+1}}.
\end{equation}
\end{theorem}

Theorem {\ref{thm:1}} provides a non-asymptotic upper bound for the prediction error.

{{The prediction error decays at the rate of {$(\log^5(T)/T)^{2\gamma^\ast/(4\gamma^\ast+1)}$} as the sample size {{$T$}} becomes sufficiently large.}}
The term $C_D$ captures the effects of the response dimension $d_0$, factor dimensions $(r_1,r_2)$, effective dimensions $(D_1,D_2)$, and ${\ell_2}$-norm bounds $(M_2,M_3)$, while $C_\theta$ reflects temporal dependence inherited from the generalization bound in Proposition {\ref{prop:generalization_our}}. 
{{Compared with the prediction error rate {$(\log^6(n)/n)^{2\gamma^\ast/(2\gamma^\ast+1)}$} in the FAR-NN under {$i.i.d.$} sampling {\citep{fan-2023}}, our bound exhibits a slower polynomial rate due to the temporal dependence.}}
The finite-sample performance implied by this theoretical bound is examined empirically in Section {\ref{sec:simulation}}.

\begin{remark}
    {{The upper bound in {\eqref{eq:prediction_error}} can be further improved when {$T$} is sufficiently large.}} Specifically, if $T$ satisfies
    $c_5C_D^{\frac{1}{2\gamma^\ast+1}}({\log^5 T}/{T})^{\frac{2\gamma^\ast}{2\gamma^\ast+1}}
    \leq\frac{2\gamma^\ast+1}{\gamma^\ast}C_\theta
    <c_5C_D^{\frac{1}{2\gamma^\ast+1}}(\log T)^{\frac{10\gamma^\ast}{2\gamma^\ast+1}}T$ such that we can choose
    $\nu_0=
    \frac{2\gamma^\ast+1}{4\gamma^\ast+1}
    {\log\!\left(\frac{c_5\gamma^\ast C_D^{\frac{1}{2\gamma^\ast+1}}(\log T)^{\frac{10\gamma^\ast}{2\gamma^\ast+1}}T}{(2\gamma^\ast+1)C_\theta}\right)}
    /(\log T) \in(0,1]$,
    then there exists a constant $c_6>0$ such that
        $\mathbb{E}\big[\mathcal{R}(\widehat m_T)-\mathcal{R}(m_0)\big]
        \leq c_6C_D^{\frac{1}{4\gamma^\ast+1}}C_\theta^{\frac{2\gamma^\ast}{4\gamma^\ast+1}}\left({\log^5T}/{T}\right)^{\frac{2\gamma^\ast}{4\gamma^\ast+1}}$.
\end{remark}

\section{Simulation Studies}\label{sec:simulation}

This section evaluates the finite-sample predictive performance of the proposed DSNN estimated by LSE, benchmarking it against the existing mixed-frequency methods.  

Following the framework introduced in Section \ref{sec:method}, we generate mixed‑frequency data with low-, medium-, and high-frequency variables. In all experiments, we set the number of low-frequency predictors to zero ($d_1 = 0$), treating all low-frequency variables as responses. We vary the frequency ratios ($S,Q$), factor dimensions ($r_1, r_2$), lag order ($K$), variable dimensions ($d_3,d_2,d_0$), and sample size ($T$).    
Specifically, we consider four configurations: $(S,Q,r_1,r_2,K,d_3,d_2,d_0) = (3,4,1,2,2,10,20,5)$ with $T = 50$ (\textbf{Setting 1}) and $T = 100$ (\textbf{Setting 2}), and $(S,Q,r_1,r_2,K,d_3,d_2,d_0) = (30,3,5,1,4,20,30,10)$ with $T = 150$ (\textbf{Setting 3}) and $T = 300$ (\textbf{Setting 4}).  
Settings 1 and 2 represent scenarios with small frequency ratios, while Settings 3 and 4 correspond to scenarios with large frequency ratios, allowing us to assess performance across distinct frequency alignment challenges.

We first generate the high-frequency series $\{h_i[t,q,s] : i \in [d_3],\; t \in [T],\; q \in [Q],\; s \in [S]\}$ and the medium-frequency series $\{m_i[t,q] : i \in [d_2],\; t \in [T],\; q \in [Q]\}$ independently using an AR$(1)$ model with coefficient $0.6$ and innovations drawn from the uniform distribution $U[-0.5,0.5]$.
Then the low-frequency response series $\{y_i[t] : i \in [d_0],\; t \in [T]\}$ are generated according to model \eqref{eq:total_initial} with the true function $m_0 = g_3 * g_2 * g_1$, and the errors $\{\epsilon_i[t] : i \in [d_0],\; t \in [T]\}$ independently drawn from $U[-0.3,0.3]$.  
For each function $g_j: \mathbb{R}^{d_{\mathrm{in},j}} \rightarrow \mathbb{R}^{d_{\mathrm{out},j}}$ with $j\in [3]$, its $r$-th output is constructed as an additive mapping  
$g_{j,r}(\mathbf{x}) = \sum_{i=1}^{d_{\mathrm{in},j}} g_{j,r,i}(x_i)$,  
where $\mathbf{x} = (x_1,\ldots,x_{d_{\mathrm{in},j}}) \in \mathbb{R}^{d_{\mathrm{in},j}}$, and each component $g_{j,r,i}: \mathbb{R} \rightarrow \mathbb{R}$ is randomly selected from the {{following set of {$(p,C)$}-smooth functions:}}
\begin{align*}
	\left\{ \cos(\pi x),\; \sin(x),\; (1 - |x|)^2,\; \frac{1}{1 + \exp(-x)},\; 2\sqrt{|x|} - 1 \right\}.
\end{align*}

Finally, we apply the Augmented Dickey–Fuller (ADF) test to each generated series, and any series that fails the stationarity test is discarded and re‑simulated. The total number of Monte Carlo replications is set to 500.

We evaluate the forecasting performance of the DSNN estimated by LSE using a rolling-window procedure with an expanding training set. The data are divided into training, validation, and test sets with lengths $T_{\mathrm{train}}$, $T_{\mathrm{val}}$, and $T_{\mathrm{test}}$, respectively, following $T_{\mathrm{train}}:T_{\mathrm{val}}:T_{\mathrm{test}} = 8:1:1$. 
The DSNN is trained with Algorithm \ref{algorithmDSNN} using a learning rate of $10^{-4}$ over 200 epochs. For the DSNN architecture, we implement deep ReLU networks $g_{1,\mathrm{NN}}, g_{2,\mathrm{NN}}, g_{3,\mathrm{NN}}$ with depths $L_1 = L_2 = 2$ and $L_3 = 4$. The widths are chosen as $w_1 = w_2 = 16$ and $w_3 = 512$ in Settings 1 and 2, and $w_1 = w_2 = 64$ and $w_3 = 2048$ in Settings 3 and 4. 
Hyperparameters $({r}_1,{r}_2,{K})$ are selected via validation as outlined in Section \ref{subsec:hyperparameter}. The search ranges are set to $(r_{1,\max}, r_{2,\max}, K_{\max}) = (2,3,6)$ for Settings 1 and 2 and $(8,2,6)$ for Settings 3 and 4. The chosen values across the four settings are (2,3,2), (2,3,3), (4,2,5), and (5,1,3), respectively. 
At each step $t \in [T_{\mathrm{test}}]$, the DSNN is re-estimated using data up to time $T_0 + t - 1$, and a one-step-ahead forecast $\widehat{\mathbf{y}}[T_0+t] \in \mathbb{R}^{d_0}$ is generated, where $T_0 = T_{\mathrm{train}}+T_{\mathrm{val}}$. 
The forecasting performance is evaluated using the mean and standard deviation of Root Mean Square Forecast Error (RMSFE) via 500 replications, with the RMSFE defined as
\begin{align*}
	\text{RMSFE} &= \left(\frac{1}{d_0 T_{\mathrm{test}}} 
	\sum_{t=1}^{T_{\mathrm{test}}} 
	\left\| \widehat{\mathbf{y}}[T_0+t] - \mathbf{y}[T_0+t] \right\|_2^2\right)^{\frac{1}{2}}, 
\end{align*}
where $\mathbf{y}[T_0+t] \in \mathbb{R}^{d_0}$ denotes the true response vector.

To benchmark the performance of the DSNN, we compare it against existing methods that follow the High-to-Low paradigm for frequency alignment. These baseline models fall into two main categories:
\begin{itemize}
	\item \textbf{ Stacking-based methods}, where high- and medium-frequency series are first stacked into low-frequency ones and then modeled using single-frequency models:   
	\begin{itemize}
		\item \textbf{VARX:}   
		A vector autoregression with exogenous variables, treating the stacked series as exogenous inputs \citep{ghysels-2016a}. 
		
		\item \textbf{Neural networks:}  
		A Recurrent Neural Network (RNN) \citep{rumelhart-1986} and a Long Short-Term Memory (LSTM) \citep{hochreiter-1997}, both with hidden dimension 128, and a vanilla Deep Neural Network (DNN) \citep{bengio-2006} with the same hidden dimensions as our DSNN.
	\end{itemize}
	
	\item \textbf{ MIDAS-based methods}, which apply parsimonious polynomial weighting schemes to high- and medium-frequency data, and then forecast via single-frequency models: 
	\begin{itemize}
		\item \textbf{MIDAS$_1$:} Uses the Almon lag polynomial \citep{ghysels-2004}.
		\item \textbf{MIDAS$_2$:} Uses the normalized beta polynomial \citep{ghysels-2007}.
		\item \textbf{Neural networks}: The standard MIDAS regression using normalized beta polynomials, with RNN (MIDAS$_2$ + RNN) and LSTM (MIDAS$_2$ + LSTM) both with hidden dimension 128, and a vanilla DNN (MIDAS$_2$ + DNN) with the same hidden dimensions as our DSNN, applied for single-frequency forecasting.
	\end{itemize}
\end{itemize}
Note that the benchmark DNN methods and the proposed DSNN share the same forecasting setup and differ only in their frequency alignment mechanisms.

\begin{table}[ht]
\centering
\caption{Forecasting performance measured by RMSFE (standard deviation in parentheses)}
\label{tab:rmsfe_sd}
\begin{tabular}{lcccc}
\toprule
\multirow{2}{*}{\textbf{Method}} 
 & $T=50$ & $T=100$ & $T=150$ & $T=300$ \\
& $(S=3,Q=4)$ & $(S=3,Q=4)$ & $(S=30,Q=3)$ & $(S=30,Q=3)$ \\
\midrule
VARX & 0.2728 (0.0321) & 0.2540 (0.0261) & 0.2629 (0.0109) & 0.2195 (0.0061) \\
RNN & 0.2992 (0.0460) & 0.2836 (0.0310) & 0.3072 (0.0268) & 0.2642 (0.0189) \\
LSTM & 0.3361 (0.0435) & 0.3181 (0.0318) & 0.3022 (0.0235) & 0.2588 (0.0168) \\
DNN & 0.2024 (0.0222) & 0.1956 (0.0157) & 0.2566 (0.0129) & 0.2227 (0.0090) \\ \hdashline
MIDAS$_1$ & 0.2905 (0.0430) & 0.2302 (0.0231) & 0.2707 (0.0149) & -- \\
MIDAS$_2$ & 0.3022 (0.0510) & 0.2310 (0.0238) & 0.2708 (0.0148) & 0.2907 (0.0165) \\
MIDAS$_2$+RNN & 0.3435 (0.0357) & 0.2604 (0.0216) & 0.2041 (0.0094) & 0.1947 (0.0050) \\
MIDAS$_2$+LSTM & 0.3977 (0.0359) & 0.3785 (0.0259) & 0.3449 (0.0200) & 0.2556 (0.0110) \\
MIDAS$_2$+DNN & 0.3896 (0.0407) & 0.3603 (0.0273) & 0.2923 (0.0153) & 0.2504 (0.0101) \\ \hdashline
DSNN & 0.1780 (0.0168) & 0.1750 (0.0121) & 0.1777 (0.0066) & 0.1743 (0.0046) \\
\bottomrule
\end{tabular}
\end{table}

The forecasting results of DSNN and benchmark models are reported in Table \ref{tab:rmsfe_sd}. We have the following findings. First, the DSNN consistently achieves the lowest forecast errors across all settings. This superior performance is owing to its adaptive and data-driven alignment, which effectively balances dimension reduction and feature retention. In contrast, stacking-based models (both linear and nonlinear) suffer from the curse of dimensionality induced by the expanded feature space after stacking. Although MIDAS-based methods employ parametric weighting scheme to reduce dimension in frequency alignment, their pre‑specified lag structures introduce misspecification, resulting in inferior performance to the DSNN. 
Second, as the sample size $T$ increases (from Settings 1 to 2, and from Settings 3 to 4), {the forecasting accuracy of most methods improves}. Notably the DSNN remains robust even with small samples, whereas MIDAS-based methods exhibit a pronounced deterioration under limited data. Third, increasing the frequency ratios ($Q,S$) does not significantly degrade the DSNN’s performance. Most of the other methods, however, show a more noticeable decline in forecasting accuracy when moving from Settings 2 to 3.
Finally, MIDAS$_1$ becomes numerically unstable and time-consuming under high frequency ratios with large samples (Setting 4), as the exponential form of the Almon lag polynomial tends to produce explosive or extremely concentrated weights. This manifests as missing entries in Table \ref{tab:rmsfe_sd}. 

\section{Empirical Analysis}\label{sec:real data}

This section applies the DSNN to forecast five quarterly U.S. macroeconomic series---GDP, consumption, investment, government spending, and the price deflator---using economic and financial predictors available at quarterly, monthly, and daily frequencies. The data are obtained from the FRED-MD \citep{mccracken-2015} and FRED-QD \citep{mccracken-2020} databases, and are available in ALFRED at \url{https://alfred.stlouisfed.org}. The sample period spans from January 1973 to December 2024. 
Following \citet{chakraborty-2023} and \citet{giannone-2008}, the predictors are selected as follows:
\begin{itemize}
	\item Quarterly predictors (4 series): labor compensation, productivity, household assets, and corporate liabilities.  
	\item Monthly predictors (56 series), falling into six categories: (1) price and inflation, (2) credit, loans and money supply, (3) income and wages, (4) production and employment, (5) labor market and unemployment, and (6) interest rates and financial markets.
	\item Daily predictors (13 series), covering three categories: (1) stock indices, (2) H.10 foreign exchange rates, and (3) H.15 selected interest rates.
\end{itemize}
A complete listing of all response and predictor variables is provided in Tables \ref{tab:econ_indicators_part1}–\ref{tab:econ_indicators_part4} in the Appendix. 
For daily data, we first align the series to 20 trading days per month to account for varying trading calendars. 
At each forecast origin, we truncate the ragged series to the end of the preceding quarter. This excludes observations not yet released, avoiding look-ahead bias from publication delays.
The preprocessing then proceeds in three steps: (i) seasonal adjustment is applied where required, (ii) appropriate transformations are performed to achieve stationarity, and (iii) each resulting series is standardized to zero mean and unit variance. Further details are provided in Section \ref{app:real_data} of the Appendix. The resulting dataset contains $T = 208$ quarterly observations, $QT = 624$ monthly observations, and $SQT = 12{,}480$ daily observations, with $Q=3$ and $S=20$.

To isolate the effect of the COVID‑19 pandemic which triggered abrupt changes in production, investment, and consumption, we analyze two sample periods:
\begin{itemize}
	\item \textbf{Period I (Pre-COVID):} 1973--2019. The data are split into 148 quarterly training observations, 20 validation observations, and 20 test observations (testing window: March 2015--December 2019).
	\item \textbf{Period II (COVID \& Post-COVID):} 1973--2024. This period contains 168 quarterly training observations, 20 validation observations, and 20 test observations (testing window: March 2020--December 2024).
\end{itemize}
Period II includes the phase of the pandemic (2020--2021) and its subsequent recovery phase, thus allowing us to evaluate forecasting performance under high-volatility conditions.

We train the DSNN using Algorithm \ref{algorithmDSNN} with a learning rate of $10^{-5}$, a batch size of 64, and 200 epochs. The architecture consists of deep ReLU networks $g_{j,\mathrm{NN}}$ for $j\in [3]$, each with depth $L_j = 2$. The hidden widths are set to $w_1 = w_2 = 64$ and $w_3 = 128$.
Hyperparameters $({r}_1,{r}_2,K)$ are selected via validation following the procedure in Section \ref{subsec:hyperparameter}. The search ranges are $(r_{1,\max}, r_{2,\max}) = (8,6)$ and $K$ is chosen from $\{1,4,8,12\}$, corresponding to lags of one quarter, one year, two years, and three years for a one-quarter-ahead forecast. The selected hyperparameters are $(5,2,1)$ for Period I and $(5,4,1)$ for Period II. The choice of $K=1$ is supported by the partial autocorrelation functions of most series in this dataset, which show the strongest autocorrelation at the first lag and relatively weak higher‑order dependencies, consistent with \citet{chakraborty-2023}.
Forecasts are produced using a rolling forecasting procedure with expanding window as in Section \ref{sec:simulation}, and the predictive performance is evaluated with the RMSFE and MAFE. 

For comparison, we consider the following mixed-frequency methods:
\begin{itemize}
	
	\item 
    \textbf{Stacking-based methods}, where high- and medium-frequency series are first stacked into low-frequency ones and then modeled using single-frequency models:   
	\begin{itemize}[leftmargin=*]
		\item \textbf{(M1) VARX:} Vector autoregression with exogenous variables; see Section~\ref{sec:simulation}.
		
		\item \textbf{(M2) Neural networks}:   
		A RNN and a LSTM both with a hidden size of 128, and a vanilla DNN with the same hidden dimensions as our DSNN.
		
		\item \textbf{(M3) BMF}: Bayesian mixed-frequency model of \citet{chakraborty-2023}. The prior distribution of decay parameter $\theta$ is $U[0,1]$. Gibbs sampling is performed with 1{,}000 burn-in iterations followed by 2{,}000 additional iterations.
	\end{itemize}
	
	\item 
    \textbf{MIDAS-based methods}, which apply weighting schemes to high- and medium-frequency data and then modeled using single-frequency models: 
	\begin{itemize}[leftmargin=*]
		\item \textbf{(M4) VARX}: The standard MIDAS regression using Almon lag polynomials \\ (MIDAS$_1$+VARX) and normalized beta polynomials (MIDAS$_2$+VARX), with VARX model applied for single-frequency forecasting. 
		
		\item \textbf{(M5) Neural networks}: The standard MIDAS regression using normalized beta polynomials, with RNN (MIDAS$_2$ + RNN) and LSTM (MIDAS$_2$ + LSTM) both with hidden dimension 128, and a vanilla DNN (MIDAS$_2$ + DNN) with the same hidden dimensions as our DSNN, applied for single-frequency forecasting.
		
		\item \textbf{(M6) Seq2one}: LSTM-based model from \citet{lin-2024}. The hidden state dimensions of pre- and post-attention LSTM are $n_a=128$ and $n_s=256$, and the final fully connected layer has the dimension of 128. The dropout ratio is 0.4.
		
		\item \textbf{(M7) Transformer}: Transformer-based model from \citet{lin-2024}. The encoder and decoder are configured with a key dimension of 16, with fully-connected layers of dimensions 76 and 32 respectively. It uses 4 attention heads, followed by a feed‑forward layer with 128 units. The dropout ratio is 0.4.
	\end{itemize}
\end{itemize} 

Forecasting is performed under three information sets: (S1) quarterly predictors only, (S2) quarterly and monthly predictors, and (S3) quarterly, monthly and daily predictors. Particularly, stacking-based methods (M1--M2) are applied in all scenarios, and MIDAS methods (M4--M5) are used in (S2)--(S3), while the other models (M3, M6, M7) are evaluated only in (S2) for their two-frequency design. All comparisons use the same rolling forecasting procedure with an expanding window.

The forecasting results of the DSNN and seven benchmark models are reported in Table \ref{tab:real-data}. 
First, the DSNN consistently achieves the lowest forecast errors in both periods, demonstrating robust performance even during the high-volatility COVID \& Post‑COVID era. Its adaptability is reflected in the chosen factor dimensions: $r_2=2<Q=3$ (dimension reduction) in the stable Pre‑COVID period and $r_2=4>Q=3$ (feature expansion) in the volatile later period, illustrating how the method dynamically balances information compression and retention. 
In contrast, the standard MIDAS methods with the linear VARX model (M4) perform poorly overall, particularly when daily predictors (with frequency ratio $S=20$) are introduced, a consequence of its rigid and pre‑specified lag structure. Nonlinear MIDAS variants (M5--M7) perform better in Period II, indicating that nonlinear specifications can capture more complex dependencies during turbulent phases.
Moreover, when only quarterly predictors are used (Scenario S1), standard RNN-type models deliver acceptable accuracy, whereas VARX deteriorates seriously in Period II. Finally, adding higher-frequency data via stacking generally impairs performance due to the curse of dimensionality, a drawback that affects both linear and nonlinear stacking-based methods. An exception is the BMF model, which ranks second overall, approaching the DSNN owing to its structured parameter matrix that effectively reduces parameter dimensionality.

\begin{table}[t!]
	\centering
	\caption{Comparison of forecasting performance between DSNN and benchmark models. Period I: Pre-COVID period; Period II: covers the COVID \& Post-COVID period. Scenarios S1--S3 correspond to different information sets.}
	\label{tab:real-data}
	\setlength{\extrarowheight}{0pt}
	\begin{tabular}{llcccc}
		\toprule
		\multirow{2}{*}{\textbf{Scenario}} & \multirow{2}{*}{\textbf{Method}} & \multicolumn{2}{c}{\textbf{Period I}} & \multicolumn{2}{c}{\textbf{Period II}} \\
		\cmidrule(lr){3-4} \cmidrule(lr){5-6}
		&  & RMSFE & MAFE & RMSFE & MAFE \\
		\midrule
		\multirow{4}{*}{\shortstack{S1 \\ \\ \\ Quarterly }}
		& VARX        & 0.555 & 0.412 & 2.872 & 1.392 \\
		& RNN         & 0.564 & 0.381 & 2.029 & 1.106 \\
		& LSTM        & 0.558 & 0.375 & 2.014 & 1.103 \\ 
		& DNN         & 0.561 & 0.378 & 2.018 & 1.104 \\
		\hline
		\multirow{12}{*}{\shortstack{S2 \\ \\ \\ Quarterly \\ + \\ Monthly }}
		& VARX            & 0.726 & 0.523 & 2.364 & 1.554 \\
		& RNN             & 0.571 & 0.419 & 1.991 & 1.107 \\
		& LSTM            & 0.559 & 0.372 & 2.007 & 1.097 \\
		& DNN             & 0.560 & 0.374 & 2.026 & 1.109 \\
		& BMF             & 0.537 & 0.384 & 1.626 & 1.003 \\
		\cdashline{2-6}
		& MIDAS$_1$+VARX  & 0.700 & 0.525 & 4.306 & 1.645 \\
		& MIDAS$_2$+VARX  & 0.768 & 0.548 & 3.311 & 1.417 \\
		& MIDAS$_2$+RNN   & 0.696 & 0.515 & 1.745 & 1.139 \\
		& MIDAS$_2$+LSTM  & 0.716 & 0.532 & 2.031 & 1.184 \\ 
		& MIDAS$_2$+DNN   & 0.688 & 0.516 & 1.955 & 1.224 \\ 
		& Seq2one         & 0.650 & 0.452 & 1.693 & 1.027 \\
		& Transformer     & 0.759 & 0.589 & 1.785 & 1.139 \\
		\hline
		\multirow{10}{*}{\shortstack{S3 \\ \\ \\ Quarterly \\ + \\ Monthly \\ + \\ Daily}}
		& VARX            & 0.638 & 0.494 & 2.146 & 1.366 \\
		& RNN             & 0.756 & 0.541 & 1.814 & 1.236 \\
		& LSTM            & 0.573 & 0.399 & 1.991 & 1.144 \\
		& DNN             & 0.540 & 0.412 & 1.846 & 1.076 \\
		\cdashline{2-6}
		& MIDAS$_1$+VARX  & 0.885 & 0.641 & 3.032 & 1.416 \\
		& MIDAS$_2$+VARX  & 0.818 & 0.616 & 3.336 & 1.539 \\
		& MIDAS$_2$+RNN   & 0.554 & 0.404 & 1.932 & 1.101 \\
		& MIDAS$_2$+LSTM  & 0.558 & 0.383 & 1.994 & 1.095 \\ 
		& MIDAS$_2$+DNN   & 0.593 & 0.410 & 1.835 & 1.044 \\ \cdashline{2-6}      
		& DSNN (ours) & \textbf{0.523} & \textbf{0.372} & \textbf{1.587} & \textbf{0.970} \\
		\bottomrule
	\end{tabular}
\end{table}

\section{Conclusion and Discussion}\label{sec:conclusion and discussion}

This paper proposes a novel Depth-Separable Neural Network (DSNN) for mixed-frequency data forecasting. The DSNN integrates adaptive frequency alignment with nonlinear representation learning, overcoming the parametric, linear, and dimension-reducing constraints of conventional mixed-frequency models. Specifically, parameter sharing at the stage ensures separability and scalability with many higher-frequency predictors. Theoretically, we establish the approximation and prediction error bounds for the DSNN. 
Simulation studies show that our method outperforms existing approaches across varying frequency ratios, particularly with numerous higher-frequency time series. In an empirical application to forecasting U.S. quarterly macroeconomic variables using monthly and daily indicators, the DSNN achieves significant predictive gains over existing mixed-frequency methods, highlighting its practical utility. 

This paper can be extended along three directions.
First, the alignment functions in the model could be generalized to semi‑parametric forms, balancing interpretability and flexibility for improved frequency alignment.
Second, while the DSNN employs neural networks for factor extraction, alternative approaches such as variational autoencoders (VAEs) could be integrated to provide probabilistic latent representations and further refine the alignment process \citep{xiu-2024}. 
Third, in nature the mixed-frequency time series belong to a type of multi-modal data. We may extend the proposed alignment method to incorporate multi-modal time series data (e.g., video, audio, and text), and it could significantly enhance forecasting accuracy for target series \citep{cheng-2022,jia-2024,liu-2025,jiang-2025}.  




\clearpage
\newpage

\addtolength{\oddsidemargin}{-.1in}
\addtolength{\textwidth}{.2in}
\setlength{\columnwidth}{\textwidth}
\setlength{\linewidth}{\textwidth}
\setlength{\hsize}{\textwidth}

\setcounter{section}{0}           
\setcounter{equation}{0}          
\setcounter{figure}{0}            
\setcounter{table}{0}             

\numberwithin{equation}{section}
\numberwithin{lemma}{section}
\numberwithin{theorem}{section}
\numberwithin{assumption}{section}
\numberwithin{proposition}{section}
\numberwithin{definition}{section}
\numberwithin{figure}{section}

\renewcommand{\thesection}{A.\arabic{section}}
\renewcommand{\theHsection}{appendix.\arabic{section}}
\renewcommand{\thelemma}{A.\arabic{lemma}}
\renewcommand{\thefigure}{A.\arabic{figure}}
\renewcommand{\thetable}{A.\arabic{table}}

\begin{center}
    \singlespacing 
    {\Large\bf Appendix for ``Mixed-Frequency Time Series Forecasting via Depth-Separable Neural Networks''}
    
    \vspace{1.5em}
    
\end{center}

\vspace{2em}

\doublespacing
\renewcommand{\arraystretch}{1.0}
\begin{abstract}
	\vspace{-2mm}
	This appendix provides technical details for all theorem and propositions, as well as additional empirical results.
	Specifically, Section \ref{app:sec-theory} presents the detailed proofs for Propositions \ref{prop:approximation_our}--\ref{prop:generalization_our} and Theorem \ref{thm:1}. Section \ref{app:sec-lemma} introduces lemmas which give some preliminary results for proving the aforementioned propositions.    
	Finally, Section \ref{app:real_data} reports additional empirical results that complement those presented in the manuscript. 
	Throughout the appendix, we adopt the following notation. Let $[m] = \{1, \ldots, m\}$. Denote by $\mathbb{R}$ the set of real numbers, $\mathbb{N}_0 = \{0,1,2,\dots\}$ the set of nonnegative integers, and $\mathbb{N}_+ = \{1,2,3,\dots\}$ the set of positive integers. 
	Bold lowercase letters denote vectors, e.g., $\mathbf{x} = [x_1, \ldots, x_d]^\top$ for a $d$-dimensional vector. Its $\ell_q$ norm is $\|\mathbf{x}\|_q = \left(\sum_{i=1}^{d} |x_i|^q \right)^{1/q}$, and its $\ell_\infty$ norm is $\|\mathbf{x}\|_\infty = \max_{1 \leq i \leq d} |x_i|$. 
	Bold uppercase letters denote matrices, e.g., $\mathbf{A} = [A_{i,j}]_{i \in [n], j \in [m]}$. We define $\|\mathbf{A}\|_2 = \sup_{\mathbf{x} : \|\mathbf{x}\|_2 = 1} \|\mathbf{A}\mathbf{x}\|_2$
    and $\|\mathbf{A}\|_F = \sqrt{\sum_{i,j} A_{i,j}^2}$.
	For functions $f(n)$ and $g(n)$, we write $f(n) \lesssim g(n)$ to mean $f(n) \leq C \cdot g(n)$ for some constant $C$ independent of $n$. Similarly, $f(n) \gtrsim g(n)$ 
	means there exists a constant $C > 0$ independent of $n$ such that $f(n) \geq C \cdot g(n)$. We write $f(n) \asymp g(n)$ if both $f(n) \lesssim g(n)$ and $f(n) \gtrsim g(n)$ hold. Moreover, the union of function class is denoted by $\bigcup_{j=1}^n \mathcal{F}_j \;=\; \{ f \mid f \in \mathcal{F}_j \text{ for some } j \in \{1,\ldots,n\} \}$.
    In addition, denote $\mathcal{N}_p(\mathbb{T}, \epsilon)$, 
    $\mathcal{N}_{\infty}(\mathbb{T}, \epsilon)$ and $\mathcal{N}_F(\mathbb{T}, \epsilon)$ as the $\epsilon$-covering numbers of the space $\mathbb{T}$ under the $\ell_p$-norm, the $\ell_{\infty}$-norm and the Frobenius norm, respectively.  
\end{abstract}

\newpage
\section*{Contents}
\startcontents
\printcontents{ }{1}{}

\section{Proofs of Main Results} \label{app:sec-theory}
This section presents the proofs of Propositions \ref{prop:approximation_our}--\ref{prop:generalization_our} and Theorem \ref{thm:1} in the manuscript. We first recall the function classes $\mathcal{F}_{\mathrm{DSNN}}$ and $\mathcal{F}_{\ast}$ defined for $m$ and $m_0$, respectively. 
For $m\in\mathcal{F}_{\mathrm{DSNN}}$, we assume $m= g_{3,\mathrm{NN}} \ast g_{2,\mathrm{NN}} \ast g_{1,\mathrm{NN}}$ with $g_{j,\mathrm{NN}}\in \mathcal{G}(L_j,{d}_{\mathrm{in},j},{d}_{\mathrm{out},j},w_j)$ for $j \in [3]$, where $\mathcal{G}$ is the deep ReLU network class defined in Definition \ref{def:DNN} of the manuscript, $L_j,w_j \in \mathbb{N}_+$, the input dimensions are $(d_{\mathrm{in},1},d_{\mathrm{in},2},d_{\mathrm{in},3})=(S,Q,KD_3)$ with $D_3 = d_0 + d_1 + r_2 d_2 + r_1 r_2 d_3$, and the output dimensions are $(d_{\mathrm{out},1},d_{\mathrm{out},2},d_{\mathrm{out},3})=(r_1,r_2,d_0)$.  
For $m_0\in\mathcal{F}_{\ast}$, we assume $m_0= g_3 \ast g_2 \ast g_1$ with each output component $g_{j,i}\in \mathcal{H}({d}_{\mathrm{in},j},l_j,\mathcal{P})$ for $j \in [3]$ and $i\in [{d}_{\mathrm{out},j}]$, where $\mathcal{H}$ is the hierarchical composition model class defined in Definition \ref{def:HCM} of the manuscript.

\subsection{Proof of Proposition~\ref{prop:approximation_our}}
\label{app:proof_approximation_our}
\begin{proof}
	This proposition derives an upper bound for the approximation error when the hierarchical function $m_0 = g_3 \ast g_2 \ast g_1$ is approximated by the DSNN $m = g_{3,\mathrm{NN}} \ast g_{2,\mathrm{NN}} \ast g_{1,\mathrm{NN}}$. The proof proceeds in two steps. 
	First, Lemma \ref{lmma:approximation} provides a bound for approximating a single component $g_j\in\mathcal{H}(d,l,\mathcal{P})$ by a deep ReLU network $g_{j,\mathrm{NN}}\in\mathcal{G}$. 
	Then, it suffices to analyze how these componentwise errors propagate through the hierarchical structure. Combining these results yields the overall approximation bound for the DSNN.
	
	Denote by $\epsilon_j \coloneqq \| g_{j,\mathrm{NN}} - g_j \|_{\infty}$ the componentwise approximation error for $j \in [3]$. We have
	\begin{align} \label{ineq:approximation_begin}
		&\|g_{3,\rm{NN}} \ast g_{2,\rm{NN}} \ast g_{1,\rm{NN}} - g_3 \ast g_2 \ast g_1\|_{\infty} \nonumber\\
		\leq &  {\|g_{3,\rm{NN}} \ast g_{2,\rm{NN}} \ast g_{1,\rm{NN}} 
			- g_{3,\rm{NN}} \ast g_{2,\rm{NN}} \ast g_1\|_{\infty}} 
		+ {\|g_{3,\rm{NN}} \ast g_{2,\rm{NN}} \ast g_1 
			- g_{3,\rm{NN}} \ast g_2 \ast g_1\|_{\infty}} \nonumber\\
		&+ {\|g_{3,\rm{NN}} \ast g_2 \ast g_1 
			- g_3 \ast g_2 \ast g_1\|_{\infty}} \nonumber\\
		= &  \underbrace{\|g_{3,\rm{NN}} \ast g_{2,\rm{NN}} \ast \epsilon_1    \|_{\infty}}_{\text{(I)}}  
		+ \underbrace{\|g_{3,\rm{NN}} \ast \epsilon_2\|_{\infty}}_{\text{(II)}} 
		+ {\epsilon_3}.
	\end{align}
	The term (I) arises from the approximation error $\epsilon_1$ between $g_{1,\mathrm{NN}}$ and $g_1$, which is sequentially propagated through $g_{2,\mathrm{NN}}$ and $g_{3,\mathrm{NN}}$.   
	For $l \in [L_2+L_3]$, let $\mathbf{y}_{l}$ and $\mathbf{y}^{\prime}_{l}$ be the $l$-th layer inputs of $g_{3,\mathrm{NN}} \ast g_{2,\mathrm{NN}}$ corresponding to $\mathbf{y}_{1} = g_1(\mathbf{x}) \in \mathbb{R}^{r_1}$ and $\mathbf{y}^{\prime}_{1} = g_{1,\mathrm{NN}}(\mathbf{x}) \in \mathbb{R}^{r_1}$, respectively.  
	Note that $\|\mathbf{W}_{j,l}\|_2 \leq U_{j,l}$ and $\prod_{l=1}^{L_j+1}  U_{j,l} \leq M_j$ by Assumption \ref{assump:Boundness}(iii), and the last layer of $g_{j,\mathrm{NN}}$ uses linear mapping.
	Then for the error propagated through $g_{2,\mathrm{NN}}: \mathbb{R}^{Q} \rightarrow \mathbb{R}^{r_2}$, 
	\begin{align*}
		\| g_{2,NN} \ast \epsilon_1\|_{\infty}
		&
		\overset{(i)}\leq \|(\mathbf{W}_{2,L_2+1} \mathbf{y}_{L_2+1} + \mathbf{b}_{2,L_2+1}) - (\mathbf{W}_{2,L_2+1} \mathbf{y}^{\prime}_{L_2+1} + \mathbf{b}_{2,L_2+1})\|_{2} \\
		&\overset{(ii)}{\leq}   U_{2,L_2+1} \| \mathbf{y}_{L_2+1} - \mathbf{y}^{\prime}_{L_2+1}\|_{2}
		{=} U_{2,L_2+1}\|\sigma(\mathbf{W}_{2,L_2} \mathbf{y}_{L_2} + \mathbf{b}_{2,L_2}) - \sigma(\mathbf{W}_{2,L_2} \mathbf{y}^{\prime}_{L_2} + \mathbf{b}_{2,L_2})\|_{2} \\
		&\overset{(iii)}{\leq} U_{2,L_2+1} \| \mathbf{W}_{2,L_2} (\mathbf{y}_{L_2} - \mathbf{y}^{\prime}_{L_2})\|_{2} 
		{\leq}   U_{2,L_2+1} U_{2,L_2} \| \mathbf{y}_{L_2} - \mathbf{y}^{\prime}_{L_2}\|_{2} 
		\leq \cdots \\
		&\leq \prod_{l=1}^{L_2+1} ( U_{2,l}) \|\mathbf{y}_{1}- \mathbf{y}^{\prime}_{1}\|_{2} 
		\overset{(iv)}{\leq} \prod_{l=1}^{L_2+1} ( U_{2,l}) \sqrt{r_1} 
		\|\mathbf{y}_{1}- \mathbf{y}^{\prime}_{1}\|_{\infty}
		\leq \sqrt{r_1}  M_2 \epsilon_1.
	\end{align*}
	Inequalities ($i$) and ($iv$) use the fact that $\|\mathbf{x}\|_{\infty} \leq \|\mathbf{x}\|_{2} \leq \sqrt{d} \|\mathbf{x}\|_{\infty}$ for any vector $\mathbf{x}\in \mathbb{R}^d$.  
	Inequality ($ii$) applies $\|\mathbf{Ax}\|_{2} \leq \|\mathbf{A}\|_{2}\|\mathbf{x}\|_{2}$.  
	Inequality ($iii$) uses the fact that $\sigma$ is Lipschitz with constant one.  
	Similarly, the error $\| g_{2,NN} \ast \epsilon_1\|_{\infty}$ is further propagated through $g_{3,\mathrm{NN}}$, leading to (I) $\leq \sqrt{r_1 r_2}\, M_2 M_3 \epsilon_1$. 
	
	The second term (II) comes from the approximation error $\epsilon_2$ between $g_{2,\rm{NN}}$ and $g_2$, which is propagated through $g_{3,\mathrm{NN}}$. Similar to the proof of $\| g_{2,NN} \ast \epsilon_1\|_{\infty}$, we obtain (II) $\leq \sqrt{r_2} M_3 \epsilon_2$. 
	
	These together with \eqref{ineq:approximation_begin} and Lemma \ref{lmma:approximation}, imply that
	\begin{align}\label{ineq:supp_approximation}
		&\sup_{m_0 \in \mathcal{F}_{\ast}} 
		\inf_{m \in \mathcal{F}_{\mathrm{DSNN}}} \|m - m_0\|_{\infty} \nonumber\\
		= &
		\sup_{\substack{ g_{j,i} \in \mathcal{H}(d_{\mathrm{in},j}, l_j, \mathcal{P}) \\ 
				j \in [3], i \in [d_{\mathrm{out},j}]
		}} 
		\inf_{\substack{{g}_{j,\mathrm{NN}} \in \mathcal{G}(L_j,d_{\mathrm{in},j},d_{\mathrm{out},j},w_j) \nonumber\\
				j \in [3]
		}} 
		\|g_{3,\rm{NN}} \ast g_{2,\rm{NN}} \ast g_{1,\rm{NN}} - g_3 \ast g_2 \ast g_1\|_{\infty} \\
		\leq & \sqrt{r_1 r_2}M_2M_3   \epsilon_1
		+ \sqrt{r_2}M_3 \epsilon_2
		+ \epsilon_3 \nonumber\\
		\leq &c_1 \left(\sqrt{r_1 r_2}M_2M_3  (\bar{L}_1 \bar{w}_1)^{-2 \gamma^*} 
		+ \sqrt{r_2}M_3 (\bar{L}_2 \bar{w}_2)^{-2 \gamma^*} 
		+ (\bar{L}_3 \bar{w}_3)^{-2 \gamma^*}\right),
	\end{align}
	where $c_1$ is a positive constant. 
	This completes the proof.
\end{proof}

\subsection{Proof of Proposition~\ref{prop:generalization_our}}
\label{app:proof_generalization_our}
\begin{proof}

	Note that 
		$\mathbb{E}
		\sup_{m \in \mathcal{F}_{\mathrm{DSNN}}} \left| \mathcal{R}(m) - {\mathcal{R}}_T(m) \right|
		= \int_0^{\infty} \mathbb{P}\left\{\sup_{m \in \mathcal{F}_{\mathrm{DSNN}}} \left| \mathcal{R}(m) - {\mathcal{R}}_T(m) \right| > t\right\} dt$.
	To bound the generalization error of the DSNN, we first derive a concentration inequality in Lemma \ref{lmma:generalization_error_tail_bound}, formulated as follows:
	\begin{align*}
		&\mathbb{P}\left( \sup_{m \in \mathcal{F}_{\mathrm{DSNN}}} \left| \mathcal{R}(m) - {\mathcal{R}}_T(m) \right| > \epsilon \right)
		\leq 2\mathcal{N}_{\infty}\left( \mathcal{F}_{\mathrm{DSNN}}, \frac{\epsilon}{4G} \right) 
		\exp\left( \frac{-T^{\nu_0} \epsilon + T^{2 \nu_0 -1} {C_\theta}}{2{C^2}} \right),
	\end{align*}
	where ${T \in \mathbb{N}_+}$, $\epsilon > 0$, $\nu_0 \in {(0,1]}$,
	$ {\theta_{\infty,T}} $ is defined in Definition \ref{def:weak_dependence}, $U_x \coloneqq \sup_t\|\mathbf{x}[t]\|_2 < \infty$ under Assumption \ref{assump:Boundness}(i), and $G$ {and $C$} are defined in Lemma \ref{lmma:generalization_error_tail_bound}.
	To complete the proof, we subsequently derive the expectation bound from this probability by controlling the covering number, which is provided separately in Lemma \ref{lmma:covering_our}.
	
	Recall that $C_\theta$ is defined in Proposition \ref{prop:generalization_our}.
    By Lemma \ref{lmma:generalization_error_tail_bound} and 
    $\mathcal{N}_{\infty}\left( \mathcal{F}, {\epsilon_1} \right)\leq\mathcal{N}_{\infty}\left( \mathcal{F}, {\epsilon_2} \right)$ for any $\mathcal{F}$ and $\epsilon_2 \leq\epsilon_1$, for any $\epsilon \geq \epsilon_T = T^{-L^\ast}$, it follows that
	\begin{align}\label{eq:integration}
		&\mathbb{E}
		\sup_{m \in \mathcal{F}_{\mathrm{DSNN}}} \left| \mathcal{R}(m) - {\mathcal{R}}_T(m) \right|
		\nonumber\\
		= &\int_0^{\epsilon} \mathbb{P}\left\{\sup_{m \in \mathcal{F}_{\mathrm{DSNN}}} \left| \mathcal{R}(m) - {\mathcal{R}}_T(m) \right| > t\right\} dt 
		+ \int_{\epsilon}^\infty \mathbb{P}\left\{\sup_{m \in \mathcal{F}_{\mathrm{DSNN}}} \left| \mathcal{R}(m) - {\mathcal{R}}_T(m) \right| > t\right\} dt \nonumber\\
		\leq &\epsilon + 2\exp\left(\frac{ {C_\theta}T^{2\nu_0 -1}}{2C^2} \right) 
		\int_{\epsilon}^\infty 
		\mathcal{N}_{\infty}\left( \mathcal{F}_{\mathrm{DSNN}},
		\frac{t}{4G} \right)
		\exp\left( - \frac{1}{2C^2} T^{\nu_0} t \right) dt \nonumber\\
		\leq &\epsilon + 2\exp\left( \frac{ {C_\theta}T^{2\nu_0 -1}}{2C^2} \right)  
		\mathcal{N}_{\infty}\left( \mathcal{F}_{\mathrm{DSNN}}, \frac{\epsilon_T}{4G} \right) 
		\int_{\epsilon}^\infty 
		\exp\left( - \frac{1}{2C^2} T^{\nu_0} t \right) dt \\
		= &\epsilon + {4C^2}\exp\left( \frac{ {C_\theta}T^{2\nu_0 -1}}{2C^2} \right)
		\mathcal{N}_{\infty}\left( \mathcal{F}_{\mathrm{DSNN}}, \frac{\epsilon_T}{4G} \right)  
		T^{-\nu_0} \exp\left( - \frac{1}{2C^2} T^{\nu_0} \epsilon \right), \nonumber
	\end{align}
	which is a function of $ \epsilon $ and attains its minimum at
	\[
	\epsilon = 2C^2T^{-\nu_0} 
	\log\left( 2\mathcal{N}_{\infty}
	\left( \mathcal{F}_{\mathrm{DSNN}},
	\frac{\epsilon_T}{4G} \right) \right) 
	+ { {C_\theta} T^{\nu_0 -1}},
	\]
    which exceeds {$\epsilon_T = T^{-L^\ast} $}
	since $C\geq1$ and the covering number is at least one. Consequently, by Lemma \ref{lmma:covering_our}, there exist some constants $c_4, c_7>0$ such that
	\begin{align}\label{ineq:supp_generalization}
		&\mathbb{E} \sup_{m \in \mathcal{F}_{\mathrm{DSNN}}} \left| \mathcal{R}(m) - {\mathcal{R}}_T(m) \right|  \nonumber\\
		\leq & T^{-\nu_0} \left\{
		2{C^2} \log \left(
		{2}\mathcal{N}_{\infty}\left( 
		\mathcal{F}_{\mathrm{DSNN}}, \frac{\epsilon_T}{4G} \right)  
		\right) 
		+ {C_\theta}{T^{2\nu_0 -1}} +2{C^2}
		\right\} \nonumber\\
		\overset{(i)}{\leq } & 2{C^2}T^{-\nu_0} 
		\log \left(
		{2}\mathcal{N}_{2}\left( 
		\mathcal{F}_{\mathrm{DSNN}}, \frac{\epsilon_T}{4\sqrt{d_0}G} \right)  
		\right) 
		+ { {C_\theta} T^{\nu_0 -1}} + 2{C^2}T^{-\nu_0}
		\nonumber \\
		\overset{(ii)}{\leq }& 2 c_7 C^2 T^{-\nu_0} 
		({D_1L_1w_{1}^2 + D_2L_2  w_{2}^2 + L_3 w_{3}^2}) 
		\log({4 \sqrt{d_0} G L^\ast w^\ast \bar{U}_x(\bar{U} T)^{L^\ast}}) 
		+ {C_\theta} T^{\nu_0 -1}  + {2C^2(1+\log 2)}T^{-\nu_0} \nonumber\\
		\leq & c_4 T^{-\nu_0} \log(T)
		L^\ast ({D_1L_1w_{1}^2 + D_2L_2  w_{2}^2 + L_3 w_{3}^2})
		+  {C_\theta} T^{\nu_0 -1} \nonumber\\
        \coloneqq & c_4{C_{\operatorname{NN}}}(\log T)T^{-\nu_0}+C_\theta T^{\nu_0-1},
	\end{align}
	where 
	$\bar{U}=\max_{j,l}\{U_{j,l},1\}$ with $U_{j,l}$ defined in Assumption \ref{assump:Boundness}(iii), $\bar{U}_x \coloneqq \max\{\sup_{t}\|\mathbf{x}[t]\|_2,1\} < \infty$ under Assumption \ref{assump:Boundness}(i),
    and $C_{\operatorname{NN}}$ is defined in Proposition \ref{prop:generalization_our}.
	Inequality (i) is due to the fact that $\mathcal{N}_{\infty}(\mathcal{F},\epsilon) \leq \mathcal{N}_{2}(\mathcal{F},\epsilon/\sqrt{d})$ for any $\epsilon>0$ and arbitrary space $\mathcal{F}$ with output dimension $d$. Inequality (ii) follows from Lemma \ref{lmma:covering_our}.
    
If $\log T/T\leq C_\theta/(c_4C_{\operatorname{NN}})<T\log T$, then \eqref{ineq:supp_generalization} achieves its minimum when $\nu_0=\log(c_4C_{\operatorname{NN}}T\log T/C_\theta)/(2\log T)\in(0,1]$, which leads to
\begin{equation*}
\mathbb{E}\sup_{m\in\mathcal{F}_{\mathrm{DSNN}}}\left|\mathcal{R}(m)-\mathcal{R}_T(m)\right|\leq2\sqrt{c_4C_{\operatorname{NN}}C_\theta}\left(\frac{\log T}{T}\right)^{1/2}.
\end{equation*}
	The proof is complete.
\end{proof}

\subsection{Proof of Theorem~\ref{thm:1}}
\label{app:proof_thm1}
\begin{proof}
	For any $m \in \mathcal{F}_{\mathrm{DSNN}}$, the excess risk $\mathcal{R}(\widehat{m}_T) - \mathcal{R}(m_0)$ can be decomposed as follows 
	\begin{align}\label{ineq:supp_decomposition}
		\mathcal{R}(\widehat{m}_T) - \mathcal{R}(m_0) &= \mathcal{R}(\widehat{m}_T) - \mathcal{R}_T(\widehat{m}_T) + \mathcal{R}_T(\widehat{m}_T) - \mathcal{R}_T(m) 
		+ \mathcal{R}_T(m) - \mathcal{R}(m) + \mathcal{R}(m) - \mathcal{R}(m_0) \nonumber\\
		& \overset{(i)}{\leq} \left\{\mathcal{R}(\widehat{m}_T) - {\mathcal{R}}_T(\widehat{m}_T)\right\} + \left\{{\mathcal{R}}_T(m) - \mathcal{R}(m)\right\} + \left\{\mathcal{R}(m) - \mathcal{R}(m_0)\right\} \nonumber\\
		& \overset{(ii)}{\leq} 2 \sup_{m \in \mathcal{F}_{\mathrm{DSNN}}} \left| {\mathcal{R}}(m) - \mathcal{R}_T(m) \right| + 
        d_0 \|m - m_0\|_{\infty}^2. 
	\end{align} 
	Inequality (i) uses the definition of the LSE, which implies ${\mathcal{R}}_T(\widehat{m}_T) \leq {\mathcal{R}}_T(m)$. 
    Inequality (ii) follows because $\mathcal{R}(m) - \mathcal{R}(m_0) \leq \sup_{\mathbf{x}}\|m(\mathbf{x}) - m_0(\mathbf{x})\|_{2}^2 \leq d_0 \|m - m_0\|_{\infty}^2$ holds for any $m \in \mathcal{F}_{\mathrm{DSNN}}$ under the independence condition between innovation $\bm\epsilon[t]$ and predictors $\mathbf{x}[s]$ for $s<t$. 
    Hence, it suffices to upper bound $d_0\inf_{m \in \mathcal{F}_{\mathrm{DSNN}}} \|m - m_0\|_{\infty}^2$ for the second term.
	Combining the approximation bound \eqref{ineq:supp_approximation} in Proposition \ref{prop:approximation_our}, the generalization bound \eqref{ineq:supp_generalization} in Proposition \ref{prop:generalization_our} and \eqref{ineq:supp_decomposition}, then there exist some constants $c_1,c_4>0$ such that
	\begin{align*} 
		\mathbb{E}\left[ 
		\mathcal{R}(\widehat{m}_T) - \mathcal{R}(m_0)\right]
		&\leq 2c_4 T^{-\nu_0} \log(T)
		{L^\ast \big(D_1L_1w_{1}^2 + D_2L_2w_{2}^2 + L_3w_{3}^2\big)}
		+ 2 {C_\theta} T^{\nu_0 -1}  \\
		& \quad + {3c_1^2} {r_1 r_2} d_0 M_2^2 M_3^2 (\bar{L}_1 \bar{w}_1)^{-4 \gamma^*} 
		+ {3c_1^2} {r_2} d_0 M_3^2 (\bar{L}_2 \bar{w}_2)^{-4 \gamma^*}
		+ {3c_1^2} d_0 (\bar{L}_3 \bar{w}_3)^{-4 \gamma^*},
	\end{align*}
	where $U_x \coloneqq\sup_t\|\mathbf{x}[t]\|_2$, $L_i = c_2 \big[\bar{L}_i \log \bar{L}_i\big]$ and $w_i = c_3 d_{\mathrm{out},i}\big[\bar{w}_i \log \bar{w}_i\big]$ for $i = 1,2,3$. 
Without loss of generality, we assume $L_1$, $L_2$ and $L_3$ grow at the same rate, then $L^\ast$ can be treated as a multiple of $L_1$, $L_2$ or $L_3$. Hence, the error can be approximated by 
	\begin{align*}
		&\underset{L_1,w_1}{\min} \Big\{2c_4 D_1 T^{-\nu_0} \log(T) L_1^2w_1^2 
		+ {3c_1^2}{r_1 r_2} d_0 M_2^2 M_3^2 (\bar{L}_1 \bar{w}_1)^{-4 \gamma^*} \Big\} \\
		+& \underset{L_2,w_2}{\min} \Big\{2c_4 D_2 T^{-\nu_0} \log(T) L_2^2 w_2^2 
		+ {3c_1^2}{r_2} d_0 M_3^2 (\bar{L}_2 \bar{w}_2)^{-4 \gamma^*} \Big\} \\
		+ &\underset{L_3,w_3}{\min} \Big\{ 2c_4 T^{-\nu_0} \log(T) L_3^2 w_3^2 
		+ {3c_1^2} d_0(\bar{L}_3 \bar{w}_3)^{-4 \gamma^*} \Big\} 
		+ 2{C_\theta} T^{\nu_0 -1}.
	\end{align*}
	It suffices to minimize $A_i T^{-\nu_0} \log(T) L_i^2 w_i^2 + B_i (\bar{L}_i \bar{w}_i)^{-4 \gamma^*}$ for some $A_i$ and $B_i$ independent of $L_i$ and $w_i$. For this term, we have
	\begin{align*}
		&A_i T^{-\nu_0} \log(T) L_i^2 w_i^2 + B_i (\bar{L}_i \bar{w}_i)^{-4 \gamma^*} \\
		\leq& A_i T^{-\nu_0} \log(T)(\bar{L}_i \bar{w}_i)^2 \log^2(\bar{L}_i)\log^2(\bar{w}_i) 
		+ B_i (\bar{L}_i \bar{w}_i)^{-4 \gamma^*} \\
		\leq& A_i T^{-\nu_0} \log(T)(\bar{L}_i \bar{w}_i)^2 \left(\tfrac{\log(\bar{L}_i)+\log(\bar{w}_i)}{2}\right)^4  
		+ B_i (\bar{L}_i \bar{w}_i)^{-4 \gamma^*} \\
		\leq& \tfrac{1}{16} A_i T^{-\nu_0} \log(T)(\bar{L}_i \bar{w}_i)^2 \log^4(\bar{L}_i\bar{w}_i)  
		+ B_i (\bar{L}_i \bar{w}_i)^{-4 \gamma^*} \coloneqq h(\bar{L}_i \bar{w}_i).
	\end{align*}
	
	
	Choosing $\bar{L}_i \bar{w}_i \asymp \left(\frac{B_i T^{\nu_0}}{A_i \log^{5}(T)}\right)^{\frac{1}{4\gamma^*+2}}$, we obtain a value close to the minimum of $h(\bar{L}_i \bar{w}_i)$, which is of order $B_i^{\frac{1}{2\gamma^*+1}} \big(\tfrac{A_i \log^5(T)}{T^{\nu_0}}\big)^{\frac{2\gamma^*}{2\gamma^*+1}}$.	
	Therefore, if we choose 
	\[\bar{L}_1 \bar{w}_1 \asymp \left(\tfrac{{r_1 r_2} d_0 M_2^2 M_3^2 T^{\nu_0}}{D_1 \log^{5}(T)}\right)^{\frac{1}{4\gamma^*+2}}, 
	\bar{L}_2 \bar{w}_2 \asymp \left(\tfrac{{r_2}d_0 M_3^2 T^{\nu_0}}{D_2 \log^{5}(T)}\right)^{\frac{1}{4\gamma^*+2}}, 
	\quad\text{and}\quad \bar{L}_3 \bar{w}_3 \asymp \left(\tfrac{d_0T^{\nu_0}}{\log^{5}(T)}\right)^{\frac{1}{4\gamma^*+2}},\]
	then, there exists a constant $c_5>0$ such that 
	\begin{align*} 
		\mathbb{E}\left[ 
		\mathcal{R}(\widehat{m}_T) - \mathcal{R}(m_0)\right]
		&\leq c_5 {C_D}^{\frac{1}{2\gamma^*+1}} \cdot \left(\frac{ \log^5(T)}{T^{\nu_0}}\right)^{\frac{2\gamma^*}{2\gamma^*+1}}
		+ 2{C_\theta} T^{\nu_0 -1},
	\end{align*} 
	where {$C_D$ is defined in Theorem \ref{thm:1}.
 Choosing $\nu_0=(2\gamma^\ast+1)/(4\gamma^\ast+1)$, we get the bound below,
\begin{equation*}
\mathbb{E}\big[\mathcal{R}(\widehat m_T)-\mathcal{R}(m_0)\big]\leq\left[c_5C_D^{\frac{1}{2\gamma^\ast+1}}(\log T)^{\frac{10\gamma^\ast}{2\gamma^\ast+1}}+2C_\theta\right]T^{-\frac{2\gamma^\ast}{4\gamma^\ast+1}}.
\end{equation*}
The proof is complete.
}
\end{proof}

\section{Auxiliary Lemmas}\label{app:sec-lemma}
\begin{lemma}[Upper bound on deep ReLU network approximation error for $\mathcal{H}(d, l, \mathcal{P})$]  \label{lmma:approximation}
	Suppose that Assumption \ref{assump:Boundness} holds. Let $f=(f_1,\ldots,f_{d_\mathrm{out}}){\in \mathbb{R}^{d_\mathrm{out}}}$. There exist universal constants $ c_1, c_2, c_3 >0 $ that depend only on $ l $, $ \sup_{(p, \tilde{d}) \in \mathcal{P}} \max\{p, {\tilde{d}}\} $, and the constant $ C $ from Definition \ref{def:HCM} of the manuscript, such that for any $ \bar{L}, \bar{w} \geq 3 $,
	\[
	\sup_{\substack{f_i \in \mathcal{H}(d, l, \mathcal{P})\\
        i \in [d_{\mathrm{out}}]}} 
    \inf_{\widehat{f} \in \mathcal{G}(L,d,d_{\mathrm{out}},{w})} 
    \|\widehat{f} - f\|_{\infty} 
    \leq c_1  (\bar{L}\bar{w})^{-2 \gamma^*},
	\]
	where $L = c_2 \lceil \bar{L} \log \bar{L} \rceil, \; {w} = c_3 d_{\mathrm{out}} \lceil\bar{w} \log \bar{w}\rceil$, 
	and $\gamma^*=\min_{(p,\tilde d)\in\mathcal P}p/\tilde d$.

\end{lemma}

\begin{proof}[Proof of Lemma~\ref{lmma:approximation}]
This result follows directly from the proof of Proposition 3.4 in \citet{fan-2024}, after extending their arguments from scalar outputs to $d_{\mathrm{out}}$-dimensional vector outputs. It is worth noting that \citet{fan-2024} employs a truncation level $M$ in the construction of deep ReLU networks. Under Assumptions \ref{assump:Boundness}(i) and (iii), the support of $\mathbf{x}[t]$ is bounded, and for each of the three deep ReLU networks $g_{j,\mathrm{NN}}$ with $j \in [3]$, the product $\prod_{l=1}^{L_j+1}  U_{j,l} \leq M_j$ is bounded by $M_j$. This boundedness plays an analogous role to the truncation parameter used in the original proof.

For each $i\in[d_\mathrm{out}]$, let $\widehat{f}_i\in \mathcal{G}(L,d,1,w)$ be a deep ReLU network of the form $\widehat{f}_i = \varphi_{i,L+1} \circ \sigma \circ \varphi_{i,L} \circ \cdots \circ \varphi_{i,1}$, where $\varphi_{i,l}(\mathbf{x}) =  \mathbf{W}_{i,l} \mathbf{x} + \mathbf{b}_{i,l}$ is an affine transformation with the weight matrix $\mathbf{W}_{i,l} \in \mathbb{R}^{w_l \times w_{l-1}}$ and bias vector ${\mathbf{b}_{i,l}} \in \mathbb{R}^{w_l}$ for $l \in [L+1]$. The hidden-layer widths are set as $w_1=\ldots=w_L=w$, and the output layer has dimension $w_{L+1}=1$.   
Proposition 3.4 in \citet{fan-2024} provides a general bound for $\mathcal{H}(d, l, \mathcal{P})$ approximated by scalar-output deep ReLU network class $\mathcal{G}(L,d,1,w)$. By applying it coordinate-wise to $f=(f_1,\ldots,f_{d_\mathrm{out}}){\in \mathbb{R}^{d_\mathrm{out}}}$,
we obtain $d_\mathrm{out}$ approximation error bounds, each sharing the same rate $c_1  (\bar{L}\bar{w})^{-2 \gamma^*}$. The same error rate across coordinates follows from the fact that all components share the same smoothness parameter $\gamma^*$, the same network parameters of depth $L$ width $w$, and the same constant $c_1$. Specifically, for each output coordinate $i=1,\ldots,d_\mathrm{out}$, we have  
\begin{align*} 	
\sup_{\substack{f_i \in \mathcal{H}(d, l, \mathcal{P})}} 
    \inf_{\widehat{f}_i \in \mathcal{G}(L,d,1,{w})} 
\|f_i - \widehat{f}_i\|_{\infty} 
&\leq c_1  (\bar{L}\bar{w})^{-2 \gamma^*},  
\end{align*}
where the depth and width are defined as $L = c_2 \lceil \bar{L} \log \bar{L} \rceil$ and ${w} = c_3  \lceil\bar{w} \log \bar{w}\rceil$, with positive constants $ c_1, c_2, c_3$ depending only on $ l $, $\sup_{(p, \tilde{d}) \in \mathcal{P}} \max\{p, {\tilde{d}}\}$, and the constant $C$ from Definition \ref{def:HCM} of the manuscript.

Define $\epsilon_i = \|f_i - \widehat{f}_i\|_{\infty}$ for $i \in [d_{\mathrm{out}}]$ and let $\boldsymbol{\epsilon} = (\epsilon_1, \epsilon_2, \ldots, \epsilon_{d_{\mathrm{out}}})$.
Taking maximum over the family of coordinate-wise bounds derived above, we obtain that for every collection of target functions $f_i \in \mathcal{H}(d, l, \mathcal{P})$, there exist approximants $\widehat{f}_i \in \mathcal{G}(L,d,1,w)$ such that 
\begin{align*}
\max\{\epsilon_1, \epsilon_2, \ldots, \epsilon_{d_{\mathrm{out}}}\} 
= \|\bm{\epsilon}\|_{\infty}
\leq c_1 (\bar{L}\bar{w})^{-2 \gamma^*}.
\end{align*}
Then, for any $f=(f_1,\ldots,f_{d_\mathrm{out}}){\in \mathbb{R}^{d_\mathrm{out}}}$ with $f_i \in \mathcal{H}(d, l, \mathcal{P})$, there exists a larger deep ReLU network $\widehat{f}= \varphi_{L+1} \circ \sigma \circ \varphi_{L} \circ \cdots \circ \varphi_{1}$ with $\varphi_{l}(\mathbf{x}) =  \mathbf{W}_{l} \mathbf{x} + \mathbf{b}_{l}$, such that
\begin{align*}
\|f - \widehat{f}\|_{\infty}
= \|\bm{\epsilon}\|_{\infty}
\leq c_1  (\bar{L}\bar{w})^{-2 \gamma^*},
\end{align*}
where
$$
\mathbf{W}_{l} = 
    \begin{pmatrix}
        \mathbf{W}_{1,l} & &  \\
        & \ddots &  \\
        & & \mathbf{W}_{{d_{\mathrm{out}}},l}
    \end{pmatrix} 
\in \mathbb{R}^{{d_{\mathrm{out}}} w_l \times {d_{\mathrm{out}}} w_{l-1}}
\quad \text{and} \quad
\mathbf{b}_{l} = 
    \begin{pmatrix}
        \mathbf{b}_{1,l}   \\
        \vdots  \\
        \mathbf{b}_{{d_{\mathrm{out}}},l}
    \end{pmatrix} 
\in \mathbb{R}^{{d_{\mathrm{out}}} w_l }.
$$
This completes the proof.
\end{proof}

\begin{lemma}\label{lmma:generalization_error_tail_bound}
	Suppose that Assumptions \ref{assump:Boundness} and \ref{assump:Weakly dependence} hold. Then, for any $ {T \in \mathbb{N}_+,}$  $\epsilon > 0 $ and $ \nu_0 \in {(0,1]} $, we have
	\begin{align*}
		&\mathbb{P}\left( 
		\sup_{m \in \mathcal{F}_{\mathrm{DSNN}}}
		\left| \mathcal{R}(m) - {\mathcal{R}}_T(m) \right| > \epsilon \right)
		\leq 2\mathcal{N}_{\infty} 
		\left( 
		\mathcal{F}_{\mathrm{DSNN}}, \frac{\epsilon}{4G} \right) 
		\exp\left( \frac{-T^{\nu_0} \epsilon + {C_\theta} T^{2 \nu_0 -1}}{2{C^2}} \right),
	\end{align*}
	where 
        $G \coloneqq \underset{\substack{m_1,m_2 \in \mathcal{F}_{\mathrm{DSNN}} \\ 
			m_1 \neq m_2}}{\sup} 
	\underset{\mathbf{x},\mathbf{y}}{\sup}
	\frac{|\ell(m_1(\mathbf{x}), \mathbf{y}) - \ell(m_2(\mathbf{x}), \mathbf{y})|}{\|m_1 - m_2\|_\infty} < \infty$ 
	with $\ell(\mathbf{a},\mathbf{b})=\|\mathbf{a}-\mathbf{b}\|_2^2$,
    {$C$ is a Lipschitz bound for $\ell$ on the data support},
    and $C_\theta$ is defined in Proposition \ref{prop:generalization_our}.
\end{lemma}

\begin{proof}[Proof of Lemma~\ref{lmma:generalization_error_tail_bound}]
	To establish this lemma, we first derive a concentration inequality for a fixed function $f\in\mathcal{F}_{\mathrm{DSNN}}$, i.e., a bound for $\mathbb{P}\left(\left|\mathcal{R}(f)-\mathcal{R}_T(f)\right| \geq \epsilon\right)$. We then extend this result to a uniform bound over the entire function class, $\mathbb{P}\left(\sup_{f \in \mathcal{F}_{\mathrm{DSNN}}}\left|\mathcal{R}(f)-\mathcal{R}_T(f)\right| \geq \epsilon \right)$, using a standard covering number argument.
	
	{We take the maximum lag order $K=1$ for simplicity of notation.}
	For a fixed function $f\in \mathcal{F}_{\mathrm{DSNN}}$, any $\lambda>0$ and $\epsilon>0$, we have
	\begin{align*}
		\mathbb{P}\left(
		\left|\mathcal{R}(f)-\mathcal{R}_T(f)\right| \geq \epsilon
		\right) 
		\leq& \mathbb{P}\left( \lambda \left|
		\sum_{t=1}^T \left\{
		\mathbb{E}[\ell(f(\mathbf{x}[t]), \mathbf{y}[t])]-\ell(f(\mathbf{x}[t]), \mathbf{y}[t]) \right\}\right|
		\geq \epsilon  T \lambda
		\right) \\
		\leq& \mathbb{P}\left( \exp \left(\lambda \left|
		\sum_{t=1}^T \left\{
		\mathbb{E}[\ell(f(\mathbf{x}[t]), \mathbf{y}[t])]-\ell(f(\mathbf{x}[t]), \mathbf{y}[t]) \right\}\right| \right)
		\geq e^{\epsilon  T \lambda}
		\right) \\
		\overset{(i)}{\leq} &
		{e^{-\epsilon  T \lambda}}
		{\mathbb{E}\left[\exp \left(\lambda \left|
			\sum_{t=1}^T \left\{
			\mathbb{E}[\ell(f(\mathbf{x}[t]), \mathbf{y}[t])]-\ell(f(\mathbf{x}[t]), \mathbf{y}[t]) \right\}\right| \right)\right]} \\
		\overset{(ii)}{\leq} &
		2 \exp \left( \frac{ T {C^2}{C_\theta}\lambda^2}{2} -\epsilon  T \lambda \right),
	\end{align*}
	where $C_\theta$ is defined in Proposition \ref{prop:generalization_our}. 
{{{{Inequality (i) applies Markov's inequality, and Inequality (ii) uses Lemma {{\ref{lmma:Hoeffding_ineq}}} since {{$G<\infty$}} due to the boundedness of loss function {{$\ell$}} under Assumptions {{\ref{assump:Boundness}}}(i) and (iii).}}}}

	Note that the above inequality holds for any $\lambda>0$. To facilitate the minimization with respect to $\epsilon$ after the corresponding integration in \eqref{eq:integration}, we choose $\lambda=T^{\nu_0-1}/C^2$ for any $\nu_0 \in (0,1]$; see also the proof of Theorem 4.1 of \cite{kengne-2023}. Hence, we have 
	\begin{align}\label{ineq:generalization_single_point}
		\mathbb{P}\left(
		\left|\mathcal{R}(f)-\mathcal{R}_T(f)\right| \geq \epsilon
		\right) \leq 
		2\exp\left( \frac{-T^{\nu_0}\epsilon
        + C_\theta T^{2\nu_0-1}/2}{C^2} \right).
	\end{align}

	Next we consider $\mathbb{P}\left(\sup_{f \in \mathcal{F}_{\mathrm{DSNN}}} 
	\left|\mathcal{R}(f)-\mathcal{R}_T(f)\right| \geq \epsilon \right)$. Let $n=\mathcal{N}_{\infty}\left(\mathcal{F}_{\mathrm{DSNN}},\tfrac{\epsilon}{4G}\right)$ and $\{f_j\}_{j=1}^n$ be the $\tfrac{\epsilon}{4G}$-covering net for $\mathcal{F}_{\mathrm{DSNN}}$ under the $\ell_{\infty}$-norm. 
	We denote $\mathcal{F}_j = \{f \in \mathcal{F}_{\mathrm{DSNN}}: \|f - f_j\|_{\infty} \leq \frac{\epsilon}{4G}\}$ for $j \in [n]$, which leads to $\mathcal{F}_{\mathrm{DSNN}} \subseteq \bigcup_{j=1}^n \mathcal{F}_j$.
	Then, we have
	\begin{align*}
		|\mathcal{R}(f) - \mathcal{R}(f_j)| 
		=& \left| \mathbb{E}[\ell(f(\mathbf{x}[1]), \mathbf{y}[1])] - \mathbb{E}[\ell(f_j(\mathbf{x}[1]), \mathbf{y}[1])] \right| \\
		\leq& \|\ell(f(\mathbf{x}[1]), \mathbf{y}[1]) - \ell(f_j(\mathbf{x}[1]), \mathbf{y}[1])\|_{\infty}, \quad\text{and} \\ 	
		|\mathcal{R}_T(f) - \mathcal{R}_T(f_j) |
		=& \left|\frac{1}{T} \sum_{{t}=1}^T [\ell(f(\mathbf{x}[t]), \mathbf{y}[t]) - \ell(f_j(\mathbf{x}[t]), \mathbf{y}[t])] \right| \\
		\leq& \|\ell(f(\mathbf{x}[1]), \mathbf{y}[1]) - \ell(f_j(\mathbf{x}[1]), \mathbf{y}[1])\|_{\infty}.
	\end{align*}
	Hence, we have
	\begin{align*}
		\sup_{f \in \mathcal{F}_j} 
		\left| \underbrace{\mathcal{R}(f) - \mathcal{R}_T(f)}_{L(f)} 
		- \underbrace{[\mathcal{R}(f_j) - \mathcal{R}_T(f_j)]}_{L(f_j)} \right| 
		\leq& 2 \sup_{f \in \mathcal{F}_j} \|\ell(f(\mathbf{x}[1]), \mathbf{y}[1]) - \ell(f_j(\mathbf{x}[1]), \mathbf{y}[1])\|_{\infty} \\
		\leq& 2G \sup_{f \in \mathcal{F}_j} \|f - f_j\|_{\infty} 
		\leq \frac{\epsilon}{2}.
	\end{align*}
	As a result, for any $f \in \mathcal{F}_j$, we have
	\begin{align*}
		\{|L(f)| \geq \epsilon\} \subset \left\{|L(f_j)| \geq \frac{\epsilon}{2}\right\}
		\Rightarrow \mathbb{P}(|L(f)| \geq \epsilon) \leq \mathbb{P}\left(|L(f_j)| \geq \frac{\epsilon}{2}\right).
	\end{align*}
	This together with \eqref{ineq:generalization_single_point}, implies that
	\begin{align*}
		\mathbb{P}\left(\sup_{f \in \mathcal{F}_j} 
		\left|{\mathcal R(f)-\mathcal R_T(f)}\right| \geq \epsilon \right) 
		\leq \mathbb{P}\left(|\mathcal{R}(f_j) - \mathcal{R}_T(f_j) |\geq \tfrac{\epsilon}{2}\right) 
		\leq 2\exp\left( \frac{-T^{\nu_0} \epsilon + {C_\theta} T^{2 \nu_0 -1}}{2{C^2}} \right).
	\end{align*}
	As $\mathcal{F}_{\mathrm{DSNN}} = \bigcup_{j=1}^n \mathcal{F}_j$, we have
	\begin{align*}
		\mathbb{P}\left(\sup_{f \in \mathcal{F}_{\mathrm{DSNN}}} 
		\left|{\mathcal R(f)-\mathcal R_T(f)}\right| \geq \epsilon \right) 
		\leq& \sum_{j=1}^n \mathbb{P}\left(\sup_{f \in \mathcal{F}_j} 
		\left|{\mathcal R(f)-\mathcal R_T(f)}\right| \geq \epsilon \right) \\
		\leq& 2\mathcal{N}_{\infty}\left(\mathcal{F}_{\mathrm{DSNN}}, \frac{\epsilon}{4G}\right) 
		\exp\left( \frac{-T^{\nu_0} \epsilon + {C_\theta} T^{2 \nu_0 -1}}{2{C^2}} \right),
	\end{align*}
	which completes the proof.
\end{proof}

\begin{lemma} \label{lmma:covering_our}
	(Upper bound on the covering number of the DSNN class). 
	Suppose that Assumptions \ref{assump:Boundness}(i) and (iii) hold.
	There exists a constant $ c_7 >0$ such that
	\begin{align*}
		\log \left( \mathcal{N}_2(\mathcal{F}_{\mathrm{DSNN}}, \epsilon) \right) 
		\leq c_7 ({D_1L_1w_{1}^2 + D_2L_2  w_{2}^2 + L_3 w_{3}^2}) \log \left(\frac{ L^\ast w^\ast \bar{U}^{L^\ast} \bar{U}_x}{\epsilon}\right),
	\end{align*}
	where $L^\ast=\sum_{j=1}^3 L_j, w^\ast =\max\{d_0 + d_1 + Qd_2 + w_1 D_1, \; d_0 + d_1 + w_{2} D_2, \;w_{3}\}$, $\bar{U}=\max_{j,l}\{U_{j,l},1\}$ with $U_{j,l}$ defined in Assumption \ref{assump:Boundness}(iii), and $\bar{U}_x \coloneqq \max\{\sup_{t}\|\mathbf{x}[t]\|_2,1\} < \infty$.
\end{lemma}

\begin{proof}[Proof of Lemma~\ref{lmma:covering_our}]
	By Lemma \ref{lmma:3NN}, 
	for any $m \in \mathcal{F}_{\mathrm{DSNN}}$, it holds that $m \in \mathcal{G}(L^\ast, d_{\mathrm{in}}, d_{\mathrm{out}}, w^\ast)$. 
	{{Then, for any {$f \in \mathcal{F}_{\mathrm{DSNN}}$}, there exists a {$f^{\prime} \in \mathcal{F}_{\mathrm{DSNN}}$} such that}}
	\begin{align*}
		\sup_{\mathbf{x}[t]}
		\| f(\mathbf{x}[t]) - f^{\prime}(\mathbf{x}[t])\|_2
		\leq L_w \sum_{l=1}^{L^*+1}\| \mathbf{W}_l- \mathbf{W}_l^{\prime} \|_F + L_b \sum_{l=1}^{L^*+1} \| \mathbf{b}_l- \mathbf{b}_l^{\prime} \|_2,
	\end{align*}
	for some constants $L_w$ and $L_b$ (will be specified below) depending on $L^*, U_{j,l}$, and $U_x \coloneqq \sup_{t}\|\mathbf{x}[t]\|_2 < \infty$ under Assumption \ref{assump:Boundness}(i). 
	To cover the DSNN function class $\mathcal{F}_{\mathrm{DSNN}}$, it suffices to cover the weight matrices $\mathbf{W}_l$'s and bias vectors $\mathbf{b}_l$'s.
	
	For $ l \in  [L^\ast+1]$, denote by $\mathbf{y}_l$ and $\mathbf{y}_l^{\prime}$ the inputs of the $l$-th layer corresponding to $f$ and $f^{\prime}$ respectively, with $\mathbf{y}_1 = \mathbf{y}_1^{\prime} = \mathbf{x}[t]$ satisfying $\|\mathbf{y}_1\|_2 \leq U_x$. 
	Assumption \ref{assump:Boundness}(iii) implies that there exist constants $U_w, U_b >0$ such that $\|\mathbf{W}_{j,l}\|_2 \leq U_w$ and $\|\mathbf{b}_{j,l}\|_2 \leq U_b$.
	For $l > 1$, then we have 
	\begin{align*}
		\|\mathbf{y}_l\|_2 &= \| \sigma(\mathbf{W}_{l-1} \mathbf{y}_{l-1} + \mathbf{b}_{l-1})\|_2 
		\overset{(i)}{\leq} \| \mathbf{W}_{l-1} \mathbf{y}_{l-1} + \mathbf{b}_{l-1}\|_2 \\
		&\overset{(ii)}{\leq} U_w \|\mathbf{y}_{l-1}\|_2 +  U_b 
		\leq \dots 
		\leq U_w^{l-1} \|\mathbf{y}_1\|_2 + \frac{U_b(U_w^{l-1} -1)}{U_w -1} \\
		&\leq U_w^{l-1} U_x + \frac{U_b(U_w^{l-1} -1)}{U_w -1}
		\coloneqq A_l.
	\end{align*}
	Inequality (i) uses the fact that the activation function $\sigma$ is Lipschitz with constant one.  
	Inequality (ii) follows from $\|\mathbf{a}+\mathbf{b}\|_2 \leq \|\mathbf{a}\|_2 + \|\mathbf{b}\|_2$ and $\|\mathbf{a}\mathbf{b}\|_2 \leq \|\mathbf{a}\|_2 \|\mathbf{b}\|_2$.  
	Then we have 
	\begin{align*}
		\| \mathbf{{y}}_{{L^\ast}+1} - \mathbf{y}^{\prime}_{{L^\ast}+1} \|_2 &= \| \sigma(\mathbf{W}_{L^\ast} \mathbf{y}_{L^\ast} + \mathbf{b}_{L^\ast}) - \sigma(\mathbf{W}^{\prime}_{L^\ast} \mathbf{y}^{\prime}_{L^\ast} + \mathbf{b}^{\prime}_{L^\ast})\|_2 \\
		&\leq  \| \mathbf{W}_{L^\ast} \mathbf{y}_{L^\ast} - \mathbf{W}_{L^\ast} \mathbf{y}^{\prime}_{L^\ast} + \mathbf{W}_{L^\ast} \mathbf{y}^{\prime}_{L^\ast} - \mathbf{W}_{L^\ast}^{\prime} \mathbf{y}^{\prime}_{L^\ast} + \mathbf{b}_{L^\ast} - \mathbf{b}_{L^\ast}^{\prime}\|_2 \\
		&\leq   U_w \| \mathbf{y}_{L^\ast} - \mathbf{y}^{\prime}_{L^\ast} \|_2 
		+ \|\mathbf{y}_{L^\ast}^{\prime}\|_2 \| \mathbf{W}_{L^\ast} - \mathbf{W}_{L^\ast}^{\prime}\|_F 
		+ \| \mathbf{b}_{L^\ast} - \mathbf{b}_{L^\ast}^{\prime}\|_2  \\
		&\leq U_w \| \mathbf{y}_{L^\ast} - \mathbf{y}^{\prime}_{L^\ast} \|_2 
		+  A_{L^\ast} \| \mathbf{W}_{L^\ast} - \mathbf{W}_{L^\ast}^{\prime}\|_F 
		+  \| \mathbf{b}_{L^\ast} - \mathbf{b}_{L^\ast}^{\prime}\|_2 \leq \cdots\\
		&\leq U_w^{{L^\ast}} \| \mathbf{y}_1 - \mathbf{y}^{\prime}_1 \|_2
		+  \sum_{l=1}^{{L^\ast}} {U_w^{{L^\ast}-l}} {A_{l}} \| \mathbf{W}_l - \mathbf{W}_l^{\prime}\|_F 
		+  \sum_{l=1}^{{L^\ast}} {U_w^{{L^\ast}-l}} \| \mathbf{b}_l - \mathbf{b}_l^{\prime}\|_2 \\
		&=  \sum_{l=1}^{{L^\ast}} {U_w^{{L^\ast}-l}} {A_{l}} \| \mathbf{W}_l - \mathbf{W}_l^{\prime}\|_F 
		+  \sum_{l=1}^{{L^\ast}} {U_w^{{L^\ast}-l}} \| \mathbf{b}_l - \mathbf{b}_l^{\prime}\|_2.
	\end{align*}
	This together with the final linear mapping, implies that
	\begin{align*}
		\| f(\mathbf{x}[t]) - f^{\prime}(\mathbf{x}[t])\|_2 =&
		\| (\mathbf{W}_{L^\ast+1} \mathbf{y}_{L^\ast+1} + \mathbf{b}_{L^\ast+1}) - (\mathbf{W}^{\prime}_{L^\ast+1} \mathbf{y}^{\prime}_{L^\ast+1} + \mathbf{b}^{\prime}_{L^\ast+1})\|_2 \\
		\leq&  \sum_{l=1}^{{L^\ast+1}} {U_w^{{L^\ast+1}-l}} {A_{l}} \| \mathbf{W}_l - \mathbf{W}_l^{\prime}\|_F 
		+  \sum_{l=1}^{{L^\ast+1}} {U_w^{{L^\ast+1}-l}} \| \mathbf{b}_l - \mathbf{b}_l^{\prime}\|_2 \\
		\leq & \sum_{l=1}^{{L^\ast+1}} \left(U_x  U_w^{{L^\ast}} +  U_b \frac{U_w^{{L^\ast}} - U_w^{{L^\ast+1}-l}}{U_w-1} \right) \| \mathbf{W}_l - \mathbf{W}_l^{\prime}\|_F 
		+  \sum_{l=1}^{{L^\ast+1}}  {U_w^{{L^\ast+1}-l}} \| \mathbf{b}_l - \mathbf{b}_l^{\prime}\|_2 \\
		\leq& L_w \sum_{l=1}^{{L^\ast+1}} \| \mathbf{W}_l - \mathbf{W}_l^{\prime}\|_F 
		+ L_b \sum_{l=1}^{{L^\ast+1}} \| \mathbf{b}_l- \mathbf{b}_l^{\prime} \|_2,
	\end{align*}
	where $L_w = U_x  U_w^{L^\ast} + { U_b } \frac{ U_w^{L^\ast} - 1}{U_w-1}$
	and $L_b =  \max\{U_w^{L^\ast}, 1\}$.

    {{Then, it suffices to cover the space {$\mathbb{S}_{W_l} \coloneqq \{\mathbf{W}_l \in \mathbb{R}^{w_l\times w_{l-1}}:\|\mathbf{W}_l\|_F \leq \frac{\epsilon}{2 (L^\ast +1) L_w}\}$} and {$\mathbb{S}_{b_l} \coloneqq\{\mathbf{b}_l \in \mathbb{R}^{w}_l:\|\mathbf{b}_l\|_2\leq \frac{\epsilon}{2 (L^\ast +1) L_b}\}$}.}}
	It follows that
	\begin{align*}
		\mathcal{N}_2(\mathcal{F}_{\mathrm{DSNN}}, \epsilon)
		&\leq \prod_{l=1}^{L^\ast+1} \mathcal{N}_F \left(\mathbf{W}_l, \frac{\epsilon}{2(L^\ast+1)L_w} \right)
		\cdot \prod_{l=1}^{L^\ast+1} \mathcal{N}_2 \left(\mathbf{b}_l, \frac{\epsilon}{2(L^\ast+1)L_b} \right).
	\end{align*}
	
	To calculate the first covering number at the RHS, we apply Lemma 8 in \cite{chen-2019} and employ the block diagonal structure of the weight matrices $\mathbf{W}_l^{\ast}$ at \eqref{eq:weight_bias} in Lemma \ref{lmma:3NN}, together with $L_w = U_x  U_w^{L^\ast} + { U_b } \frac{ U_w^{L^\ast} - 1}{U_w-1} \lesssim U_x U^{L^\ast}$ as $U=\max\{U_w, U_b\}$, leading to 
	\begin{align*}
		\prod_{l=1}^{L^\ast+1} \mathcal{N}_F \left(\mathbf{W}_l^\ast, \frac{\epsilon}{2(L^\ast+1)L_w} \right) 
		&\leq \left(1+ \frac{2(L^\ast+1)L_w w^\ast}{\epsilon}\right)^{D_1L_1w_{1}^2 + D_2 L_2 w_{2}^2 + L_3 w_{3}^2} \\
		&\leq \left(\frac{ c_8 L^\ast w^\ast  U^{L^\ast}U_x}{\epsilon}\right)^{D_1L_1w_{1}^2 +D_2 L_2 w_{2}^2 + L_3 w_{3}^2},
	\end{align*}
	where $c_8>0$ is the constant and $U \coloneqq\underset{j,l}{\max}\{U_{j,l}\}$.
	For the second covering number at the RHS, we apply Exercise 5.8 in \cite{wainwright-2019} with $L_b =  \max\{U_w^{L^\ast}, 1\}$ and obtain
	\begin{align*}
		\prod_{l=1}^{L^\ast+1} \mathcal{N}_2 \left(\mathbf{b}_l^\ast, \frac{\epsilon}{2(L^\ast+1)L_b} \right)
		&\leq \left(1+ \frac{2(L^\ast+1)L_b }{\epsilon}\right)^{D_1L_1w_{1} + D_2 L_2w_{2} + L_3 w_{3}} \\
		&\leq \left(1+\frac{ 2(L^\ast+1) \max\{ U^{L^\ast},1\} }{\epsilon}\right)^{D_1L_1w_{1} + D_2 L_2 w_{2} + L_3 w_{3}}
		.
	\end{align*}
	These imply that
	\begin{align*}
		\log(\mathcal{N}_2(\mathcal{F}_{\mathrm{DSNN}}, \epsilon) )
		\leq &({D_1L_1w_{1}^2 + D_2L_2  w_{2}^2 + L_3 w_{3}^2}) \log \left(\frac{c_8 L^\ast w^\ast U^{L^\ast} U_x}{\epsilon}\right) \\
		& + ({D_1L_1w_{1} + D_2 L_2 w_{2} + L_3 w_{3}}) \log \left(1+\frac{ 2(L^\ast+1) \bar{U}^{L^\ast} }{\epsilon}\right)\\
		\leq &c_7({D_1L_1w_{1}^2 + D_2L_2  w_{2}^2 + L_3 w_{3}^2}) \log \left(\frac{L^\ast w^\ast \bar{U}^{L^\ast} \bar{U}_x}{\epsilon}\right),
	\end{align*} 
	where $\bar{U}=\max\{U,1\}$,$\bar{U}_x=\max\{U_x,1\}$, and $c_7>0$ is a constant.
	This proof is complete.
\end{proof}


\begin{lemma}\label{lmma:3NN}
	Suppose that Assumption \ref{assump:Boundness}(iii) holds.
	For any $m \in \mathcal{F}_{\mathrm{DSNN}}$, it holds that $m \in \mathcal{G}(L^\ast, d_{\mathrm{in}}, d_{\mathrm{out}}, w^\ast)$, 
	where the depth $ L^\ast = \sum_{j=1}^3 L_j $, 
	the input dimension $d_{\mathrm{in}} = K D $ with $D = d_0 + d_1 + Q d_2 + S Q d_3$, the output dimension $d_{\mathrm{out}} = d_0$, and the width $w^\ast = \max\{d_0 + d_1 + Qd_2 + w_1 D_1, \; d_0 + d_1 + w_{2} D_2, \;w_{3}\}$ with $D_1 = Q d_3$ and $D_2 = d_2 + r_1 d_3$.
\end{lemma}

\begin{proof}[Proof of Lemma~\ref{lmma:3NN}]
	By definition of the DSNN function class $\mathcal{F}_{\mathrm{DSNN}}$, we have 
	\begin{align*} 
		m = g_{3,\mathrm{NN}} \circ &
		\begin{pmatrix}
			g_{2,\mathrm{NN}} \\
			\vdots \\
			g_{2,\mathrm{NN}} 
		\end{pmatrix} 
		\circ \mathcal{S}_2 \circ 
		\begin{pmatrix}
			g_{1,\mathrm{NN}} \\
			\vdots \\
			g_{1,\mathrm{NN}} 
		\end{pmatrix} 
		\circ \mathcal{S}_1
		\coloneqq g_{3,\mathrm{NN}} \ast g_{2,\mathrm{NN}} \ast g_{1,\mathrm{NN}}, \\
		&\quad \scriptscriptstyle D_2 \; \text{terms} \quad \quad \quad \quad \scriptscriptstyle D_1 \; \text{terms} \nonumber
	\end{align*}
	where $ g_{j,\mathrm{NN}} \in \mathcal{G}(L_j, d_{\mathrm{in},j}, d_{\mathrm{out},j}, w_j)$ for $j \in [3]$,
	$L_j,w_j \in \mathbb{N}_+$, the input dimensions are $(d_{\mathrm{in},1}, d_{\mathrm{in},2}, d_{\mathrm{in},3})=(S, Q, KD_3)$ with $D_3 = d_0 + d_1 + r_2 d_2 + r_1 r_2 d_3$ and the output dimensions are $(d_{\mathrm{out},1}, d_{\mathrm{out},2}, d_{\mathrm{out},3}) =(r_1, r_2, d_0)$.
	By Definition \ref{def:DNN}, $g_{j,\mathrm{NN}} = \varphi_{j,L_j+1} \circ \sigma_{L_j} \circ \varphi_{j,L_j} \circ \cdots \circ \sigma_{1} \circ \varphi_{j,1}$, where $\varphi_{j,l}(\mathbf{x}) = \mathbf{W}_{j,l} \mathbf{x} + \mathbf{b}_{j,l}$ is an affine transformation with the weight matrix $\mathbf{W}_{j,l} \in \mathbb{R}^{w_{j,l} \times w_{j,l-1}}$ and bias vector $\mathbf{b}_{j,l} \in \mathbb{R}^{w_{j,l}}$, and $\sigma_l$ is the ReLU activation function applied to each entry of a $w_l$-dimensional vector.
	By Assumption \ref{assump:Boundness}(iii), it holds that $\max\{\|\mathbf{W}_{j,l}\|_{2}, \|\mathbf{b}_{j,l}\|_{2}\} \leq U_{j,l}$ for $l \in [L_j]$, $j \in [3]$, and $\prod_{l=1}^{L_j+1}  U_{j,l} \leq M_j$ for some constants $M_j>0$.
	
	For simplicity of the notation, the data are organized into tensor or matrix forms according to their frequencies. The high-frequency predictors form a tensor $\cm{H} \in \mathbb{R}^{d_3 \times T \times Q \times S}$ with the first mode $(\cm{H}_1,\cm{H}_2,\dots,\cm{H}_{d_3})$, where $\cm{H}_i \coloneqq \{h_i[t,q,s], t \in [T], q \in [Q], s \in[S]\}$ for $i \in [d_3]$. The medium-frequency predictors are represented as a tensor $\cm{M} \in \mathbb{R}^{d_2 \times T \times Q}$ with the first mode $(\mathbf{M}_1,\mathbf{M}_2,\dots,\mathbf{M}_{d_2})$, where $\mathbf{M}_i \coloneqq \{m_i[t,q], t \in [T], q \in [Q]\}$ for $i \in [d_2]$. The low-frequency predictors and responses are expressed as matrices $\mathbf{Z} \coloneqq \{z_i[t], i \in [d_1], t \in [T]\} \in \mathbb{R}^{d_1 \times T}$ and $\mathbf{Y} \coloneqq \{y_i[t], i \in [d_0], t \in [T]\} \in \mathbb{R}^{d_0 \times T}$, respectively.
	We use $\mathcal{S}_2 \circ \mathcal{S}_1$ to partition $\cm{H}$ into $S Q  d_3$ low-frequency time series $\{\mathbf{h}_j \in \mathbb{R}^{T}\}_{j=1}^{S Q  d_3}$, and $\mathcal{S}_2$ to partition $\cm{M}$ into $Q d_2$ low-frequency time series $\{\mathbf{m}_j \in \mathbb{R}^{T}\}_{j=1}^{Q d_2}$. Then we can concatenate $\mathcal{S}_2 \circ \mathcal{S}_1 \circ \cm{H}$, $\mathcal{S}_2 \circ \cm{M}$, $\mathbf{Z}$ and $\mathbf{Y}$, yielding a total of $D = d_0 + d_1 + Q d_2 + S Q d_3$ low-frequency time series denoted by $\{\mathbf{x}^{\mathrm{Input}}_j \in \mathbb{R}^T\}_{j=1}^{D}$, and the resulting input matrix by $\mathbf{X}^{\mathrm{Input}} \in \mathbb{R}^{D \times T}$.
	
	To prove this lemma, we need to show there exists of a deep ReLU network $f \in \mathcal{G}(L^\ast, d_{\mathrm{in}}, d_{\mathrm{out}}, w^\ast)$ that generates the same output as $m$, with $d_{\mathrm{in}}=D$ and $d_{\mathrm{out}}=d_0$.
	Then by Definition \ref{def:DNN}, it suffices to show that:
	\begin{enumerate}
		\item[(i)] the depth $L^\ast$, determined by the number of ReLU activation functions, is well-defined;
		\item[(ii)] the weight matrix $\mathbf{W}^\ast_{l}$ and the bias vector $\mathbf{b}^\ast_{l}$ are specified for each layer;
		\item[(iii)] the upper bound $U_l^\ast$ is specified for $\max\{\|\mathbf{W}_{j,l}^\ast\|_2,\|\mathbf{b}_{j,l}^\ast\|_2\} < U_{j,l}$, and $\prod_{l=1}^{L^\ast+1}  U^\ast_{l} \leq M$ holds for some constant $M>0$. 
	\end{enumerate} 
	
	\textbf{The first stage}: the first $L_1$ layers of $f$
	
	The first stage handles the high-frequency variables, which can be written as follows:
	\begin{align*}
		{\mathbf{x}}_{l+1,t} 
		&= \sigma \circ\varphi_l(\mathbf{x}_{l,t}) 
		= \sigma (\widetilde{\mathbf{W}}_l \mathbf{x}_{l,t} + \widetilde{\mathbf{b}}_l)
		, \quad l \in [L_1],\; t \in [T] \quad \text{and} \\
		{\mathbf{x}}_{L_1+2,t} 
		&= \varphi_{L_1+1}(\mathbf{x}_{L_1+1,t}) 
		= \widetilde{\mathbf{W}}_{L_1+1} \mathbf{x}_{L_1+1,t} + \widetilde{\mathbf{b}}_{L_1+1}, \quad t \in [T],
	\end{align*}
	where 
	\begin{align*}
		\widetilde{\mathbf{W}}_l = &
		\begin{pmatrix}
			\mathbf{W}_{1,l} & & & \\
			& \ddots & & \\
			& & \mathbf{W}_{1,l} & \\
			& & & \mathbf{I}_{d_0 + d_1 + Q d_2}
		\end{pmatrix} 
		\quad \text{and} \quad
		\widetilde{\mathbf{b}}_l = 
		\begin{pmatrix}
			\mathbf{b}_{1,l} \\
			\vdots \\
			\mathbf{b}_{1,l} \\
			\mathbf{0}_{d_0 + d_1 + Q d_2}
		\end{pmatrix}, \quad l \in [L_1 + 1].\\
		&\text{(with } D_1 + 1 \text{ diagonal blocks)} \nonumber
	\end{align*}
	Here, $\mathbf{W}_{1,l} $ and $\mathbf{b}_{1,l}$ repeat $D_1$ times,
	$\sigma$ is the ReLU activation function, and
	$\mathbf{x}_{l,t}$ is the input of the $l$-th layer for each $t$ with $\mathbf{x}_{1,t} = \mathbf{X}^{\text{Input}}_{:,t} \in \mathbb{R}^D$.
	The output of the first stage is  
	$\mathbf{x}_{L_1+2,t} \in \mathbb{R}^{d_0 + d_1 + Q d_2 + r_1 D_1}, t \in [T]$.

	\textbf{The second stage: the next $L_2$ layers of ${f}$}
	
	The second stage handles the medium-frequency variables,
	which can be written as follows:
	\begin{align*}
		{\mathbf{x}}_{L_1+l+2, t} 
		&= \sigma \circ\varphi_{L_1+1+l}(\mathbf{x}_{L_1+l+1,t}) 
		= \sigma (\bar{\mathbf{W}}_{L_1+l} \mathbf{x}_{L_1+l+1,t} + \bar{\mathbf{b}}_{L_1+l}), \quad l \in [L_2],\; t \in [T]   \quad \text{and} \\
		\tilde{\mathbf{x}}_{L_1+L_2+3,t} 
		&= \varphi_{L_1+L_2+2}(\mathbf{x}_{L_1+L_2+2,t}) 
		= \bar{\mathbf{W}}_{L_1+L_2+1} \mathbf{x}_{L_1+L_2+2,t} + \bar{\mathbf{b}}_{L_1+L_2+1}, \quad t \in [T],
	\end{align*}
	where 
	\begin{align*}
		\bar{\mathbf{W}}_{L_1+l} = &
		\begin{pmatrix}
			\mathbf{W}_{2,l} & & & \\
			& \ddots & & \\
			& & \mathbf{W}_{2,l} & \\
			& & & \mathbf{I}_{d_0+d_1}
		\end{pmatrix} 
		\quad \text{and} \quad
		\bar{\mathbf{b}}_{L_1+l}  = 
		\begin{pmatrix}
			\mathbf{b}_{2,l} \\
			\vdots \\
			\mathbf{b}_{2,l} \\
			\mathbf{0}_{d_0+d_1}
		\end{pmatrix}, \quad l \in [L_2+1].\\
		&\text{(with } D_2 + 1 \text{ diagonal blocks)} \nonumber
	\end{align*}
	
	Here, $\mathbf{W}_{2,l} $ and $\mathbf{b}_{2,l}$ repeat $D_2$ times. 
	Moreover, we concatenate the lag terms of the output of $(L_1+L_2)$ layers as ${\mathbf{x}}_{L_1+L_2+3,t} \coloneqq (\tilde{\mathbf{x}}_{L_1+L_2+3,t}^{\prime},\cdots,\tilde{\mathbf{x}}_{L_1+L_2+3,t-K+1}^{\prime})^{\prime} \in \mathbb{R}^{KD_3}$ for $K+1\leq t\leq T$.

	\textbf{The third stage: the final $L_3$ layers of ${f}$}
	
	The last stage is used for prediction, and the final $L_3$ layers can be specified as follows:
	\begin{align*}
		{\mathbf{x}}_{L_1+L_2+l+3, t} &= \sigma \circ \varphi_{L_1 + L_2 + l + 2}(\mathbf{x}_{L_1+L_2+l+2, t}) = \sigma \circ \varphi_{3,l}({\mathbf{x}_{L_1+L_2+l+2, t}}), 
		\quad l \in [L_3], K+1\leq t\leq T, \\
		\mathbf{x}_{L_1+L_2+L_3+4, t} & = \varphi_{L_1 + L_2 + L_3 + 3}({\mathbf{x}}_{L_1+L_2+L_3+3, t})=   \varphi_{3,L_3+1}({\mathbf{x}_{L_1+L_2+L_3+3, t}}), 
		\quad K+1\leq t\leq T,
	\end{align*}
	where $\varphi_{3,l}(\mathbf{x})= \mathbf{W}_{3,l} \mathbf{x} + \mathbf{b}_{3,l}$ with $l\in [L_3+1]$,
	and $K$ is the maximum lag order. 
	This leads to the same output as the function $m$.
	
	In conclusion, based on the above three stages, we construct a deep ReLU network $f \in \mathcal{G}(L^\ast, d_{\mathrm{in}}, d_{\mathrm{out}}, w^\ast)$, where $L^\ast = {\sum_{j=1}^3 L_j}$, and the maximum width of $f$ is given by:
	\[
	w^\ast =\max\{d_0 + d_1 + Qd_2 + w_1 D_1, \; d_0 + d_1 + w_{2} D_2, \;w_{3}\}.
	\] 
	Moreover, the weight matrix $\mathbf{W}^\ast_{l}$ and the bias vector $\mathbf{b}^\ast_{l}$ are specified as follows:
	\begin{align}\label{eq:weight_bias}
		\mathbf{W}^\ast_{l}
		= \begin{cases}
			\widetilde{\mathbf{W}}_l, 
			& \text{if } 1\leq l \leq L_1, \\
			\bar{\mathbf{W}}_{l} \widetilde{\mathbf{W}}_{l}, 
			& \text{if } l = L_1 +1, \\
			\bar{\mathbf{W}}_{l},
			& \text{if } L_1 + 2 \leq l \leq L_1 +L_2, \\
			\mathbf{W}_{3,1}  \bar{\mathbf{W}}_{l},
			& \text{if } l = L_1 + L_2 +1, \\
			\mathbf{W}_{3,l-L_1-L_2} ,
			& \text{if }  L_1 + L_2 +2 \leq l \leq L_1 + L_2 + L_3+1, \\
		\end{cases}\\
		\mathbf{b}^\ast_{l}
		= \begin{cases}
			\widetilde{\mathbf{b}}_l, 
			& \text{if } 1\leq l \leq L_1, \\
			\bar{\mathbf{W}}_{l} \widetilde{\mathbf{b}}_{l} +\bar{\mathbf{b}}_{l}, 
			& \text{if } l = L_1 +1, \\
			\bar{\mathbf{b}}_{l},
			& \text{if } L_1 + 2 \leq l \leq L_1 +L_2, \\
			\mathbf{W}_{3,1}  \bar{\mathbf{b}}_{l}+\mathbf{b}_{3,1},
			& \text{if } l = L_1 + L_2 +1, \\
			\mathbf{b}_{3,l-L_1-L_2} ,
			& \text{if }  L_1 + L_2 +2 \leq l \leq L_1 + L_2 + L_3+1. \nonumber\\
		\end{cases}
	\end{align}
	By Assumption \ref{assump:Boundness}(iii),  
	$\max\{\|\widetilde{\mathbf{W}}_l\|_2, \|\widetilde{\mathbf{b}}_l\|_2\}
	= \max\left\{ \|\mathbf{W}_{1,l}\|_{2},\; \sqrt{D_1}\|\mathbf{b}_{1,l}\|_{2},\; 1 \right\} = \max\{ \sqrt{D_1}U_{1,l},\; 1 \}\coloneqq \bar{U}_{1,l} $, $\max\{\|\bar{\mathbf{W}}_{L_1+l}\|_2, \|\bar{\mathbf{b}}_{L_1+l}\|_2\} = \max\{\|\mathbf{W}_{2,l}\|_{2}, \sqrt{D_2}\|\mathbf{b}_{2,l}\|_{2},1\} = \max\{\sqrt{D_2} U_{2,l}, 1\} \coloneqq \bar{U}_{2,l} $, and $\max\{\|\mathbf{W}_{3,l}\|_{2}, \|\mathbf{b}_{3,l}\|_{2}\} = U_{3,l}$.
	Then, the $\ell_2$ upper bound of $\mathbf{W}^\ast_{l}$ and $\mathbf{b}^\ast_{l}$ are as follows:
	\begin{align*}
		U^\ast_{l}
		= \begin{cases}
			\bar{U}_{1,l}, 
			& \text{if } 1\leq l \leq L_1, \\
			\bar{U}_{2,1} (\bar{U}_{1,l}+1), 
			& \text{if } l = L_1 +1, \\
			\bar{U}_{2,l-L_1},
			& \text{if } L_1 + 2 \leq l \leq L_1 +L_2, \\
			{{U}}_{3,1} (\bar{U}_{2,l-L_1}+1), 
			& \text{if } l = L_1 + L_2 +1, \\
			U_{3,l-L_1-L_2} ,
			& \text{if }  L_1 + L_2 +2 \leq l \leq L_1 + L_2 + L_3+1,\\
		\end{cases}
	\end{align*}
	and thus the product of $\{U^\ast_l\}$ with $l \in [L^\ast +1]$ leads to a finite $M$. 
	These verify (i)--(iii) for the function $f$, which completes the proof.
\end{proof}

\begin{lemma}[Hoeffding inequality under $\theta_{\infty,T}$-mixing, cf. Lemma 1 in \citet{alquier-2012}]
	\label{lmma:Hoeffding_ineq}
	Suppose that Assumption \ref{assump:Boundness}(i) and Assumption \ref{assump:Weakly dependence} hold.
	Let $h: (\mathbb{R}^d)^T \to \mathbb{R}$ be a function such that, for all $\mathbf{z}_1[1], \ldots, \mathbf{z}_1[T], \mathbf{z}_2[1], \ldots, \mathbf{z}_2[T] \in \mathbb{R}^d$, 
	\begin{align*}
		\big| h(\mathbf{z}_1[1], \ldots, \mathbf{z}_1[T]) - h(\mathbf{z}_2[1], \ldots, \mathbf{z}_2[T]) \big| 
		\leq \sum_{t=1}^{T} \|\mathbf{z}_1[t] - \mathbf{z}_2[t]\|_2. 
	\end{align*}  
	Then for any $\lambda > 0$, we have
	\begin{align*} 
		\mathbb{E} \left[ \exp \left( \lambda \left\{ \mathbb{E}\left[h(\mathbf{x}[1], \ldots, \mathbf{x}[T])\right] - h(\mathbf{x}[1], \ldots, \mathbf{x}[T]) \right\} \right) \right]
		\leq \exp \left( \frac{\lambda^2 T \big( \theta_{\infty, T} +U_x \big)^2}{2} \right),
	\end{align*}
	where $\theta_{\infty, T}$ is defined as in Definition \ref{def:weak_dependence}, and $U_x \coloneqq \sup_t\|\mathbf{x}[t]\|_2 < \infty$.
\end{lemma}


\section{Additional Results for Empirical Analysis} \label{app:real_data}
This section introduces details for data preprocessing in Section \ref{sec:real data} of the manuscript.

We begin by aligning daily series to $20$ trading days per month to account for varying trading calendars. For a month with $x$ extra days beyond 20, the $\lfloor i(x+20)/(x+1) \rfloor$-th observation is removed for each $i = 1, \ldots, x$, thus uniformly down-sampling to 20 days. Subsequently, a three-step preprocessing is applied: seasonal adjustment (where needed), transformation to stationarity, and standardization to zero mean and unit variance.
Seasonal adjustment is indicated by a binary code $\mathtt{S}$: {N = not adjusted, and Y = adjusted.}  
To induce stationarity, each series is transformed according to one of five operations coded as $\mathtt{T}$, following \citet{mccracken-2015}:  
1 = no transformation,  
2 = $\Delta x_t$,  
3 = $\Delta^2 x_t$,  
4 = $\Delta \log(x_t)$,  
5 = $\Delta^2 \log(x_t)$.  
The specific $(\mathtt{S}, \mathtt{T})$ code for every variable is provided in Tables \ref{tab:econ_indicators_part1}–\ref{tab:econ_indicators_part4}.  
Finally, all series are standardized to mean zero and variance one, stabilizing network inputs. The resulting dataset comprises $T = 208$ quarterly, $QT = 624$ monthly, and $SQT = 12{,}480$ daily observations, with $Q=3$ and $S=20$.

\begin{table}[h]
\centering
\small
\caption{Quarterly Economic Indicators}
\label{tab:econ_indicators_part1}
\begin{tabularx}{\linewidth}{p{0.5cm} p{2.1cm} >{\raggedright\arraybackslash}X p{0.4cm} p{0.4cm}}
\toprule
ID & Abbreviation & Description & $\mathtt{S}$ & $\mathtt{T}$  \\
\midrule
Q1  & GDPC1              & Real Gross Domestic Product & Y & 4 \\
Q2  & PCECC96            & Real Personal Consumption Expenditures & Y & 4 \\
Q3  & GPDIC1             & Real Gross Private Domestic Investment & Y & 4 \\
Q4  & GCEC1              & Real Government Consumption Expenditures and Gross Investment & Y & 4 \\
Q5  & IPDBS              & Implicit Price Deflator & Y & 5 \\
Q6  & RCPHBS             & Real Compensation Per Hour & Y & 4 \\
Q7  & OPHPBS             & Real Output Per Hour of All Persons & Y & 4 \\
Q8  & TABSHNO            & Real Total Assets of Households and Nonprofit Organizations & N & 4 \\
Q9  & TLBSNNCB           & Real Nonfinancial Corporate Business Sector Liabilities & N & 4 \\
\bottomrule
\end{tabularx}
\end{table}

\begin{table}[h]
\centering
\small
\caption{Monthly Economic Indicators (Part 1: M1--M28)}
\label{tab:econ_indicators_part2}
\begin{tabularx}{\linewidth}{p{0.5cm} p{3.6cm} >{\raggedright\arraybackslash}X p{0.3cm} p{0.3cm}}
\toprule
ID & Abbreviation & Description & $\mathtt{S}$ & $\mathtt{T}$  \\
\midrule
M1  & CPIAUCSL            & Consumer Price Index for All Urban Consumers: All Items in U.S. City Average & Y & 5 \\
M2  & CPIULFSL            & Consumer Price Index for All Urban Consumers All Items Less Food in U.S. City Average & Y & 5 \\
M3  & CPIAPPSL            & Consumer Price Index for All Urban Consumers Apparel in U.S. City Average & Y & 5 \\
M4  & PCEPI               & Personal Consumption Expenditures & Y & 5 \\
M5  & CPITRNSL            & Consumer Price Index for All Urban Consumers Transportation & Y & 5 \\
M6  & DSERRG3M086SBEA     & Debt Service Ratio (Households) & Y & 5 \\
M7  & BUSLOANS            & Commercial and Industrial Loans, All Commercial Banks & Y & 5 \\
M8  & NONREVSL            & Nonrevolving Consumer Credit Owned and Securitized & Y & 5 \\
M9  & PPI                 & Producer Price Index by Commodity & N & 5 \\
M10 & DTCFCNM             & Total Consumer Loans and Leases Owned and Securitized by Finance Companies Level & N & 5 \\
M11 & CES5000000008       & Average Hourly Earnings of Production and Nonsupervisory Employees Manufacturing & Y & 5 \\
M12 & CES0600000008       & Average Hourly Earnings of Production and Nonsupervisory Employees Goods Producing & Y & 5 \\
M13 & CUSR0000SA0L2       & Consumer Price Index for All Urban Consumers All Items Less Shelter in U.S. City Average & Y & 5 \\
M14 & DTCOLNVHFNM         & Consumer Motor Vehicle Loans Owned by Finance Companies Level & N & 5 \\
M15 & DNDGRG3M086SBEA     & Personal Consumption Expenditures Nondurable Goods (chain-type price index) & Y & 5 \\
M16 & CPIMEDSL            & Consumer Price Index for All Urban Consumers Medical Care in U.S. City Average & Y & 5 \\
M17 & DDURRG3M086SBEA     & Personal Consumption Expenditures Durable Goods (chain-type price index) & Y & 5 \\
M18 & M2SL                & M2 money stock & Y & 5 \\
M19 & INVEST              & Private Fixed Investment (National Accounts) & Y & 5 \\
M20 & CES2000000008       & Average Hourly Earnings of Production and Nonsupervisory Employees Construction & Y & 5 \\
M21 & M1SL                & M1 money stock & Y & 5 \\
M22 & INDPRO          & Industrial Production: Total Index & Y & 4 \\
M23 & IPDMAT          & Industrial Production: Durable Goods Materials & Y & 4 \\
M24 & RPI             & Real Personal Income & Y & 4 \\
M25 & IPBUSEQ         & Industrial Production: Equipment Business Equipment & Y & 4 \\
M26 & W875RX1         & Real Personal Income Excluding Current Transfer Receipts & Y & 4 \\
M27 & DPCERA3M086SBEA & Real Personal Consumption Expenditures (chain-type quantity index) & Y & 4 \\
M28 & CE16OV          & Employment Level & Y & 4 \\
\bottomrule
\end{tabularx}
\end{table}

\begin{table}[h]
\centering
\small
\caption{Monthly Economic Indicators (Part 2: M29--M56)}
\label{tab:econ_indicators_part3}
\begin{tabularx}{\linewidth}{p{0.6cm} p{2.8cm} >{\raggedright\arraybackslash}X p{0.5cm} p{0.5cm}}
\toprule
ID & Abbreviation & Description & $\mathtt{S}$ & $\mathtt{T}$  \\
\midrule
M29 & CLF16OV         & Civilian Labor Force Level & Y & 4 \\
M30 & IPFINAL         & Industrial Production Final Products & Y & 4 \\
M31 & IPFPNSS          & Industrial Production Final Products and Nonindustrial Supplies & Y & 4 \\
M32 & IPFUELS         & Industrial Production Non-Durable Consumer Energy Products Fuels & Y & 4 \\
M33 & IPNCONGD        & Industrial Production Non-Durable Consumer Goods & Y & 4 \\
M34 & USFIRE          & All Employees Financial Activities & Y & 4 \\
M35 & USCONS          & All Employees Construction & Y & 4 \\
M36 & USGOOD          & All Employees Goods Producing & Y & 4 \\
M37 & M2REAL          & Real M4 Money Stock & Y & 4 \\
M38 & DMANEMP         & All Employees Durable Goods & Y & 4 \\
M39 & IPMAT           & Industrial Production: Materials & Y & 4 \\
M40 & IPB51222S       & Industrial Production: Non-Durable Consumer Energy Products: Residential Utilities & Y & 4 \\
M41 & MANEMP          & All Employees, Manufacturing & Y & 4 \\
M42 & USGOVT          & All Employees, Government & Y & 4 \\
M43 & IPMANSICS       & Industrial Production: Manufacturing (SIC) & Y & 4 \\
M44 & USTRADE         & All Employees, Retail Trade & Y & 4 \\
M45 & IPDCONGD        & Industrial Production: Durable Consumer Goods & Y & 4 \\
M46 & IPCONGD         & Industrial Production: Consumer Goods & Y & 4 \\
M47 & CES1021000001     & All Employees, Mining & Y & 4 \\
M48 & USWTRADE        & All Employees, Wholesale Trade & Y & 4 \\
M49 & UNRATE      & Civilian Unemployment Rate & Y & 2 \\
M50 & AAA         & Moody's Seasoned Aaa Corporate Bond Yield & N & 2 \\
M51 & UEMPMEAN    & Average Weeks Unemployed & Y & 2 \\
M52 & AWOTMAN     & Average Weekly Overtime Hours of Production and Nonsupervisory Employees Manufacturing & Y & 2 \\
M53 & T10YFFM     & 10-Year Treasury Constant Maturity Minus Federal Funds Rate & N & 1 \\
M54 & TB3SMFFM    & 3-Month Treasury Bill Minus Federal Funds Rate & N & 1 \\
M55 & TB6SMFFM    & 6-Month Treasury Bill Minus Federal Funds Rate & N & 1 \\
M56 & SP500       & S\&P 500 Index & N & 4 \\
\bottomrule
\end{tabularx}
\end{table}

\begin{table}[h]
\centering
\small
\caption{Daily Economic Indicators}
\label{tab:econ_indicators_part4}
\begin{tabularx}{\linewidth}{p{0.8cm} p{2.5cm} >{\raggedright\arraybackslash}X p{0.5cm} p{0.5cm}}
\toprule
ID & Abbreviation & Description & $\mathtt{S}$ & $\mathtt{T}$ \\
\midrule
D1  & NASDAQCOM     & NASDAQ Composite Index & N & 4 \\
D2  & DEXJPUS       & Japanese Yen to U.S. Dollar Spot Exchange Rate & N & 4 \\
D3  & DEXCAUS       & Canadian Dollars to U.S. Dollar Spot Exchange Rate & N & 4 \\
D4  & DEXSZUS       & Swiss Francs to U.S. Dollar Spot Exchange Rate & N & 4 \\
D5  & DGS1          & Market Yield on U.S. Treasury Securities at 1-Year Constant Maturity, Quoted on an Investment Basis & N & 4 \\
D6  & DGS10         & Market Yield on U.S. Treasury Securities at 10-Year Constant Maturity, Quoted on an Investment Basis & N & 4 \\
D7  & DTB6          & 6-Month Treasury Bill Secondary Market Rate, Discount Basis & N & 4 \\
D8  & DFF           & Federal Funds Effective Rate & N & 4 \\
D9  & RIFSPBLPND    & Bank Prime Loan Rate & N & 4 \\
D10 & DTB3          & 3-Month Treasury Bill Secondary Market Rate, Discount Basis & N & 2 \\
D11 & T1YFF         & 1-Year Treasury Constant Maturity Minus Federal Funds Rate & N & 2 \\
D12 & T5YFF         & 5-Year Treasury Constant Maturity Minus Federal Funds Rate & N & 2 \\
D13 & USRECD        & NBER based Recession Indicators for the United States from the Period following the Peak through the Trough & N & 1 \\
\bottomrule
\end{tabularx}
\end{table}


\stopcontents
\bibliography{MF}

@article{fan-2023,
	author = {Fan, Jianqing and Gu, Yihong},
	journal = {Journal of the American Statistical Association},
	month = {10},
	pages = {2680--2694},
	title = {{Factor augmented sparse throughput deep RELU neural networks for high dimensional regression}},
	volume = {119},
	year = {2023},
	doi = {10.1080/01621459.2023.2271605},
	url = {https://doi.org/10.1080/01621459.2023.2271605},
}

@article{mccracken-2015,
	author = {McCracken, Michael W. and Ng, Serena},
	journal = {Journal of Business \& Economic Statistics},
	month = {9},
	pages = {574--589},
	title = {{FRED-MD: a monthly database for macroeconomic research}},
	volume = {34},
	year = {2015},
	doi = {10.1080/07350015.2015.1086655},
	url = {https://doi.org/10.1080/07350015.2015.1086655},
}

@techreport{mccracken-2020,
	author = {McCracken, Michael and Ng, Serena},
	month = {3},
	title = {{FRED-QD: a quarterly database for macroeconomic research}},
	year = {2020},
	doi = {10.3386/w26872},
	url = {https://www.nber.org/papers/w26872},
}

@article{xu-2022,
	author = {Xu, Xiaoxiang and Liao, Mingqiu},
	journal = {Atmosphere},
	month = {3},
	pages = {423},
	title = {{Prediction of carbon emissions in China’s power industry based on the Mixed-Data Sampling (MIDAS) regression model}},
	volume = {13},
	year = {2022},
	doi = {10.3390/atmos13030423},
	url = {https://www.mdpi.com/2073-4433/13/3/423},
    
}

@article{audrino-2011,
  title={A general multivariate threshold GARCH model with dynamic conditional correlations},
  author={Audrino, Francesco and Trojani, Fabio},
  journal={Journal of Business \& Economic Statistics},
  volume={29},
  pages={138--149},
  year={2011},
  publisher={Taylor \& Francis}
}

@book{murty-2016,
  title={Support Vector Machines and Perceptrons: Learning, Optimization, Classification, and Application to Social Networks},
  author={Murty, M Narasimha and Raghava, Rashmi},
  year={2016},
  publisher={Springer}
}

@misc{ghysels-2004,
	author = {Ghysels, Eric and Santa-Clara, Pedro and Valkanov, Rossen},
	month = {1},
	title = {{The MIDAS Touch: Mixed Data Sampling Regression Models}},
	year = {2004},
	url = {https://ideas.repec.org/p/cdl/anderf/qt9mf223rs.html},
    note = {UCLA: Finance. Retrieved from https://escholarship.org/uc/item/9mf223rs}
}

@article{ghysels-2007,
	author = {Ghysels, Eric and Sinko, Arthur and Valkanov, Rossen I.},
	journal = {Econometric Reviews},
	month = {2},
    pages = {53-90},
	title = {{MIDAS regressions: Further results and new directions}},
    volume = {26},
	year = {2007},
	doi = {10.2139/ssrn.885683},
	url = {https://doi.org/10.2139/ssrn.885683},
}

@article{clements-2008,
	author = {Clements, Michael and Galvão, Ana},
	journal = {Journal of Business \& Economic Statistics},
	pages = {546-554},
	title = {{Macroeconomic Forecasting with Mixed-Frequency data}},
	volume = {26},
	year = {2008},
	url = {https://www.tandfonline.com/doi/abs/10.1198/073500108000000015},
}

@article{andreou-2010,
	author = {Andreou, Elena  and Ghysels, Eric and Kourtellos, Andros },
	journal = {Journal of Econometrics},
	month = {10},
	pages = {246-261},
	title = {{Regression models with mixed sampling frequencies}},
	volume = {158},
	year = {2010},
	url = {https://www.sciencedirect.com/science/article/abs/pii/S0304407610000072},
}

@article{frale-2010,
	author = {Frale, Cecilia and Monteforte, Libero},
	month = {1},
	title = {{Famidas: A Mixed Frequency Factor Model with MIDAS Structure}},
	year = {2010},
	doi = {10.2139/ssrn.1664951},
	url = {https://doi.org/10.2139/ssrn.1664951},
    journal = {SSRN Electronic Journal. SSRN:1829984}
}

@article{andreou-2013,
	author = {Andreou, Elena and Ghysels, Eric and Kourtellos, Andros},
	journal = {Journal of Business \& Economic Statistics},
	month = {2},
	pages = {240--251},
	title = {{Should macroeconomic forecasters use daily financial data and how?}},
	volume = {31},
	year = {2013},
	doi = {10.1080/07350015.2013.767199},
	url = {https://doi.org/10.1080/07350015.2013.767199},
}

@article{engle-2013,
	author = {Engle, Robert F. and Ghysels, Eric and Sohn, Bumjean},
	journal = {The Review of Economics and Statistics},
	month = {1},
	pages = {776--797},
	title = {{Stock market volatility and macroeconomic fundamentals}},
	volume = {95},
	year = {2013},
	doi = {10.1162/rest\{_}a\{_}

@article{amendola-2018,
	author = {Amendola, Alessandra and Candila, Vincenzo and Gallo, Giampiero M.},
	journal = {Economic Modelling},
	month = {8},
	pages = {135--152},
	title = {{On the asymmetric impact of macro–variables on volatility}},
	volume = {76},
	year = {2018},
	doi = {10.1016/j.econmod.2018.07.025},
	url = {https://doi.org/10.1016/j.econmod.2018.07.025},
}

@article{foroni-2015,
	author = {Foroni, Claudia and Marcellino, Massimiliano and Schumacher, Christian},
	journal = {Journal of the Royal Statistical Society Series A (Statistics in Society)},
	month = {12},
	pages = {57--82},
	title = {{Unrestricted Mixed Data Sampling (MIDAS): MIDAS Regressions with Unrestricted Lag Polynomials}},
	volume = {178},
	year = {2015},
	doi = {10.1111/rssa.12043},
	url = {https://doi.org/10.1111/rssa.12043},
}

@article{ghysels-2014,
	author = {Ghysels, E.},
	journal = {Journal of Financial Econometrics},
	month = {9},
	pages = {620--644},
	title = {{Conditional Skewness with Quantile Regression Models: SoFiE Presidential Address and a Tribute to Hal White}},
	volume = {12},
	year = {2014},
	doi = {10.1093/jjfinec/nbu021},
	url = {https://doi.org/10.1093/jjfinec/nbu021},
}

@article{ghysels-2016b,
	author = {Ghysels, Eric and Plazzi, Alberto and Valkanov, Rossen},
	journal = {The Journal of Finance},
	month = {5},
	pages = {2145--2192},
	title = {{Why invest in emerging markets? The role of conditional return asymmetry}},
	volume = {71},
	year = {2016},
	doi = {10.1111/jofi.12420},
	url = {https://doi.org/10.1111/jofi.12420},
}

@article{babii-2022,
	author = {Babii, Andrii and Ghysels, Eric and Striaukas, Jonas},
	journal = {Journal of Business \& Economic Statistics},
	pages = {1094-1106},
	title = {{Machine learning time series regressions with an application to nowcasting}},
	volume = {40},
	year = {2022},
	url = {https://www.tandfonline.com/doi/abs/10.1080/07350015.2021.1899933},
}

@article{xu-2018,
	author = {Xu, Qifa and Zhuo, Xingxuan and Jiang, Cuixia and Liu, Yezheng},
	journal = {Expert Systems with Applications},
	month = {10},
	pages = {127--139},
	title = {{An artificial neural network for mixed frequency data}},
	volume = {118},
	year = {2018},
	doi = {10.1016/j.eswa.2018.10.013},
	url = {https://doi.org/10.1016/j.eswa.2018.10.013},
}

@article{xu-2021,
	author = {Xu, Qifa and Liu, Shuting and Jiang, Cuixia and Zhuo, Xingxuan},
	journal = {Neurocomputing},
	month = {6},
	pages = {84--105},
	title = {{QRNN-MIDAS: A novel quantile regression neural network for mixed sampling frequency data}},
	volume = {457},
	year = {2021},
	doi = {10.1016/j.neucom.2021.06.006},
	url = {https://doi.org/10.1016/j.neucom.2021.06.006},
}

@article{wang-2025,
	author = {Wang, Chunzi and Xie, Fusheng and Yan, Junpeng and Xia, Yiqing},
	journal = {Energy Reports},
	month = {12},
	pages = {16--26},
	title = {{A U-MIDAS modeling framework for forecasting carbon dioxide emissions based on LSTM network and LASSO regression}},
	volume = {13},
	year = {2025},
	doi = {10.1016/j.egyr.2024.11.069},
	url = {https://doi.org/10.1016/j.egyr.2024.11.069},
}

@article{galvao-2013,
  title={Changes in predictive ability with mixed frequency data},
  author={Galv{\~a}o, Ana Beatriz},
  journal={International Journal of Forecasting},
  volume={29},
  pages={395--410},
  year={2013},
  publisher={Elsevier}
}

@article{breitung-2015,
  title={Forecasting inflation rates using daily data: A nonparametric MIDAS approach},
  author={Breitung, J{\"O}rg and Roling, Christoph},
  journal={Journal of Forecasting},
  volume={34},
  pages={588--603},
  year={2015},
  publisher={Wiley Online Library}
}

@article{wei-2025,
  title={Fully NonParametric MIDAS: A new approach for nonparametric mixed frequency time series regression},
  author={Wei, James L and Nason, Guy P},
  journal={Electronic Journal of Statistics},
  volume={19},
  pages={3292--3316},
  year={2025},
  publisher={The Institute of Mathematical Statistics and the Bernoulli Society}
}

@article{bai-2013,
	author = {Bai, Jennie and Ghysels, Eric and Wright, Jonathan H.},
	journal = {Econometric Reviews},
	month = {3},
	pages = {779--813},
	title = {{State space models and MIDAS regressions}},
	volume = {32},
	year = {2013},
	doi = {10.1080/07474938.2012.690675},
	url = {https://doi.org/10.1080/07474938.2012.690675},
}

@article{mariano-2002,
	author = {Mariano, Roberto S. and Murasawa, Yasutomo},
	journal = {Journal of Applied Econometrics},
	month = {10},
	pages = {427--443},
	title = {{A new coincident index of business cycles based on monthly and quarterly series}},
	volume = {18},
	year = {2002},
	doi = {10.1002/jae.695},
	url = {https://doi.org/10.1002/jae.695},
}

@article{schumacher-2007,
	author = {Schumacher, Christian and Breitung, Jörg},
	month = {1},
	title = {{Real-Time Forecasting of GDP Based on a Large Factor Model with Monthly and Quarterly Data}},
	year = {2007},
	doi = {10.2139/ssrn.965685},
    journal = {SSRN Electronic Journal. SSRN:2785260}
}

@article{giannone-2008,
	author = {Giannone, Domenico and Reichlin, Lucrezia and Small, David},
	journal = {Journal of Monetary Economics},
	month = {5},
	pages = {665--676},
	title = {{Nowcasting: The real-time informational content of macroeconomic data}},
	volume = {55},
	year = {2008},
	doi = {10.1016/j.jmoneco.2008.05.010},
	url = {https://doi.org/10.1016/j.jmoneco.2008.05.010},
}

@article{banbura-2011,
    title={Nowcasting with daily data},
    author = {Bańbura, Marta and Domenico, Giannone and Michele, Modugno and Lucrezia, Reichlin},
    journal={European Central Bank},
    year={2011}
}

@article{marcellino-2016,
	author = {Marcellino, Massimiliano and Porqueddu, Mario and Venditti, Fabrizio},
	journal = {Journal of Business \& Economic Statistics},
	month = {3},
	pages = {118--127},
	title = {{Short-Term GDP forecasting with a Mixed-Frequency dynamic factor model with stochastic volatility}},
	volume = {34},
	year = {2016},
	doi = {10.1080/07350015.2015.1006773},
	url = {https://doi.org/10.1080/07350015.2015.1006773},
}

@article{fan-2024,
	author = {Fan, Jianqing and Gu, Yihong and Zhou, Wenxin},
	journal = {The Annals of Statistics},
	month = {8},
	title = {{How do noise tails impact on deep ReLU networks?}},
	volume = {52},
    pages={1845--1871},
	year = {2024},
	doi = {10.1214/24-aos2428},
	url = {https://doi.org/10.1214/24-aos2428},
}

@article{schmidt-hieber-2020,
	author = {Schmidt-Hieber, Johannes},
	journal = {The Annals of Statistics},
	month = {8},
	title = {{Nonparametric regression using deep neural networks with ReLU activation function}},
	volume = {48},
    pages={1875--1897},
	year = {2020},
	doi = {10.1214/19-aos1875},
}

@article{kohler-2021,
	author = {Kohler, Michael and Langer, Sophie},
	journal = {The Annals of Statistics},
	month = {8},
	title = {{On the rate of convergence of fully connected deep neural network regression estimates}},
    volume = {49},
    pages={2231--2249},
	year = {2021},
	doi = {10.1214/20-aos2034},
	url = {https://doi.org/10.1214/20-aos2034},
}

@article{caner-2025,
	author = {Caner, Mehmet and Daniele, Maurizio},
	journal = {Journal of Econometrics},
	month = {8},
	pages = {106083},
	title = {{Deep learning based residuals in non-linear factor models: Precision matrix estimation of returns with low signal-to-noise ratio}},
	volume = {251},
	year = {2025},
	doi = {10.1016/j.jeconom.2025.106083},
	url = {https://doi.org/10.1016/j.jeconom.2025.106083},
}

@article{farrell-2021,
	author = {Farrell, Max H. and Liang, Tengyuan and Misra, Sanjog},
	journal = {Econometrica},
	month = {1},
	pages = {181--213},
	title = {{Deep neural networks for estimation and inference}},
	volume = {89},
	year = {2021},
	doi = {10.3982/ecta16901},
	url = {https://doi.org/10.3982/ecta16901},
}

@article{rumelhart-1986,
  title={Learning representations by back-propagating errors},
  author={Rumelhart, David E and Hinton, Geoffrey E and Williams, Ronald J},
  journal={Nature},
  volume={323},
  pages={533--536},
  year={1986},
  publisher={Nature Publishing Group UK London}
}

@article{hochreiter-1997,
  title={Long short-term memory},
  author={Hochreiter, Sepp and Schmidhuber, J{\"u}rgen},
  journal={Neural Computation},
  volume={9},
  pages={1735--1780},
  year={1997},
  publisher={MIT press}
}

@article{bengio-2006,
  title={Greedy layer-wise training of deep networks},
  author={Bengio, Yoshua and Lamblin, Pascal and Popovici, Dan and Larochelle, Hugo},
  journal={Advances in Neural Information Processing Systems},
  volume={19},
  year={2006}
}

@misc{kingma-2014,
	author = {Kingma, Diederik and Ba, Jimmy},
	month = {12},
	title = {{Adam: A method for stochastic optimization}},
	year = {2014},
	url = {https://arxiv.org/abs/1412.6980},
    note={\emph{arXiv preprint arXiv:1412.6980}},
}

@article{ghysels-2016a,
	author = {Ghysels, Eric},
	journal = {Journal of Econometrics},
	month = {5},
	pages = {294--314},
	title = {{Macroeconomics and the reality of mixed frequency data}},
	volume = {193},
	year = {2016},
	doi = {10.1016/j.jeconom.2016.04.008},
	url = {https://doi.org/10.1016/j.jeconom.2016.04.008},
}

@article{chakraborty-2023,
	author = {Chakraborty, Nilanjana and Khare, Kshitij and Michailidis, George},
	journal = {Statistica Sinica},
	month = {3},
	title = {{A Bayesian Framework for sparse Estimation in High-Dimensional Mixed Frequency Vector Autoregressive models}},
	year = {2023},
    volume={33},
    pages={1629--1652},
	doi = {10.5705/ss.202021.0206},
	url = {https://doi.org/10.5705/ss.202021.0206},
}

@article{schorfheide-2015,
	author = {Schorfheide, Frank and Song, Dongho},
	journal = {Journal of Business \& Economic Statistics},
	month = {8},
	pages = {366--380},
	title = {{Real-Time forecasting with a Mixed-Frequency VAR}},
	volume = {33},
	year = {2015},
	doi = {10.1080/07350015.2014.954707},
	url = {https://doi.org/10.1080/07350015.2014.954707},
}

@article{lin-2024,
	author = {Lin, Jiahe and Michailidis, George},
	journal = {International Journal of Forecasting},
	month = {9},
	pages = {942--957},
	title = {{A multi-task encoder-dual-decoder framework for mixed frequency data prediction}},
	volume = {40},
	year = {2024},
	doi = {10.1016/j.ijforecast.2023.08.003},
	url = {https://doi.org/10.1016/j.ijforecast.2023.08.003},
}

@book{dedecker-2007,
	author = {Dedecker, Jérôme and Doukhan, Paul and Lang, Gabriel and Rafael, León R. José and Louhichi, Sana and Prieur, Clémentine},
	month = {1},
	title = {{Weak Dependence: With Examples and Applications}},
	year = {2007},
	doi = {10.1007/978-0-387-69952-3},
	url = {https://doi.org/10.1007/978-0-387-69952-3},
    publisher={Springer}
}

@article{ou-2026,
	author = {Ou, Weigutian and Bölcskei, Helmut},
	journal = {Foundations of Computational Mathematics},
	month = {9},
	title = {{Covering Numbers for Deep ReLU Networks with Applications to Function Approximation and Nonparametric Regression}},
	year = {2026},
	doi = {10.1007/s10208-026-09762-w},
	url = {https://doi.org/10.1007/s10208-026-09762-w},
}

@book{wainwright-2019,
  title={High-Dimensional Statistics: A Non-Asymptotic Viewpoint},
  author={Wainwright, Martin J},
  year={2019},
  publisher={Cambridge university press}
}

@misc{chen-2019,
	author = {Chen, Minshuo and Li, Xingguo and Zhao, Tuo},
	title = {{On generalization bounds of a family of recurrent neural networks}},
	year = {2019},
	url = {https://arxiv.org/abs/1910.12947},
    note={\emph{arXiv preprint arXiv:1910.12947}},
}

@misc{kengne-2023,
	author = {Kengne, William and Wade, Modou},
	month = {5},
	title = {{Penalized deep neural networks estimator with general loss functions under weak dependence}},
	year = {2023},
	url = {https://arxiv.org/abs/2305.06230},
    note={\emph{arXiv preprint arXiv:2305.06230}},
}

@article{alquier-2012,
  title={Prediction of time series by statistical learning: general losses and fast rates},
  author={Alquier, Pierre and Li, Xiaoyin and Wintenberger, Olivier},
  journal={arXiv preprint arXiv:1211.1847},
  year={2012}
}

@article{feng-2023,
  title={Over-parameterized deep nonparametric regression for dependent data with its applications to reinforcement learning},
  author={Feng, Xingdong and Jiao, Yuling and Kang, Lican and Zhang, Baqun and Zhou, Fan},
  journal={Journal of Machine Learning Research},
  volume={24},
  pages={1--40},
  year={2023}
}

@article{jiao-2023,
  author  = {Jiao, Yuling and Wang, Yang and Yang, Yunfei},
  title   = {Approximation bounds for norm constrained neural networks with applications to regression and {GANs}},
  journal = {Applied and Computational Harmonic Analysis},
  volume  = {65},
  pages   = {249--278},
  year    = {2023},
  doi     = {10.1016/j.acha.2023.03.004}
}

@article{xiu-2024,
  title={Deep autoencoders for nonlinear factor models: Theory and applications},
  author={Xiu, Dacheng and Shen, Zhouyu},
  year={2024},
  url = {https://ssrn.com/abstract=5074409 or http://dx.doi.org/10.2139/ssrn.5074409},
  journal={SSRN working paper No.5074409},
}

@article{cheng-2022,
  title={Financial time series forecasting with multi-modality graph neural network},
  author={Cheng, Dawei and Yang, Fangzhou and Xiang, Sheng and Liu, Jin},
  journal={Pattern Recognition},
  volume={121},
  pages={108218},
  year={2022},
  publisher={Elsevier}
}

@article{jia-2024,
  title={Gpt4mts: Prompt-based large language model for multimodal time-series forecasting},
  author={Jia, Furong and Wang, Kevin and Zheng, Yixiang and Cao, Defu and Liu, Yan},
  journal={Proceedings of the AAAI Conference on Artificial Intelligence},
  volume={38},
  pages={23343--23351},
  year={2024},
}

@inproceedings{jiang-2025,
  title={Multi-modal time series analysis: A tutorial and survey},
  author={Jiang, Yushan and Ning, Kanghui and Pan, Zijie and Shen, Xuyang and Ni, Jingchao and Yu, Wenchao and Schneider, Anderson and Chen, Haifeng and Nevmyvaka, Yuriy and Song, Dongjin},
  booktitle={Proceedings of the 31st ACM SIGKDD Conference on Knowledge Discovery and Data Mining V. 2},
  year={2025}
}

@article{liu-2025,
  title={How can time series analysis benefit from multiple modalities? a survey and outlook},
  author={Liu, Haoxin and Kamarthi, Harshavardhan and Zhao, Zhiyuan and Xu, Shangqing and Wang, Shiyu and Wen, Qingsong and Hartvigsen, Tom and Wang, Fei and Prakash, B Aditya},
  journal={arXiv preprint arXiv:2503.11835},
  year={2025}
}

@inproceedings{chollet-2017,
  title={Xception: Deep learning with depthwise separable convolutions},
  author={Chollet, Fran{\c{c}}ois},
  booktitle={Proceedings of the IEEE conference on computer vision and pattern recognition},
  pages={1251--1258},
  year={2017},
}

@inproceedings{wang-2020,
  title={Compact autoregressive network},
  author={Wang, Di and Huang, Feiqing and Zhao, Jingyu and Li, Guodong and Tian, Guangjian},
  booktitle={Proceedings of the AAAI Conference on Artificial Intelligence, 34:6145--6152},
  year={2020},
}
\end{document}